\documentclass[aps,prb,10pt,reprint,groupedaddress,showpacs,amssymb,amsfonts]{revtex4-1}
\usepackage{amsmath,amssymb,amscd,amsthm,latexsym,epsfig,bm,subfigure}
\usepackage{mathrsfs}
\usepackage[usenames,dvipsnames]{color}
\usepackage[colorlinks,bookmarks=false,citecolor=blue,linkcolor=red,urlcolor=blue]{hyperref}
\usepackage{verbatim}

\newcommand{\zz}{\mathbb{Z}_2}
\newcommand{\z}{\mathbb{Z}}
\newcommand{\fS}{\mathfrak{S}}
\newcommand{\cL}{{\cal L}}
\newcommand{\boell}{\boldsymbol{\ell}}

\theoremstyle{remark}
\newtheorem*{rem*}{\protect\remarkname}

\theoremstyle{plain}
\newtheorem{lem}{\protect\lemmaname}
\theoremstyle{plain}
\newtheorem{thm}{\protect\theoremname}
\theoremstyle{plain}
\newtheorem{prop}{\protect\propositionname}
\theoremstyle{plain}
\newtheorem{thmtext}{\protect\theoremtextname}

  \providecommand{\claimname}{Claim}
  \providecommand{\lemmaname}{Lemma}
  \providecommand{\propositionname}{Proposition}
\providecommand{\theoremname}{Theorem}
\providecommand{\theoremtextname}{Theorem}
  \providecommand{\remarkname}{Remark}

\begin{document}

\title{Space group symmetry fractionalization \\ in a family of exactly solvable models with $\zz$ topological order}

\author{Hao Song}

\author{Michael Hermele}

\affiliation{Department of Physics, University of Colorado, Boulder, Colorado
80309, USA}
\begin{abstract}
We study square lattice space group symmetry fractionalization in a family of exactly solvable models with $\zz$ topological order in two dimensions.  In particular, we have obtained a complete understanding of which distinct types of symmetry fractionalization (symmetry classes) can be realized within this class of models, which are generalizations of Kitaev's $\zz$ toric code to arbitrary lattices.  This question is motivated by earlier work of A. M. Essin and one of us (M. H.), where the idea of symmetry classification was laid out, and which, for square lattice symmetry, produces 2080 symmetry classes consistent with the fusion rules of $\zz$ topological order.  This approach does not produce a physical model for each symmetry class, and indeed there are reasons to believe that some symmetry classes may not be realizable in strictly two-dimensional systems, thus raising the question of which classes are in fact possible.  While our understanding is limited to a restricted class of models, it is complete in the sense that for each of the 2080 possible symmetry classes, we either prove rigorously that the class cannot be realized in our family of models, or we give an explicit model realizing the class.  We thus find that exactly 487 symmetry classes are realized in the family of models considered.  With a more restrictive type of symmetry action, where space group operations act trivially in the internal Hilbert space of each spin degree of freedom, we find that exactly 82 symmetry classes are realized.  In addition, we present a single model that realizes all $2^6 = 64$ types of symmetry fractionalization allowed for a single anyon species ($\zz$ charge excitation), as the parameters in the Hamiltonian are varied.  The paper concludes with a summary and a discussion of two results pertaining to more general bosonic models.
\end{abstract}
\maketitle

\section{Introduction}
\label{sec:Introduction}

\subsection{Background}
Topological phases of matter are those with an energy gap to all excitations, and host remarkable phenomena such as protected gapless edge states, and anyon quasiparticle excitations with non-trivial braiding statistics.  Following the discovery of time-reversal invariant topological band insulators,\cite{hasan10,hasan11,qi11} significant advances have been made in understanding the role of symmetry in topological phases.

Two broad families of such phases are \emph{symmetry protected topological} (SPT) phases,\cite{chen11a, fidkowski11, schuch11,turner11,chen13} and \emph{symmetry enriched topological} (SET) phases.  SPT phases, which include topological band insulators, reduce to the trivial gapped phase if the symmetries present are weakly broken.  These phases lack anyon excitations in the bulk, and many characteristic physical properties are confined to edges and surfaces.  SET phases, on the other hand, are topologically ordered, with anyon excitations in the bulk.  Topological order is robust to arbitrary perturbations provided the gap stays open, and SET phases remain non-trivial even when all symmetries are broken.  In the presence of symmetry, there can be an interesting interplay between symmetry and topological order.  This interplay is important, because properties tied to symmetry are often easier to observe experimentally.  For example, in fractional quantum Hall liquids,\cite{tsui82, laughlin83} quantization of Hall conductance\cite{tsui82} and fractional charge\cite{depicciotto97,saminadayar97,Martin2004} have been directly observed, and arise from the interplay between  ${\rm U}(1)$ charge symmetry and topological order.  The example of fractional quantum Hall liquids makes it clear that the study of SET phases has a long history, which cannot be adequately reviewed here; instead, we simply mention two areas of prior work that have close ties with the focus and results of the present paper.  First, topologically ordered quantum spin liquids are another much-studied class of SET phases.\cite{kalmeyer1987,wen89b,wen91,read91,sachdev92,balents99,senthil00,moessner01b}  Second, a systematic understanding of the role of symmetry in SET phases has recently been developing, including work on classification of such phases; some representative studies are found in Refs.~\onlinecite{wen02,wen03,kitaev06,fwang06,kou09,huh11,cho12,gchen12,levin12,essin13,Mesaros2013,
Hung2013,Hung2013b,Lu2013,Xu2013,cwang13,essin14,chen14,ygu14,ymlu14,Huang2014,hermele14,reuther14,neupert2014wire}.

Most of the recent work on SPT and SET phases has focused on on-site symmetries such as time reversal, ${\rm U}(1)$ charge symmetry, and ${\rm SO}(3)$ spin symmetry.  For SPT phases, this restriction makes sense physically, because a generic edge or surface will not have any spatial symmetries, but may have on-site symmetry.  Of course, there can be clean edges and surfaces, and some works have examined the role of space group symmetry in SPT phases.\cite{turner10,turner10b,fu11,hughes11,teo13,fang12,zwang12b,chiu13,fzhang13,slager2013space}  For SET phases, there is not a good physical justification to ignore spatial symmetries; the presence of  anyon quasiparticles means that symmetries of the bulk can directly impact characteristic physical properties. Indeed, a number of studies have focused on the role of space group symmetry in SET phases.\cite{wen02,fwang06,kou09,huh11,cho12,gchen12,essin13,essin14,ymlu14,reuther14}  However, many recent works on SET phases have limited attention to on-site symmetry.

Recently, A. M. Essin and one of us (M. H.), building on earlier work,\cite{wen02,wen03} introduced a \emph{symmetry classification} approach to bosonic SET phases in two dimensions, designed to handle both on-site and spatial symmetries.\cite{essin13}  The basic idea is to consider a fixed Abelian topological order and fixed symmetry group $G$, and establish symmetry classes corresponding to distinct possible actions of symmetry on the anyon quasiparticles, so that two phases in different symmetry classes must be distinct (as long as the symmetry is preserved).  Under the simplifying assumption that symmetry does not permute the various anyon species, the approach of Ref.~\onlinecite{essin13} amounts to classifying distinct types of \emph{symmetry fractionalization}, where this term reflects the fact that the action of symmetry fractionalizes at the operator level when acting on anyons.  

Distinct types of symmetry fractionalization are referred to as \emph{fractionalization classes}, and characterize the projective representations giving the action of the symmetry group on individual anyons.  Assigning a fractionalization class to each type of anyon specifies the \emph{symmetry class} of a SET phase.
Ref.~\onlinecite{essin13} focused primarily on the simple case of $\zz$ topological order,
giving a symmetry classification for square lattice space group plus time reversal symmetry, that can easily be generalized to any desired symmetry group.  For $\zz$ topological order with symmetry group $G$, a symmetry class is specified by fractionalization classes $[\omega_e]$ and $[\omega_m]$, for $e$ particle ($\zz$ charge) and $m$ particle ($\zz$ flux) excitations, respectively.  Mathematically, distinct fractionalization classes are elements of the cohomology group $H^2(G, \zz)$.  In more detail, a symmetry class is an un-ordered pair $\langle [\omega_e], [\omega_m] \rangle \simeq \langle [\omega_m], [\omega_e] \rangle$, where the lack of ordering comes from the fact that the distinction between $e$ and $m$ particle excitations is arbitrary, and we are always free to make the  relabeling $e \leftrightarrow m$.

A crucial issue left open by the general considerations of Ref.~\onlinecite{essin13} is the realization of symmetry classes in microscopic models (or physically reasonable low-energy effective theories).  In this paper, focusing on $\zz$ topological order and square lattice space group symmetry, we address this issue via a systematic study of a family of exactly solvable lattice models, in which many symmetry classes are realized.  This is interesting for several reasons.  First, to our knowledge, a general framework to describe SET phases with space group symmetry has not yet emerged, and concrete models for such phases are likely to be useful in developing such a framework.  This contrasts with SET phases with on-site symmetry, where powerful tools are available, including approaches based on Chern-Simons theory,\cite{levin12,Lu2013,Hung2013b}  on classification of topological terms using group cohomology,\cite{Mesaros2013,Hung2013} and on tensor category theory.\cite{lukaszpc, barkeshli14}  Second, it is likely that not all symmetry classes are realizable in strictly two-dimensional systems.  For on-site symmetry, some symmetry classes can only arise on the surface of a $d=3$ SPT phase.\cite{vishwanath12,metlitski13,cwang13,chen14}  Understanding which space group symmetry classes can be realized in simple models is a step toward addressing the more challenging general question of which classes can (and cannot) occur strictly in two dimensions.    Finally, the explicit models we construct can be used as a testing ground for new ideas to probe and detect the characteristic properties of SET phases, in both experiments and numerical studies of more realistic microscopic models.

The models we consider are generalizations of Kitaev's $\zz$ toric code\cite{kitaev03} to arbitrary two-dimensional lattices with square lattice space group symmetry (a precise definition appears in Sec.~\ref{sec:genlatt}).  By appropriately choosing the lattice geometry, varying the \emph{signs} of terms in the Hamiltonian, and allowing symmetry to act non-trivially on spin operators, many but not all symmetry classes can be realized.  Varying the signs of terms in the Hamiltonian modulates the pattern of background $\zz$ fluxes and charges in the ground state, and this in turn affects the symmetry fractionalization of $e$ and $m$ particles, respectively.  In addition, non-trivial action of symmetry on the spin degrees of freedom also affects symmetry fractionalization.  We have obtained a complete understanding for the specific family of models considered, in the sense that for every symmetry class consistent with the considerations of Ref.~\onlinecite{essin13}, we either give an explicit model realizing this symmetry class, or we prove rigorously that it cannot occur within our family of models.

The idea of choosing the lattice geometry and varying the signs of terms in the Hamiltonian can be viewed as implementations of a ``string flux'' mechanism for fractionalization in topologically ordered phases, recently introduced by one of us (M.H.).\cite{hermele14}  In Ref.~\onlinecite{hermele14}, exactly solvable $\z_n$ toric code models were constructed with on-site, unitary symmetry $G$, for $G$ an arbitrary finite group.  These models can realize arbitrary symmetry fractionalization for anyons corresponding to $\z_n$ gauge charges, and do so by encoding a pattern of fluxes into the ground state, so that the wavefunction acquires phase factors when the strings attached to anyons slide over these fluxes.  The present work differs significantly from Ref.~\onlinecite{hermele14} in the focus on space group symmetry, and in the fact that we allow for and find non-trivial symmetry fractionalization for both $\zz$ charge and flux anyons.  A perhaps even more important distinction is the emphasis here on obtaining a complete understanding for a given family of models, as compared to the emphasis in Ref.~\onlinecite{hermele14} of devising a simple means to encode physically the underlying mathematical structure of fractionalization classes.

\subsection{Outline of the paper}
Due to the length of the paper, we first point out that readers can find the main results in Section~\ref{sec:Toric}. Readers familiar with the necessary background should be able to understand the statements of results in Sec.~\ref{sec:Toric}, after quickly consulting Sec.~\ref{sec:fracreview}, and especially Eqs.~(\ref{eqn:projpx}-\ref{eqn:projtypx}), to become familiar with notation and conventions used to present symmetry classes.

Now, to overview the main results, the aim of this paper is to explore the possible symmetry classes associated to the space group $G$ of the square lattice within a particular family of local bosonic models with $\mathbb{Z}_2$ topological order. We call this family of models $TC(G)$, and it consists of variations of Kitaev's $\zz$ toric code\cite{kitaev03} obtained by changing the lattice geometry, varying the signs of terms in the Hamiltonian, and allowing symmetry to act non-trivially on spin operators (referred to as spin-orbit coupling).  Section~\ref{sec:Toric} studies symmetry fractionalization in these models, beginning with a specific example and moving towards increasing generality.
First, in Sec.~\ref{sec:ep-model} we describe a single model realizing all $e$ particle fractionalization classes while the $m$ particle always has trivial symmetry fractionalization.  The constraints that arise when both $e$ and $m$ particles have non-trivial symmetry fractionalization are considered in the following subsections.  In Sec.~\ref{sub:tc_woso}, we examine a  subclass of models,  $TC_0(G) \subset TC(G)$, where no spin-orbit coupling is allowed.  The main result of Sec.~\ref{sub:tc_woso} is Theorem~\ref{thm:nosoc_maintext}, which describes all symmetry classes that are realized by models in $TC_0(G)$.  Following the statement of the theorem, example models realizing all possible symmetry classes for $TC_0(G)$ are presented.  Finally, in Sec.~\ref{sub:tc_so}, we treat the general case of $TC(G)$, and state Theorem~\ref{thm:soc_maintext}, which describes all symmetry classes that are realized by models in $TC(G)$; the discussion parallels that of Sec.~\ref{sub:tc_woso}.  The detailed proofs of the theorems are left to the appendices, together with the presentation of models realizing all possible symmetry classes for $TC(G)$.  Our results establish that certain symmetry classes are possible in two dimensional models. For symmetry classes that are not realized by models $TC(G)$, a more general understanding of which such symmetry classes are possible strictly in two dimensions is still lacking.

%We now give an outline of the remainder of the paper, which also serves as a summary of our main results. 
Now we describe how the rest of the paper is organized.
Section~\ref{sec:z2review} reviews $\zz$ topological order, and Sec.~\ref{sec:toricreview} gives a review of the simplest $\zz$ Kitaev toric code model, on the two-dimensional square lattice.  The crucial objects are the $e$ ($\zz$ charge) and $m$ ($\zz$ flux) excitations of $\zz$ topological order, referred to as $e$ and $m$ particles.
Readers already familiar with these topics may wish to skim Sections~\ref{sec:z2review} and~\ref{sec:toricreview}, and proceed to Section~\ref{sec:genlatt}, where we introduce the family of toric code models on general lattices with square lattice symmetry; some technical details are presented in Appendices~\ref{app:complete-set} and~\ref{app:psi0}.  We actually introduce two families of models;  in one of these, square lattice symmetry acts only by moving spin degrees of freedom from one spatial location to another, but all symmetries act trivially within the internal Hilbert space of each spin.  This situation is referred to in our paper as that of \emph{no spin-orbit coupling}, and the resulting family of models is called $TC_0(G)$, where $G$ is the square lattice space group.  We also consider a larger family of models, $TC(G)$, that contains $TC_0(G)$.  In $TC(G)$, symmetries are allowed to act non-trivially on the spin degrees of freedom, and we refer to this as the presence of \emph{spin-orbit coupling}.  It should be noted that our usage of the term spin-orbit coupling is a generalization of the usual usage; in particular, our spins do not necessarily transform as electron spins do under a given rigid motion of space.  Such a generalization is physically reasonable, because there are many ways in which two-component pseudospin degrees of freedom arise in real systems, and such degrees of freedom do not always transform like electron spins under symmetry.

With the models of interest having been introduced, Sec.~\ref{sec:fracreview} follows Ref.~\onlinecite{essin13} and reviews the notions of fractionalization and symmetry classes.  It should be noted that, as in Ref.~\onlinecite{essin13}, we always make the simplifying assumption that symmetry does not permute the anyon species.  Indeed, the family of models $TC(G)$ is defined so that permutations of anyons under symmetry never occur.  Section~\ref{sec:fc-in-solvable-models} proceeds to give a detailed description of how symmetry fractionalization is realized in the solvable toric code models for both $e$ and $m$ particle excitations.  The important notions of $e$ and $m$ localizations of the symmetry are introduced and discussed, which provide the means to calculate the fractionalization and symmetry classes for given models in $TC(G)$.  In our solvable models, the $e$ and $m$ particle excitations have different character, and it is convenient to distinguish them by introducing the notion of toric code (TC) symmetry class, which is an \emph{ordered} pair $([\omega_e], [\omega_m])$.  While we do not expect TC symmetry classes to have any universal meaning, they are useful in understanding the possibilities for toric code models.  Appendix~\ref{app:eloc} proves some general results about $e$ and $m$ localizations, and gives a general expression for these localizations that is useful in deriving constraints on which symmetry classes are possible.

The main results of the paper are presented in Section~\ref{sec:Toric}, in order of increasing generality.  First, in Section~\ref{sec:ep-model} we describe a single model that realizes all $2^6 = 64$ fractionalization classes for $e$ particle excitations, as the parameters in the Hamiltonian are varied.  In this model the $m$ particle fractionalization class is trivial.  In Section~\ref{sub:tc_woso}, we discuss models in $TC_0(G)$, the family of toric code models with square lattice symmetry and the restriction of no spin-orbit coupling.  We state Theorem~\ref{thm:nosoc_maintext}, which gives conditions ruling out most of the 2080 symmetry classes (4096 TC symmetry classes) permitted by the general considerations of Ref.~\onlinecite{essin13}.  In particular, only 95 TC symmetry classes, corresponding to 82 symmetry classes, are not ruled out by the constraints of Theorem~\ref{thm:nosoc_maintext}, which are proved in Appendix~\ref{app:nosoc-constraints}.  In fact, all 95 of these TC symmetry classes are realized by models in $TC_0(G)$; these models are exhibited in Sec.~\ref{sub:tc_woso}.  Moving on to the general case of $TC(G)$ where spin-orbit coupling is allowed, Section~\ref{sub:tc_so} states Theorem~\ref{thm:soc_maintext}, which gives constraints similar to but less restrictive than those without spin-orbit coupling; these constraints are proved in Appendix~\ref{app:tcsoc}.  In this case, 945 TC symmetry classes, corresponding to 487 symmetry classes, are not ruled out by the constraints, and again all these classes are realized by explicit models in $TC(G)$.  Some examples of such models are described in Sec.~\ref{sub:tc_so}, with the full catalog of models given in Appendix~\ref{app:models}.

The paper concludes in Sec.~\ref{sec:summary} with a summary and a discussion of two results beyond the special case of solvable toric code models.  There it is argued using a parton gauge theory construction that symmetry classes not realizable in $TC(G)$ can be realized for more generic bosonic models.  In addition, we give a connection between symmetry classes of certain on-site symmetry groups and space group symmetry classes.

Some of the notation used frequently in the paper is collected in Table~\ref{tab:notation}.

\begin{table}
\caption{Notation used in the paper.}
\label{tab:notation}
\begin{tabular}{c | c}
\hline 
Symbol & Meaning \\
\hline
\hline
${\cal H}$ & Hamiltonian \\
\hline
$G$ & \begin{tabular}{@{}c@{}} Symmetry group of ${\cal H}$ \\ (square lattice space group)\end{tabular}   \\
\hline
$\mathcal{G}=\left(V,E\right)$  &  Graph on which the model is defined  \\
\hline
${\mathscr P} : {\cal G} \to T^2$ & Planar projection map into torus $T^2$\\
\hline
$v \in V$ & Vertex $v$ in set of vertices $V$ \\
\hline
$\ell \in E$ & Edge $\ell$ in set of edges $E$ \\
\hline
$s \in W$ & Path $s$ in set of paths $W$ \\
\hline
$C$ & Set of cycles (closed paths) \\
\hline
$C_0$ & Set of contractible cycles \\
\hline
$p \in P$ & Plaquette $p$ in set of plaquettes $P$ \\
\hline
$t \in \bar{W}$ & Cut $t$ in set of cuts $\bar{W}$ \\
\hline
$\bar{C}$ & Set of closed cuts \\
\hline
$\bar{C}_0$ & Set of closed, contractible cuts \\
\hline
$h \in H$ & Hole $h$ in set of holes $H$ \\
\hline
$\sigma^{x}_\ell, \sigma^z_{\ell}$ & Pauli matrix spin operators on edge $\ell$ \\
\hline
${\cal L}^e_s$ & $e$-string on path $s \in W$ \\
\hline
${\cal L}^m_t$ & $m$-string on cut $t \in \bar{W}$ \\
\hline
$A_v, a_v$ & \begin{tabular}{@{}c@{}}  Vertex operator  \\ and corresponding eigenvalue \end{tabular}  \\
\hline
$B_p, b_p$ & \begin{tabular}{@{}c@{}}  Plaquette operator  \\ and corresponding eigenvalue \end{tabular}  \\
\hline
$o=\left(0,0\right)$ & Special points in the plane.\tabularnewline
$\tilde{o}=\left(\frac{1}{2},\frac{1}{2}\right)$ & (Units of length are chosen such that\tabularnewline
$\kappa=\left(0,\frac{1}{2}\right)$ & the size of the unit cell is $1 \times 1$.)\tabularnewline
$\tilde{\kappa}=\left(\frac{1}{2},0\right)$ & \tabularnewline
\hline 
$\left|X\right|$ & Size of any finite set $X$.\tabularnewline
%\hline 
%$e_{r}$, $a_{r}$ & \begin{tabular}{@{}c@{}}  $e_{r}=\prod_{v\in\mathscr{P}^{-1}\left(r\right)}a_{v}$, usually for $r=o,\tilde{o},\kappa,\tilde{\kappa}$; \\  $a_r$ is used only if $\left|\mathscr{P}^{-1}\left(r\right)\right|=1$   \end{tabular} \\
\hline 
$TC(G)$ & \begin{tabular}{@{}c@{}}  Family of toric code models considered,  \\  with spin-orbit coupling allowed  \end{tabular} \\
\hline
$TC_0(G)$ & \begin{tabular}{@{}c@{}}  Family of toric code models considered,  \\  \emph{no} spin-orbit coupling allowed 
 \end{tabular}  \\
\hline
\end{tabular}
\end{table}

\section{Review of $\zz$ topological order}
\label{sec:z2review}

In this paper, we focus on $\zz$ topological order in two dimensions, which is in some sense the simplest type of topological order.  $\zz$ topological order arises in the deconfined phase of $\zz$ lattice gauge theory with gapped bosonic matter  carrying the $\zz$ gauge charge.\footnote{$\zz$ lattice gauge theory with fermionic matter also gives rise to $\zz$ topological order.}  There is an energy gap to all excitations, which can carry $\zz$ gauge charge and/or $\zz$ flux.  There is a statistical interaction between charges and fluxes; the wave function acquires a statistical phase factor $e^{i \pi}$ when a charge moves around a flux or vice versa.  These properties are associated with a four-fold ground state degeneracy on a torus (\emph{i.e.} with periodic boundary conditions), although in some circumstances special boundary conditions are present that reduce the degeneracy.

$\zz$ lattice gauge theory provides a particular concrete realization of $\zz$ topological order, and it is useful to distill the essential features into  a slightly more abstract description.  Every localized excitation above a ground state can be assigned one of four particle types: $1, e, m$, and $\epsilon$.  In terms of lattice gauge theory, $e$ particles are bosonic gauge charges, $m$ particles are $\zz$-fluxes, and $\epsilon$-particles are $e$-$m$ bound states.  Excitations carrying neither $\zz$ charge nor flux are ``trivial,'' and are labeled by $1$.  

$e$, $m$ and $\epsilon$ excitations obey non-trivial braiding statistics and are thus referred to as anyons.  $e$ and $m$ are bosons, while $\epsilon$ is a fermion.  Any two distinct non-trivial particle types (for example, $e$ and $m$), have $\theta = \pi$ mutual statistics, with the wave function acquiring a phase $e^{i \pi}$ when one is brought around the other.  $1$ excitations are bosonic and have trivial mutual statistics with the other particle types.

When two excitations are brought nearby, the particle type of the resulting composite object is well-defined and is given by the fusion rules:
\begin{equation}
\begin{array}{c}
e\times e=m\times m=\epsilon\times\epsilon=1\times1=1,\\
e\times1=e,\, m\times1=m,\,\epsilon\times1=\epsilon,\\
e\times m=\epsilon,\, e\times\epsilon=m,\, m\times\epsilon=e.
\end{array}\label{eq:fusions}
\end{equation}
It is a very important property that only $1$ excitations can be locally created; that is, action with local operators cannot produce a single, isolated $e$, $m$ or $\epsilon$ (at least away from edges of the system, if there are open boundaries).  The fusion rules then  tell us that a pair of $e$, $m$ or $\epsilon$ excitations can be created locally.  An anyon can be moved from one position to another by acting with a non-local string operator connecting the initial and final positions.  There are distinct string operators for each type of anyon.

We remark that the fusion and braiding properties are invariant under the relabeling $e \leftrightarrow m$, which means we are free to make such a relabeling -- this is a kind of $\zz$ electric-magnetic duality.  This feature is important for a proper counting of symmetry classes.

\section{Review: toric code model on the square lattice}
\label{sec:toricreview}

We now review Kitaev's toric code model\cite{kitaev03} on the square lattice, which is the simplest model realizing $\zz$ topological order.  We consider a $L \times L$ square lattice with periodic boundary conditions (forming a torus), and we label vertices by $v$, edges by $\ell$, and square plaquettes by $p$.  The degrees of freedom are spin-1/2 spins, residing on the edges.  Local operators are then built from Pauli matrices $\sigma^{\mu}_\ell$ ($\mu = x,y,z$) acting on the spin at $\ell$.

We introduce operators associated with vertices and plaquettes,
\begin{eqnarray}
A_v &=& \prod_{\ell \in \operatorname{star}(v) }  \sigma^x_\ell \\
B_p &=& \prod_{\ell \in p} \sigma^z_\ell \text{,}
\end{eqnarray}
where $p$ contains the four edges in the perimeter of a square plaquette, and $\operatorname{star}(v)$ is the set of four edges touching $v$ (see Fig.~\ref{fig:toric}).  The Hamiltonian is
\begin{equation}
{\cal H} =-K^{e}\sum_{v}A_{v}-K^{m}\sum_{p}B_{p} \text{,} \label{eq:simple_toric}
\end{equation}
with $K^e, K^m > 0$.
It is easy to see that
\begin{equation}
\left[ A_v, A_{v'} \right] = \left[ B_p, B_{p'} \right]  = \left[ A_v , B_p \right] = 0 \text{,}
\end{equation}
rendering the Hamiltonian exactly solvable. The energy eigenstates can be chosen to satisfy
\begin{eqnarray}
A_v | \psi \rangle &=& a_v | \psi \rangle \\
B_p | \psi \rangle &=& b_p | \psi \rangle \text{,}
\end{eqnarray}
where $a_v, b_p \in \{ \pm 1 \}$.

The Hilbert space has dimension $2^{2 L^2}$, so we need $2 L^2$ independent Hermitian operators with eigenvalues $\pm 1$ to form a complete set of commuting observables (CSCO), whose eigenvalues uniquely label a basis of states.  Due to the periodic boundary conditions, $\prod_{v}A_{v}=\prod_{p}B_{p}=1$, and the $A_v$ and $B_p$ only give $2 L^2 - 2$ independent operators.  To obtain a CSCO, we need two additional operators, and one choice is given by
\begin{equation}
L^e_x =\prod_{\ell \in s_x} \sigma_{\ell}^{z},\, L^e_y =\prod_{\ell \in s_y} \sigma_{\ell}^{z},\label{eq:loop}
\end{equation}
with eigenvalues $l^e_{x,y} \in \{ \pm 1 \}$, where $s_x$, $s_y$ are non-contractible loops winding around the system in the $x$ and $y$ directions, respectively, as shown in Fig.~\ref{fig:toric}.  The eigenvalues $\{ a_v, b_p, l^e_x, l^e_y \}$ uniquely label a basis of energy eigenstates.  In particular, there are four ground states with $a_v = b_p = 1$, a sign of  $\zz$ topological order. 

Excitations above the ground state reside at vertices with $a_v = -1$, and plaquettes with $b_p = -1$.  These excitations have no dynamics; this is tied to the exact solubility of the model, and adding generic perturbations to the model causes the excitations to become mobile.  We identify $a_v = -1$ vertices as $e$ particles, and $b_p = -1$ plaquettes as $m$ particles.  $\epsilon$ excitations are $e$-$m$ pairs.  Acting on a ground state with $\sigma^z_\ell$ creates a pair of $e$ particles, at the two vertices touching $\ell$.  Similarly, acting with $\sigma^x_\ell$ creates two $m$ particles in the two plaquettes touching $\ell$.  Since any operator can be built from products of Pauli matrices, it follows that isolated $e$ and $m$ excitations cannot be created locally.

We now introduce $e$ and $m$ string operators.  To define an $e$-string operator, let $s$ be a set of edges $\ell$ forming a connected path, which may be either closed or open (see Fig.~\ref{fig:stringop}).  Then we define
\begin{equation}
\cL^e_s = \prod_{\ell \in s} \sigma^z_\ell \text{.} \label{eq:e-string}
\end{equation}
Suppose $s$ is an open path with endpoints $v_1$ and $v_2$.  If  $\cL^e_s$ acts on a ground state, it creates $e$ particles at $v_1$ and $v_2$.  Alternatively, acting on a state with an $e$ particle at $v_1$ and none at $v_2$, $\cL^e_s$ moves the $e$ particle from $v_1$ to $v_2$.  On the other hand, if $s$ is a closed path and is contractible (\emph{i.e.} does not wind around the torus), and if $| \psi_0 \rangle$ is a ground state, then $\cL^e_s | \psi_0 \rangle = | \psi_0 \rangle$.

$m$-strings are defined on a cut $t$, which contains the set of edges intersected by a path drawn on top of the lattice, running from plaquette to plaquette, as shown in Fig.~\ref{fig:stringop}.  Alternatively, $t$ can be viewed as a set of edges in the dual lattice forming a connected path.  The $m$-string operator is then
\begin{equation}
\cL^m_t = \prod_{\ell \in t} \sigma^x_{\ell} \text{.} \label{eq:m-string}
\end{equation}
Just as with $e$-strings, if $t$ is an open cut, with endpoints in two plaquettes $p_1$, $p_2$, $\cL^m_t$ can be used to create a pair of $m$ particles or to move a single $m$ particle from one plaquette to another.  If $t$ is a closed, contractible cut, $\cL^m_t$ gives unity acting on a ground state.

If the path $s$ and the cut $t$ cross $n_c(s,t)$ times, then
\begin{equation}
\cL^e_s \cL^m_t = (-1)^{n_c (s, t) } \cL^m_t \cL^e_s \text{.}
\end{equation}
This can be used to verify that the $e$, $m$ and $\epsilon$ excitations indeed obey the braiding statistics of $\zz$ topological order.

\begin{figure}
\includegraphics[width=0.6\columnwidth]{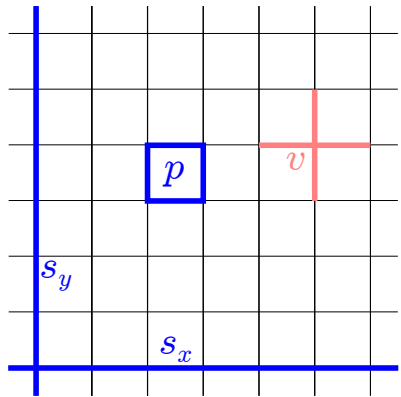}
\caption{(Color online) Illustration of some geometrical objects important in the square lattice toric code model.  
The edges in plaquette $p$ are shown as thick dark bonds (blue online), while the edges in $\operatorname{star}(v)$ are thick gray bonds (pink online).  The two strings $s_x$ and $s_y$ winding periodically around the system are also shown as thick dark bonds (blue online).}
\label{fig:toric}
\end{figure}

\begin{figure}
\includegraphics[width=0.6\columnwidth]{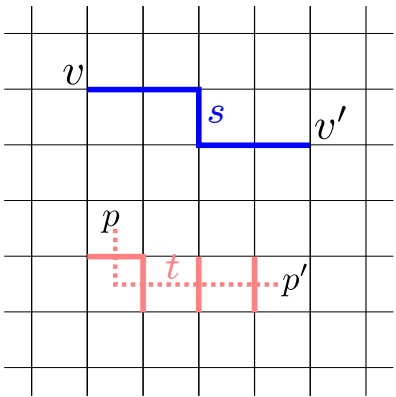}
\caption{(Color online) Depiction of $e$ and $m$ strings in the square lattice toric code.  $s$ is an open $e$-string joining vertices $v$ and $v'$, denoted with thick dark bonds (blue online).  $t$ is an open cut joining plaquettes $p$ and $p'$, shown as a dotted line.  The cut $t$ contains the thick gray bonds (pink online) intersected by the dotted line.}
\label{fig:stringop}
\end{figure}

\section{Toric codes on general two-dimensional lattices with space group symmetry}
\label{sec:genlatt}

We now introduce the family of models studied in this paper, which are generalizations of the toric code to arbitrary lattices with square lattice space group symmetry.  Sometimes it will be convenient to refer to this family of models as $TC(G)$, where in this paper $G$ is always the square lattice space group.  We will also introduce a smaller family of models $TC_0(G) \subset TC(G)$.  These two families are distinguished in that ``spin-orbit coupling'' (as defined below) is allowed for models in $TC(G)$, but is absent in $TC_0(G)$.

We begin by defining a toric code model on an arbitrary finite connected graph ${\cal G}$ with sets of vertices and edges denoted by $V$ and $E$, respectively.  The number of edges (vertices) is denoted $|E|$ ($|V|$).  We allow for the possibility that two vertices may be joined by more than one edge.  Spin-1/2 degrees of freedom reside on edges, and we again denote work with Pauli matrices $\sigma^{\mu}_\ell$ ($\mu = x,y,z$) acting on the spin at edge $\ell \in E$.

To proceed, it is helpful to introduce some notation and terminology.  A \emph{path} is a sequence of edges $s = \ell_1 \ell_2 \cdots \ell_n$ joining successive vertices; that is, $\ell_i$ and $\ell_{i+1}$ are incident on a common vertex.  Paths are considered unoriented, so that $\ell_1 \ell_2 \cdots \ell_n = \ell_n \cdots \ell_2 \ell_1$.  The set of all paths is denoted $W$.  A path may either be open with distinct endpoints $v_1, v_2 \in V$, or closed.  Two open paths $s$ and $s'$ sharing an endpoint can be composed into the path $s s'$. Since an edge may appear in $s$ more than once, 
more precisely the definition \eqref{eq:e-string} of $e$-string operator should be understood as 
\begin{equation}
\mathcal{L}_{s}^{e}=\prod_{\ell\in s}\sigma_{\ell}^{z}=\sigma^z_{\ell_{1}}\sigma^z_{\ell_{2}}\cdots\sigma^z_{\ell_{n}},\label{eq:e-string2}
\end{equation}
for for $s=\ell_{1}\ell_{2}\cdots \ell_{n}$. Since operators in the product commute, there is no harm to interpret $s$ 
as multiset of edges as well. In this paper, we use the product notation $\prod_{\ell\in X}$ for all three cases: $X$ is a set, a multiset or a sequence of edges.

The set of open paths is denoted $W_o \subset W$, while closed paths are called \emph{cycles}, and the set of cycles is $C \subset W$.  $e$-string operators are defined on paths $s \in W$ by $\cL^e_s = \prod_{\ell \in s} \sigma^z_\ell$.
An important part of the specification of a model will be the selection of a subset $P \subset C$, where elements $p \in P$ are called plaquettes.  The choice of $P$ is not entirely arbitrary, and is required to satisfy certain properties discussed below.

Just as for the square lattice,
\begin{eqnarray}
A_v &=& \prod_{\ell \in \operatorname{star}(v) }  \sigma^x_\ell \\
B_p &=& \prod_{\ell \in p} \sigma^z_\ell \text{,}
\end{eqnarray}
where $p \in P$, and $\operatorname{star}(v)$ is again the set of edges touching $v$.  It is again easy to see that
\begin{equation}
\left[ A_v, A_{v'} \right] = \left[ B_p, B_{p'} \right]  = \left[ A_v , B_p \right] = 0 \text{.}
\end{equation}
The Hamiltonian is
\begin{equation}
{\cal H} = - \sum_{v \in V} K^e_v A_v - \sum_{p \in P} K^m_p B_p \text{,}
\end{equation}
where now the coefficients $K^e_v$, $K^m_p$ are allowed to depend on the vertex or plaquette.  Only the signs of the coefficients will be important, so for convenience of notation we take $K^e_v, K^m_p \in \{ \pm 1 \}$.  Energy eigenstates can again be labeled by $a_v, b_p \in \{ \pm 1 \}$, the eigenvalues of $A_v$ and $B_p$, respectively.

Any ground state will satisfy $a_v = K^e_v$ and $b_p = K^m_p$, \emph{provided} it is possible to find such a state.  This is not guaranteed, as the couplings in the Hamiltonian could be frustrated.  We will assume the Hamiltonian is ``frustration-free,'' meaning it is possible to find at least one ground state with $a_v = K^e_v$, $b_p = K^m_p$.\footnote{This is the case provided $K^e_v$ and $K^m_p$ are compatible with constraints obeyed by $A_v$ and $B_p$ operators.  More precisely, we have $\prod_v A_v = 1$, which implies $K^e_v$ must satisfy $\prod_v K^e_v = 1$.  In addition, suppose $P'$ is a subset of $P$ for which $\prod_{p \in P'} B_p = 1$, then we must have $\prod_{p \in P'} K^m_p = 1$.}

Our discussion so far is for a general graph, but we want to specialize to two-dimensional lattices.  Essentially, this just means that we can draw the graph in two-dimensional space (with periodic boundary conditions), so that the resulting drawing has the symmetry of the square lattice.  We do \emph{not} assume the graph is planar; for instance, edges are allowed to cross or stack on top of each other when the graph is drawn in two dimensions.   

In order to make general statements about the family of models considered, it will be useful to be more precise.  First, letting $G$ be the square lattice space group, we introduce an action of $G$ on ${\cal G}$.  Group elements $g \in G$ act on vertices and edges of the graph, and we write $v \mapsto g v$, $g \mapsto g \ell$.  $G$ is generated by translation $x \to x+1$ ($T_x$), reflection $x \to -x$ ($P_x$), and reflection $x \leftrightarrow y$ ($P_{xy}$).  Translation by $y \to y+1$ is given in terms of the generators by $T_y = P_{xy} T_x P_{xy}$.  The group $G$ can be defined in terms of the generators by requiring them to obey the relations,
\begin{eqnarray}
P_{x}^{2} & = & 1,\label{eq:px}\\
P_{xy}^{2} & = & 1,\label{eq:pxy}\\
\left(T_{x}P_{x}\right)^{2} & = & 1,\label{eq:txpx}\\
\left(P_{x}P_{xy}\right)^{4} & = & 1,\label{eq:pxpxy}\\
T_{x}T_{y}T_{x}^{-1}T_{y}^{-1} & = & 1,\label{eq:txty}\\
T_{y}P_{x}T_{y}^{-1}P_{x}^{-1} & = & 1.\label{eq:typx}
\end{eqnarray}

We wish to consider a $L \times L$ lattice with periodic boundary conditions, with $L$ the integer number of square primitive cells in the $x$ and $y$ directions.  More formally, for all $v \in V$, we assume $v = (T_x)^{n_x} (T_y)^{n_y} v$ if and only if $n_x, n_y = 0 \mod L$, with the same statement holding for all $\ell \in E$.

We now introduce the planar projection ${\mathscr P} : {\cal G} \to T^2$, where $T^2$ is the 2-torus, viewed as a square with dimensions $L \times L$ and periodic boundary conditions.  ${\mathscr P}$ is a continuous map that sends vertices to points and edges to curves. (See Fig.~\ref{fig:tcg-caption} for an example.) Symmetry operations $g \in G$ act on the graph ${\cal G}$ as described above, and also act naturally on $T^2$ as rigid motions of space.  We require
\begin{equation}
g {\mathscr P} = {\mathscr P} g \text{,}
\end{equation}
which means the action of $G$ on ${\cal G}$ is compatible with the action of rigid motions on the planar projection ${\mathscr P}({\cal G})$.  The additional structure thus introduced ensures that $G$ is truly playing the role of a space group.

The above discussion implies that the planar projection ${\mathscr P}({\cal G})$ is an $L \times L$ grid of $1 \times 1$ square primitive cells.
We note that edges in ${\mathscr P}({\cal G})$ are allowed to cross at points other than vertices.  Vertices and edges are also allowed to stack on top of one another; that is, it may happen that ${\mathscr P}(v_1) = {\mathscr P}(v_2)$ for $v_1 \neq v_2$.  It is always possible to choose ${\mathscr P}(\ell)$ to be a straight line connecting its endpoints, although sometimes it will be convenient not to do so.

Now we are in a position to discuss the requirements on the set of plaquettes $P$.  First, any plaquette $p \in P$ should be in some sense local.  This can be achieved by requiring there to be a maximum size (by some measure that does not need to be precisely defined) for all $p \in P$, where the maximum size is independent of $L$.  Second, we require that any \emph{contractible} cycle can be decomposed into plaquettes.   Non-contractible cycles are those that, under the planar projection, wind around either direction of $T^2$ an odd number of times, and all others are contractible.  We let $C_0 \subset C$ be the set of contractible cycles.  The assumption that contractible cycles can be decomposed into plaquettes means that, given $s \in C_0$, there exists $\{ p_1, \dots, p_n\} \subset P$ so that $\cL^e_s = \prod_{i = 1}^n B_{p_i}$.  The physical reason for this requirement is that it ensures there are no local zero-energy excitations, as there would certainly be if we chose $P$ to be too small.  

As in the square lattice, we introduce two large cycles $s_x$ and $s_y$ that wind around the torus in the $x$ and $y$ directions, respectively.  The operators $\{ A_v \}$, $\{ B_p \}$, $\cL^e_{s_x}$ and $\cL^e_{s_y}$ form a complete set of commuting observables (Appendix~\ref{app:complete-set}).  Denoting eigenvalues of $\cL^e_{s_x, s_y}$ by $l^e_{x,y} \in \{ \pm 1\}$, it is then easy to see that ${\cal H}$ has a four-fold degenerate ground state, corresponding to the four choices of $l^e_{x,y}$ with the other eigenvalues fixed to $a_v = K^e_v$ and $b_p = K^m_p$.

Just as for the square lattice toric code, $e$ particles lie at vertices where $a_v = - K^e_v$; that is, where $a_v$ differs from its ground state value.  For $s \in W_o$, the $e$-string operator $\cL^e_s$ can be used to create $e$ particles at the two endpoints, or to move an $e$ particle from one endpoint to the other.  

Identifying $m$ particles is more tricky; the basic insight required is that $m$ particles should correspond to a threading of $\zz$ flux through ``holes'' in the planar projection ${\mathscr P}({\cal G})$.  It is easiest to proceed by defining $m$-strings, which are defined on \emph{cuts} $t \in \bar{W}$.  A cut $t$ is defined as follows:  (1) Draw a curve in $T^2$ that has no intersection with vertices ${\mathscr P}(v)$, and whose intersection with each edge ${\mathscr P}(\ell)$ contains at most a finite number of points, at which the curve is not tangent to ${\mathscr P}(\ell)$.  If the curve is open, we assume its endpoints do not lie in ${\mathscr P}({\cal G})$.  (2) The cut $t$ is then given by the sequence of edges intersected by the curve.  A cut is \emph{closed} if the curve in (1) is closed, and is \emph{simple} 
 if the curve has no self-intersections.  It is clear that a given curve produces a unique cut, but there are many possible curves that produce the same cut.

We define a $m$-string operator on a cut $t \in \bar{W}$ by  $\cL^m_t = \prod_{\ell \in t} \sigma^x_\ell$.  If $t$ is an open cut, then $\cL^m_t$ acting on a ground state creates $m$ particles at the two endpoints.  The endpoints of the $m$-string, and thus the $m$ particles it creates, naturally reside at the \emph{holes} in the planar projection; more precisely, these are the connected components of $T^2 - {\mathscr P}({\cal G})$.  We denote the set of all holes by $H$ with elements $h \in H$.  Not all $m$ excitations can be created as described above, but arbitrary such excitations can be created by first acting with $\cL^m_t$ on a ground state, then acting subsequently with operators localized near the $m$ particles created by the string operator.

Finally, we need to specify the action of symmetry on the spin degrees of freedom themselves.  Letting $U_g$ be the unitary operator representing $g \in G$, we consider
\begin{equation}
U_g \sigma^x_\ell U^{-1}_g = c^x_\ell(g) \sigma^x_{g \ell} , \quad U_g \sigma^z_\ell U^{-1}_g = c^z_\ell(g) \sigma^z_{g \ell} \text{,}
\end{equation}
assuming symmetries do not swap anyon species.
Since $U_g \sigma^{x,z}_\ell U^{-1}_g$ are hermitian and unitary simultaneously, we must have $c^{x,z}_{\ell}(g) \in \{ \pm 1 \}$.  This satisfies a general requirement that space group symmetry should be realized as a product of an on-site operation, with another operation that merely moves degrees of freedom (\emph{i.e.} $\sigma^{\mu}_\ell \mapsto \sigma^{\mu}_{g \ell}$).\footnote{The origin of these requirements is the fact that these properties holds for hold for all electrically neutral bosonic degrees of freedom  (\emph{e.g.} electron spins, bosonic atoms) that can be microscopic constituents of a condensed matter system.}  Subject to this requirement, this is the most general action of symmetry with the property that $e$-strings are taken to $e$-strings, and $m$-strings to $m$-strings; for example, $U_g \cL^e_s U^{-1}_g = (\pm 1) \cL^e_{g s}$.

Actually we need to impose a further requirement, which is that symmetry must act linearly (as opposed to projectively) on the spin operators.\cite{Note3}  In particular,
\begin{equation}
U_{g_1} U_{g_2} \sigma^{x,z}_{\ell} U^{-1}_{g_2} U^{-1}_{g_1} = U_{g_1 g_2} \sigma^{x,z}_{\ell} U^{-1}_{g_1 g_2} \text{.}
\label{eqn:linear-action}
\end{equation}
This imposes the restriction
\begin{equation}
c^{x,z}_{g_2 \ell} (g_1) c^{x,z}_\ell (g_2) = c^{x,z}_\ell (g_1 g_2) \text{,} \label{eqn:c-restriction}
\end{equation}
which holds for all $\ell \in E$ and $g_1, g_2 \in G$.  These conditions do not fix the overall ${\rm U}(1)$ phase of $U_g$, which can be adjusted (as a function of $g$) as desired.

The phase factors $c^{x,z}_{\ell}(g)$ can be modified by the unitary ``gauge'' transformation $\sigma^{x,z}_\ell  \to \gamma^{x,z}_\ell \sigma^{x,z}_\ell$, with $\gamma^{x,z}_\ell \in \{ \pm 1 \}$, which sends
\begin{equation}
c^{x,z}_\ell(g) \to \gamma^{x,z}_\ell \gamma^{x,z}_{g \ell} c^{x,z}_\ell (g) \text{.}  \label{eqn:gauge-transformation}
\end{equation}
It is always possible to choose a gauge where $c^{x,z}_\ell (T) = 1$, for all $\ell \in E$ and all translations $T \in G$; this is so because $c^{x,z}_\ell(T_x)$ and $c^{x,z}_\ell (T_y)$ behave under gauge transformation like the $x$ and $y$ components of a flux-free vector potential, residing on the links of a square lattice generated by acting on $\ell$ with translation.  We shall make this gauge choice without further comment throughout the paper.

If, in addition, it is possible to choose a gauge where $c^{x,z}_\ell (g) = 1$ for all $\ell \in E$ and $g \in G$, then by definition the model is  in $TC_0(G)$, and we say there is no ``spin-orbit coupling.''  The reason for this terminology is that, in this case, space group operations have no action on spins beyond moving them from one point in space to another.  The case of no spin-orbit coupling is simpler to analyze, and we will discuss it first before handling the general case.

It is shown in Appendix~\ref{app:psi0} that for $L$ even, it is possible to find a ground state $| \psi_{0 e} \rangle$ and make a choice of phase for $U_g$ so that
\begin{eqnarray}
U_g | \psi_{0 e} \rangle &=& | \psi_{0 e} \rangle  \label{eqn:ug-psi0e} \\
\cL^e_{s_x} | \psi_{0 e} \rangle &=& \cL^e_{s_y} | \psi_{0 e} \rangle = | \psi_{0 e} \rangle \text{,}
\end{eqnarray}
where $s_x$ and $s_y$ are closed paths chosen as described in Appendix~\ref{app:psi0} to wind once around the system in the $x$ and $y$ directions, respectively.  For the same phase choice of $U_g$, combining Eq.~(\ref{eqn:ug-psi0e}) with Eq.~(\ref{eqn:linear-action}) implies $U_{g_1} U_{g_2} = U_{g_1 g_2}$.  From now on, when we study $e$ particle excitations, we always focus on states that can be constructed by acting on $| \psi_{0 e} \rangle$ with $e$-string operators.

Appendix~\ref{app:psi0} also shows that, for $L$ even, there is a ground state $|\psi_{0 m} \rangle$ and a phase choice for $U_g$, satisfying
\begin{eqnarray}
U_g | \psi_{0 m} \rangle &=& | \psi_{0 m} \rangle  \\
\cL^m_{t_x} | \psi_{0 m} \rangle &=& \cL^m_{t_y} | \psi_{0 m} \rangle = | \psi_{0 m} \rangle \text{.}
\end{eqnarray}
Here, the electric strings have been replaced with magnetic strings, with $t_x$ and $t_y$ appropriately chosen closed cuts winding once around the system in the $x$ and $y$ directions, respectively.  When studying $m$ particle excitations, we will always consider states constructed by applying $m$-string operators to $| \psi_{0 m} \rangle$.

It should be noted that $| \psi_{0 e} \rangle$ and $| \psi_{0 m} \rangle$ cannot be the same state, because, for instance, $\cL^e_{s_x}$ and $\cL^m_{t_y}$ anticommute.  Moreover, the phase choice required to make $| \psi_{0 e} \rangle$ symmetry-invariant may not be the same as the corresponding choice for $| \psi_{0 m} \rangle$.  These points will not be problematic for us, because we always focus on excited states with either $e$ particles, or $m$ particles, but not both.  Using $ | \psi_{0 e} \rangle$ to construct $e$ particle states, and similarly $| \psi_{0 m} \rangle$ for $m$ particle states, simply provides a convenient means to calculate the $e$ and $m$ fractionalization classes.

\section{Fractionalization and Symmetry Classes}

\subsection{Review of fractionalization and symmetry classes}
\label{sec:fracreview}
We now consider in more depth the action of square lattice space group symmetry $G$ in the general class of solvable models introduced in Sec.~\ref{sec:genlatt}, showing how to determine the fractionalization classes of $e$ and $m$ particles, and the corresponding symmetry class.  We first review the general notions of fractionalization and symmetry classes, before exposing in detail the corresponding structure for the solvable models (Sec.~\ref{sec:fc-in-solvable-models}).  Readers unfamiliar with this subject may find the  review rather abstract, so we would like to emphasize that the objects involved appear in concrete and explicit fashion in the discussion of the solvable models.

Each non-trivial anyon ($e$, $m$ and $\epsilon$ in the toric code) has a corresponding fractionalization class, that describes the action of symmetry on single anyon excitations of the corresponding type.  (We assume that symmetry does not permute the anyon species.)  This structure follows from the fact that the action of symmetry factorizes into an action on individual isolated anyons.  Since physical states must contain even numbers of $e$ particles, as an example we consider a state $| \psi_{e e} \rangle$ with two $e$ particles, labeled 1 and 2.  Following the arguments of Ref.~\onlinecite{essin13}, we assume that
\begin{equation}
U_g | \psi_{e e} \rangle = U^e_g(1) U^e_g(2) | \psi_{e e} \rangle \text{,} \label{eqn:symfrac0}
\end{equation}
where $U^e_g(i)$ gives the action of symmetry on anyon $i = 1,2$.  

The physics is invariant under a redefinition
\begin{equation}
U^e_g(i) \to \lambda(g) U^e_g(i) , \quad \lambda(g) \in \{ \pm 1 \} \text{,}
\end{equation}
which we refer to as a \emph{projective transformation}.  The reason for this terminology is that the $U^e_g$ operators form a projective representation of $G$, expressed by writing
\begin{equation}
U^e_{g_1} U^e_{g_2} = \omega_e(g_1, g_2) U^e_{g_1 g_2} \text{,}
\end{equation}
where we have suppressed the anyon label $i$, and $\omega_e(g_1, g_2) \in \{ \pm 1 \}$ is referred to as a $\zz$ factor set.  The factor set satisfies the condition
\begin{equation}
\omega_e(g_1, g_2) \omega_e(g_1 g_2, g_3) = \omega_e(g_2, g_3) \omega_e(g_1, g_2 g_3) \text{,}
\end{equation}
which follows from the associative multiplication of $U^e_g$ operators.  The factor set is not invariant under projective transformations, but instead transforms as
\begin{equation}
\omega_e(g_1, g_2) \to \lambda(g_1) \lambda(g_2) \lambda(g_1 g_2) \omega_e(g_1, g_2) \text{.}
\end{equation}
A projective transformation is analogous to a gauge transformation that does not affect the physics, so such transformations should be used to group factor sets into equivalence classes.  We denote by $[\omega_e]$ the equivalence class containing the factor set $\omega_e$.  These equivalence classes are the possible fractionalization classes for $e$ particles.  It will not be important for the discussion of the present paper, but we mention that the set of fractionalization classes is the second group cohomology $H^2(G, \zz)$.  The discussion proceeds identically for $m$ particles, with $\omega_m$ the corresponding factor set, and $[\omega_m] \in H^2(G, \zz)$ the fractionalization class.  

A complete specification of fractionalization classes defines a symmetry class.  It is enough to specify $[\omega_e]$ and $[\omega_m]$, because these determine uniquely the $\epsilon$ fractionalization class.\cite{essin13}  Therefore a symmetry class is specified by the pair
\begin{equation}
\fS = \langle [\omega_e] , [\omega_m] \rangle \text{.}
\end{equation}
Because all properties of $\zz$ topological order are invariant under $e \leftrightarrow m$ (see Sec.~\ref{sec:z2review}), symmetry classes related by this relabeling are considered equivalent, that is
\begin{equation}
\langle [\omega_e], [\omega_m] \rangle \simeq  \langle [\omega_m], [\omega_e] \rangle \text{.} \label{eqn:relabel-equivalence}
\end{equation}

Despite the lack of a fundamental distinction between $e$ and $m$ particles, there is a distinction in the solvable toric code models, as is clear from the discussion of these excitations in Sec.~\ref{sec:genlatt}.  While this distinction is only well-defined within the context of the solvable models, it is not just a matter of notation; in general, we do not restrict to planar lattices, so there is not expected to be an exact duality exchanging $e \leftrightarrow m$.  Because it is relevant for the construction of solvable models, it will be useful to define \emph{toric code symmetry classes}, or TC symmetry classes, that distinguish between $e$ and $m$ particles.  A TC symmetry class is simply an ordered pair $([\omega_e], [\omega_m])$.

To determine fractionalization and symmetry classes, it is convenient to work with the generators and their relations [Eqs.~(\ref{eq:px}-\ref{eq:typx})].  Focusing on $e$ particles for concreteness, the $U^e_g$ operators obey the group relations up to possible minus signs, that is
\begin{eqnarray}
(U^e_{P_x} )^2 &=& \sigma^e_{px} \label{eqn:projpx} \\
(U^e_{P_{xy}} )^2 &=& \sigma^e_{pxy} \\
(U^e_{T_x} U^e_{P_x} )^2 &=& \sigma^e_{txpx} \\
(U^e_{P_x} U^e_{P_{xy}} )^4 &=& \sigma^e_{pxpxy} \\
U^e_{T_x} U^e_{T_y} (U^e_{T_x})^{-1} (U^e_{T_y})^{-1} &=& \sigma^e_{txty} \\
U^e_{T_y} U^e_{P_x} (U^e_{T_y})^{-1} (U^e_{P_x})^{-1} &=& \sigma^e_{typx} \label{eqn:projtypx} \text{,}
\end{eqnarray}
where $\sigma^e_{px} \in \{ \pm 1 \}$, and similarly for the other $\sigma^e$ parameters.  The $\sigma^e$'s are invariant under projective transformations, and moreover uniquely specify the fractionalization class $[\omega_e]$.\cite{essin13}  In addition, it was shown that each of the $2^6 = 64$ possible choices of the $\sigma^e$'s is mathematically possible; that is, there exists a projective representation for all choices of $\sigma^e$'s.\cite{essin13}  The same considerations lead to six $\sigma^m$ parameters characterizing the $m$ fractionalization class.  We see that 2080 symmetry classes  (4096 TC symmetry classes) are allowed by the classification of Ref.~\onlinecite{essin13}.  The reader may recall that Ref.~\onlinecite{essin13} found a larger number of symmetry classes by the same type of analysis -- the difference arises because Ref.~\onlinecite{essin13} also considered time reversal symmetry, while here we focus only on space group symmetry.

\subsection{Fractionalization and symmetry classes in the solvable models}
\label{sec:fc-in-solvable-models}

The solvable models are well-suited to the study of fractionalization and symmetry classes because the $U^e_g$ and $U^m_g$ operators can be explicitly constructed.  We focus first on $e$ particles.  It is sufficient to consider states with only two $e$ particle excitations, of the form
\begin{equation}
| \psi_e (s) \rangle = \cL^e_s | \psi_{0e} \rangle \text{,}
\end{equation}
with $s$ an open path, and $e$ particles residing on the endpoints $v_1(s)$ and $v_2(s)$.  The action of symmetry on this state is given by
\begin{equation}
U_g | \psi_e (s) \rangle = c^z_s(g) | \psi_e (g s) \rangle \text{,} \label{eqn:ug-action}
\end{equation}
where $c^z_s(g) = \prod_{\ell \in s} c^z_\ell(g)$.

The goal is to construct and study operators $U^e_g$ that act on single $e$ particles, reproducing the action of $U_g$ on  states $|\psi_e(s)\rangle$.  Consider the pair $(g, v) \in G \times V$, where $v$ is the vertex at which an $e$ particle resides, and $g \in G$ is the group operation of interest.  To each such pair we associate a number $f^e_g(v) \in \{ \pm 1 \}$ and a path $s^e_g(v)$.  The path $s^e_g(v)$ has endpoints $v$ and $gv$.  (Note that $s^e_g(v)$ is a cycle or a null path if $gv = v$.) From this data we form the operator
\begin{equation}
U^e_g(v) = f^e_g(v) \cL^e_{s^e_g(v)}  \text{.}
\end{equation}
By construction, this operator moves an $e$ particle from $v$ to $g v$, and is thus a reasonable candidate to realize the action of $g \in G$ on single $e$ particles.  In order to reproduce Eq.~(\ref{eqn:ug-action}), we require the $U^e_g(v)$ operators to obey the relation
\begin{equation}
U_g | \psi_e(s) \rangle = U^e_g[ v_1(s) ] U^e_g [ v_2(s) ] | \psi_e(s) \rangle \text{,}  \label{eqn:symfrac}
\end{equation}
which has to hold for all open paths $s \in W_o$ and all $g \in G$.  We refer to a set of $U^e_g(v)$ operators satisfying this relation as an $e$-localization of the symmetry $G$.

It should be noted that there is some redundancy in the data used to define $U^e_g(v)$.  Keeping its endpoints fixed, the path $s^e_g(v)$ can be deformed arbitrarily, at the expense of a phase factor.  When acting on states $|\psi_e(s)\rangle$ as we consider (or even on states with many $e$ particles, but no $m$ particles), this phase factor is independent of the state, and can be absorbed into a redefinition of $f^e_g(v)$.

At this point, it is important to ask whether it is always possible to find an $e$-localization, and, if it exists, whether the $e$-localization is in some sense unique.  Indeed, in Appendix~\ref{app:eloc} we prove that for toric code models as described in Sec.~\ref{sec:genlatt}, it is always possible to find an $e$-localization of $G$.  Moreover, the $e$-localization is unique up to projective transformations $U^e_g(v) \to \lambda(g) U^e_g(v)$, where $\lambda(g) \in \{ \pm 1 \}$.  This means that the $e$-localization is a legitimate tool to study the action of symmetry on $e$ particles in the solvable models.

To determine the $e$ fractionalization class from the $e$-localization, we consider the product
\begin{equation}
U^e_{g_1}(g_2 v) U^e_{g_2}(v) = F(g_1, g_2, v) U^e_{g_1 g_2}(v) \text{,}
\end{equation}
where $F(g_1, g_2, v) \in \{ \pm 1\}$, and this equation holds acting on all states containing no $m$ particle excitations [including $| \psi_e(w) \rangle$].  This relation holds because both sides of the equation are $e$ string operators joining $v$ to $g_1 g_2 v$, and can differ only by a phase factor depending on $g_1$, $g_2$ and $v$.

We now show that $F(g_1, g_2, v)$ is independent of $v$, and forms a $\zz$ factor set, so that we can write  $F(g_1, g_2, v) = \omega_e(g_1, g_2)$.  Suppose that for some $g_1, g_2$, and some vertices $v_i$, $v_j$, we have $F(g_1, g_2, v_i) \neq F(g_1, g_2, v_j)$.  Then consider the state $| \psi_e(s_{i j}) \rangle$, where $s_{i j}$ is a path joining $v_i$ to $v_j$.  We have
\begin{eqnarray}
&& U_{g_1 g_2} | \psi_e(s_{i j}) \rangle = U_{g_1} U_{g_2}  | \psi_e(s_{i j}) \rangle \\
&=& U^e_{g_1}(g_2 v_i) U^e_{g_2}(v_i) U^e_{g_1}(g_2 v_j) U^e_{g_2}(v_j)  | \psi_e(s_{i j}) \rangle \nonumber \\
&=& F(g_1, g_2, v_i) F(g_1, g_2, v_j) U^e_{g_1 g_2}(v_i) U^e_{g_1 g_2}(v_j)  | \psi_e(s_{i j}) \rangle \nonumber \\
&=& - U_{g_1 g_2} | \psi_e(s_{i j}) \rangle \text{,} \nonumber
\end{eqnarray}
a contradiction.  This shows $F = F(g_1, g_2)$, independent of $v$.  The associativity condition required for $F(g_1, g_2)$ to be a factor set follows from equating the two ways of associating the product in
\begin{equation}
U^e_{g_1} (g_2 g_3 v) U^e_{g_2} (g_3 v) U^e_{g_3}(v) | \psi_e(s) \rangle \text{,}
\end{equation}
where $| \psi_e(s) \rangle$ has one $e$ particle at $v$.
Thus we have shown
\begin{equation}
U^e_{g_1}(g_2 v) U^e_{g_2}(v) = \omega_e(g_1, g_2) U^e_{g_1 g_2}(v) \text{,}
\end{equation}
with $\omega_e$ a $\zz$ factor set.  This operator equation holds acting on all states of the form $| \psi_e(s)\rangle$, and more generally on states with any number of $e$ particle excitations created by acting on $|\psi_0\rangle$ with $e$-string operators.  The freedom to transform the $e$-localization via projective transformations induces the usual projective transformation on the factor set, so that only the fractionalization class $[\omega_e]$ is well defined. 

In addition to making explicit the general structure of fractionalization classes in the solvable models, this result also makes it simple to calculate $[\omega_e]$.  In particular, we may focus on a single $e$ particle at any desired location, and determine $[\omega_e]$ by calculating appropriate products of $U^e_g(v)$.  In particular, we can calculate the products of generators in Eqs.~(\ref{eqn:projpx}-\ref{eqn:projtypx}), and determine the $\sigma^e$ parameters.  There is then no need to check that the resulting $\sigma^e$'s are the same for every possible location of $e$ particle, because we have already established this in general.

The above discussion proceeds in much the same way for $m$ particles, which reside at holes $h \in H$ in the planar projection ${\mathscr P}({\cal G})$.  States with two $m$ particles can be written
\begin{equation}
| \psi_m(t) \rangle =\cL^m_t | \psi_{0m}\rangle \text{,}
\end{equation}
where $t$ is an open cut.  To every pair $(g, h) \in G \times H$, where the $m$ particle resides at the hole $h$, we associate a number $f^m_g(h)$ and a cut $t^m_g(h)$, which joins $h$ to $g h$.  This allows us to write
\begin{equation}
U^m_g(h) = f^m_g(h) \cL^m_{t^m_g(h)} \text{.}
\end{equation}
From this point, the discussion for $e$ particles goes over to the $m$ particle case, with only trivial modifications.  We refer to a set of $U^m_g(h)$ operators satisfying the $m$ particle analog of Eq.~(\ref{eqn:symfrac}) as a $m$-localization.  Just as in the case of $e$-localizations, Appendix~\ref{app:eloc} establishes that it is always possible to find a $m$-localization, which is unique up to projective transformations.

\section{Symmetry Classes Realized by Toric Code Models}
\label{sec:Toric}

Here, we present the main results of this work, on the realization of symmetry classes in toric code models with square lattice symmetry.  These results consist of explicit construction of models realizing various symmetry classes, as well as the derivation of general constraints showing that certain symmetry classes are impossible in the family of models under consideration.  We have obtained a complete understanding, in the sense that we have found an explicit realization of every symmetry class not ruled out by general constraints.

Below, we present our results in three stages, in order of increasing generality (and decreasing simplicity).  First, we exhibit a single model realizing all possible $e$ particle fractionalization classes $[\omega_e]$, as the parameters of the Hamiltonian are varied.  In this model, the $m$ particles always have trivial fractionalization class.  Second, we consider the family of toric code models with no spin-orbit coupling.  Finally, we consider toric code models allowing for spin-orbit coupling.

\subsection{Model realizing all $e$ particle fractionalization classes}
\label{sec:ep-model}

\begin{figure}
\begin{minipage}[t]{0.8\columnwidth}%
\includegraphics[width=1\columnwidth]{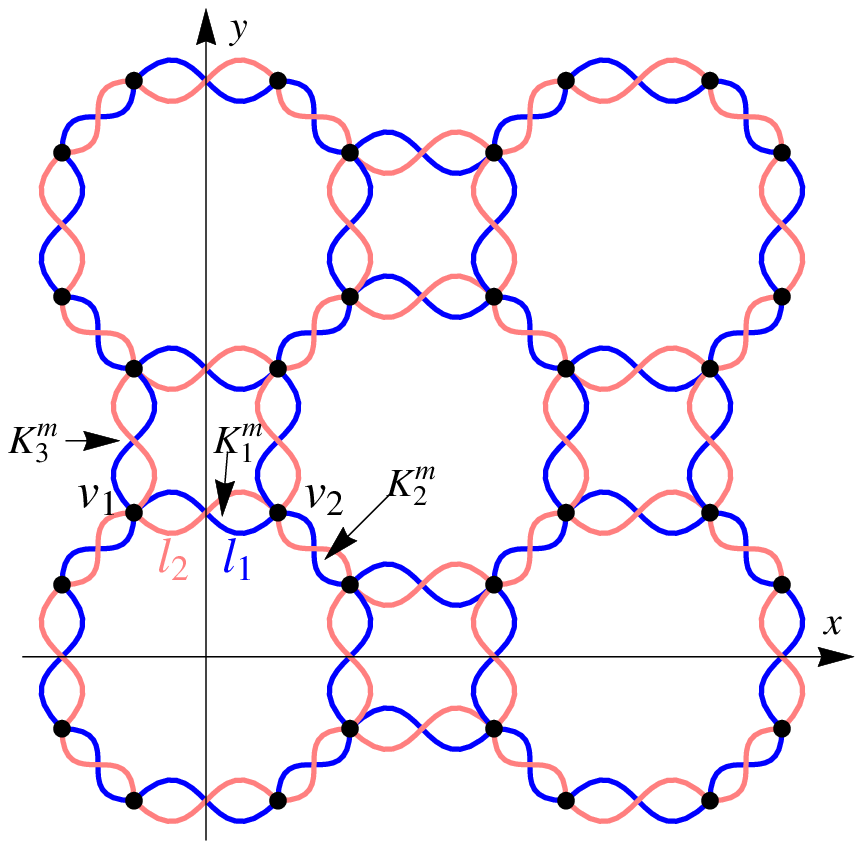}

(a)%
\end{minipage}

\begin{minipage}[t]{0.4\columnwidth}%
\includegraphics[width=1\columnwidth]{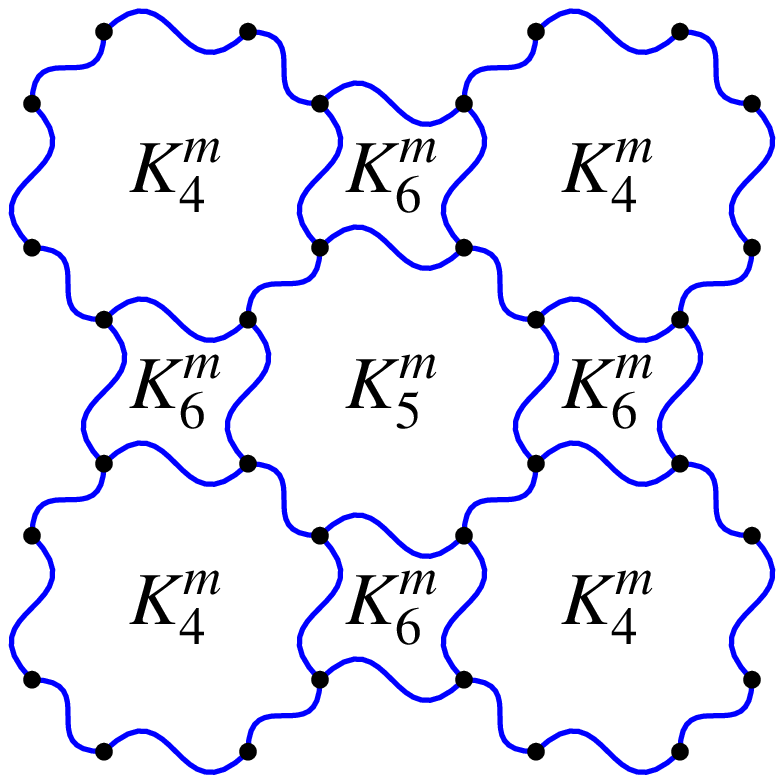}

(b)%
\end{minipage}%
\begin{minipage}[t]{0.4\columnwidth}%
\includegraphics[width=1\columnwidth]{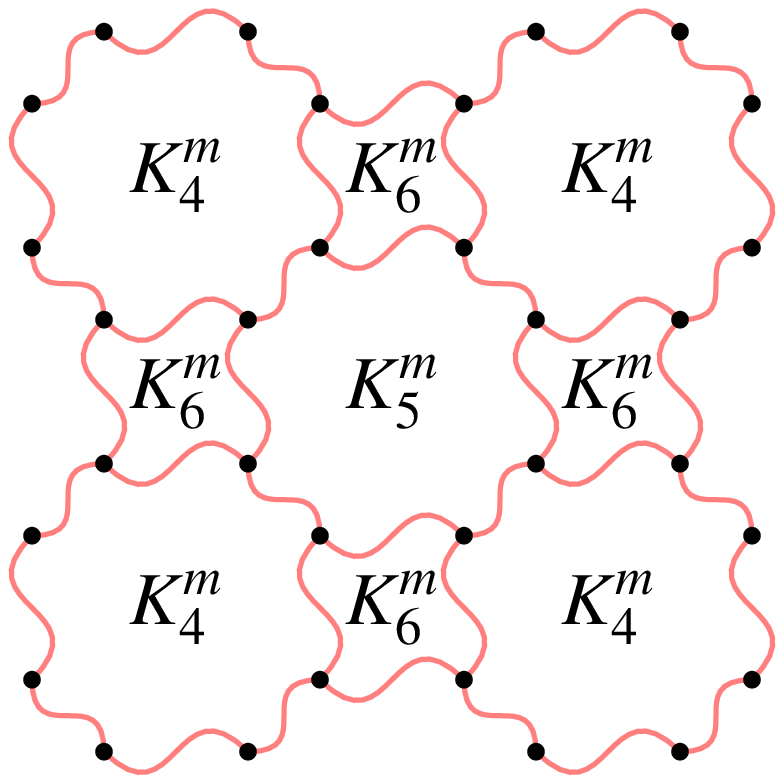}

(c)%
\end{minipage}

\protect\caption{(Color online) (a) The lattice on which all $2^6 = 64$ $e$ particle fractionalization classes can be realized. There are
six types of plaquettes not related by symmetries, and the correponding
plaquette terms are assigned independent coefficients $K_{i}^{m}$
($i=1,2,\cdots,6$).  Nearest-neighbor pairs of vertices are joined by two edges (dark and light; blue and red online), drawn curved to avoid overlapping and to be clear about their movement under space group operations. Plaquetes of type $i = 1,2,3$ are each formed by the two edges joining a nearest-neighbor pair of vertices.  Two vertices $v_1, v_2$ and two edges $l_1, l_2$ are labeled to illustrate the calculation of $\sigma^e_{px}$ discussed in the main text.
(b), (c) Subgraphs of the lattice in (a), each containing all the vertices and half the edges.  These subgraphs transform into one another under any improper space group operation (\emph{i.e.}  reflections).  
We draw these subgraphs to illustrate the plaquettes of type $i = 4,5,6$.}
\label{fig:octagon}
\end{figure}

Here, we present a model that can realize all possible $e$-fractionalization classes $[\omega_e]$, as the parameters of the Hamiltonian are varied.  In this model, the $m$-fractionalization class is always trivial.  The model is defined on the lattice shown in Fig.~\ref{fig:octagon}, and symmetry is chosen to act on the spin degrees of freedom without spin-orbit coupling.
The lattice has six types of plaquettes shown in Fig.~\ref{fig:octagon}, so that only plaquettes of the same type are related by symmetry.  Letting $P_i \subset P$ be the set of all plaquettes of type $i$ ($i = 1,\dots,6$), the Hamiltonian is 
\begin{equation}
{\cal H}=-K^{e}\sum_{v \in V}A_{v}-\sum_{i = 1}^6 K_{i}^{m}\sum_{p_{i} \in P_i}B_{p_{i}}.
\end{equation}

We choose $K^e = 1$, with arbitrary $K^m_i \in \{ \pm 1 \}$, and note that $b_i \equiv K^m_i$ is the ground-state eigenvalue of $B_{p_i}$.  Following the calculation procedure described below, we find
\begin{align}
\sigma^e_{px} &= b_1, \quad & \sigma^e_{pxy} &= b_2,  & \sigma^e_{txpx} &=  b_3,\\
\sigma^e_{pxpxy} &=  b_4, \quad & \sigma^e_{txty} &=  b_4 b_5,  & \sigma^e_{typx} &=b_{1}b_{3}b_{4}b_{6} \text{,}
\end{align}
from which it is clear that each possible $[\omega_e] \in H^2(G, \zz)$ is realized in this model for appropriate choice of $K^m_i$.  In addition we find that all the corresponding $\sigma^m$'s are unity, and thus $[\omega_m]$ is the trivial fractionalization class.

We now illustrate how these results are obtained by working through the determination of $(U^e_{P_x})^2 = \sigma^e_{px}$ as an example.  It follows from the discussion of Sec.~\ref{sec:fc-in-solvable-models} that $\sigma^e_{px}$ can be obtained by considering an $e$ particle at any desired vertex $v_1$, and then computing $(U^e_{P_x})^2$ acting on this $e$ particle.  We consider an $e$ particle at vertex $v_1$ as shown in Fig.~\ref{fig:octagon}a, so that $v_2 = P_x v_1$, and the vertices $v_1$ and $v_2$ are joined by edges $l_1, l_2$ forming a type $i=1$ plaquette.  (To be more precise, we should also specify the position of a second $e$ particle at vertex $v \neq v_1$, let $s_0$ be a path joining $v_1$ to $v$, and consider the state $| \psi_e(s_0) \rangle = \cL^e_{s_0} | \psi_{0 e} \rangle$.  However, the result for $\sigma^e_{px}$ will be independent of $v$.)

We are free to choose the $e$-localization
\begin{eqnarray}
U^e_{P_x}(v_1) &=& \sigma^z_{l_1} \\
U^e_{P_x}(v_2) &=& f \sigma^z_{l_2} \text{,}
\end{eqnarray}
where $f = \pm 1$.  To determine $f$, we consider the path $s = l_1$, which has end points $v_1$ and $v_2$.  Then we have
\begin{equation}
U_{P_x} | \psi_e(s) \rangle = U_{P_x} \sigma^z_{l_1} | \psi_{0 e} \rangle = \sigma^z_{l_2} | \psi_{0 e} \rangle \text{,}
\end{equation}
since $P_x l_1 = l_2$.  But we also have
\begin{eqnarray}
U_{P_x} | \psi_e(s) \rangle &=& U^e_{P_x}(v_1) U^e_{P_x}(v_2)  | \psi_e(s) \rangle \\
&=& ( \sigma^z_{l_1} ) ( f \sigma^z_{l_2} ) \sigma^z_{l_1} | \psi_{0 e} \rangle \\
&=& f \sigma^z_{l_2} | \psi_{0 e} \rangle \text{.}
\end{eqnarray}
Consistency of these two calculations of the action of $U_{P_x}$ then requires $f = 1$.  

Now that we have fixed the form of the $e$-localization, we can compute the action of $P_x^2$ on the $e$ particle at $v_1$.  We have
\begin{eqnarray}
\sigma^e_{px} &=& (U^e_{P_x})^2 (v_1) = U^e_{P_x}(v_2) U^e_{P_x}(v_1) \\
&=& \sigma^z_{l_2} \sigma^z_{l_1} = K^m_1 = b_1 \text{.}
\end{eqnarray}
This should be interpreted as an operator equation that hold acting on any state obtained by acting successively with $e$-string operators on $|\psi_{0 e}\rangle$.  In particular it holds acting on a state of interest, $|\psi_e(s)\rangle$, with one $e$ particle located at $v_1$.  The results for the other $\sigma^e$ parameters can be obtained by straightforward analogous calculations.

\subsection{Toric code models without spin-orbit coupling}
\label{sub:tc_woso}

\begin{figure*}
\begin{minipage}[t]{0.33\textwidth}%
\includegraphics[width=0.9\columnwidth]{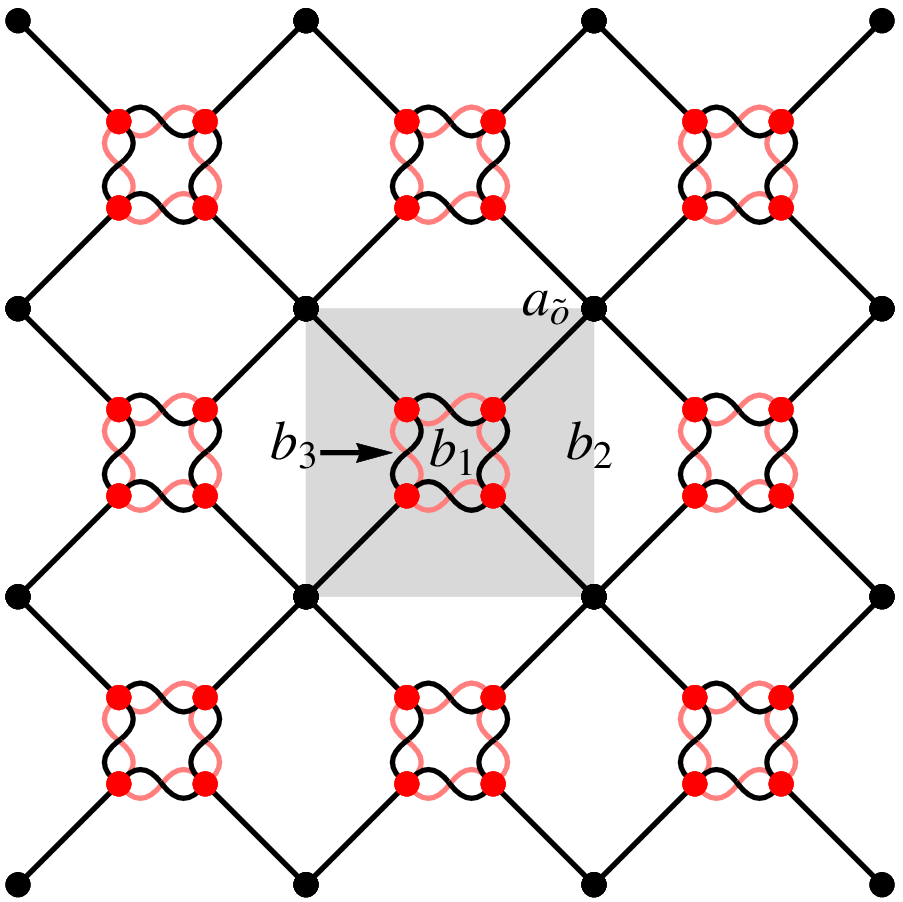}

(a) $\left(\begin{array}{cccccc}
b_{3} & 1 & 1 & b_{1} & 1 & b_{2}b_{3}\\
1 & 1 & 1 & 1 & a_{\tilde{o}} & 1
\end{array}\right)$%
\end{minipage}%
\begin{minipage}[t]{0.33\textwidth}%
\includegraphics[width=0.9\columnwidth]{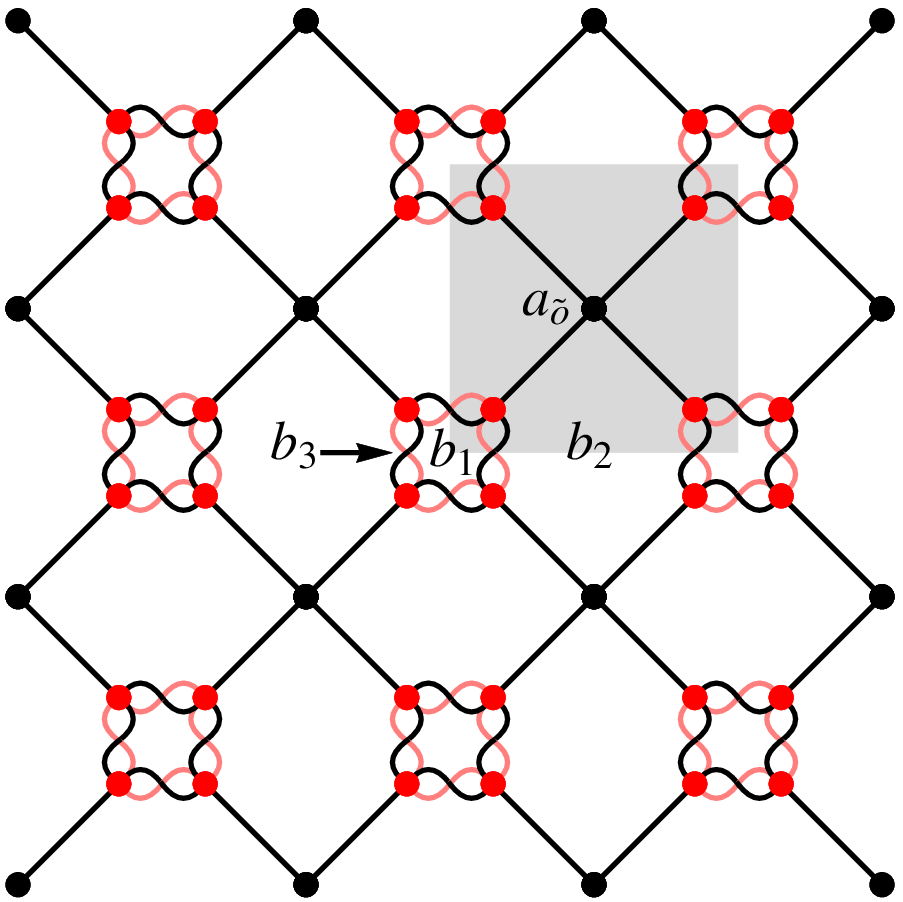}

(b) $\left(\begin{array}{cccccc}
1 & 1 & b_{3} & 1 & b_{1} & b_{2}b_{3}\\
1 & 1 & 1 & a_{\tilde{o}} & 1 & 1
\end{array}\right)$%
\end{minipage}%
\begin{minipage}[t]{0.33\textwidth}%
\includegraphics[width=0.9\columnwidth]{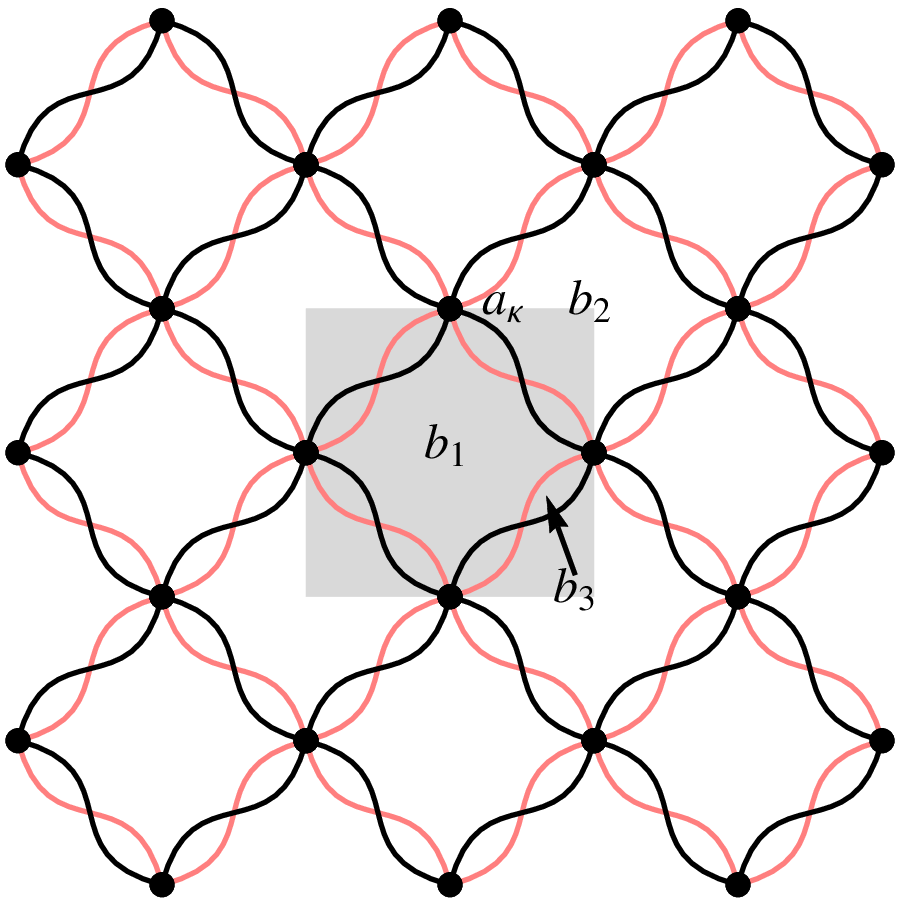}

(c) $\left(\begin{array}{cccccc}
1 & b_{3} & 1 & b_{1} & b_{2} & 1\\
1 & 1 & 1 & 1 & 1 & a_{\kappa}
\end{array}\right)$%
\end{minipage}

\bigskip{}

\begin{minipage}[t]{0.33\textwidth}%
\includegraphics[width=0.9\columnwidth]{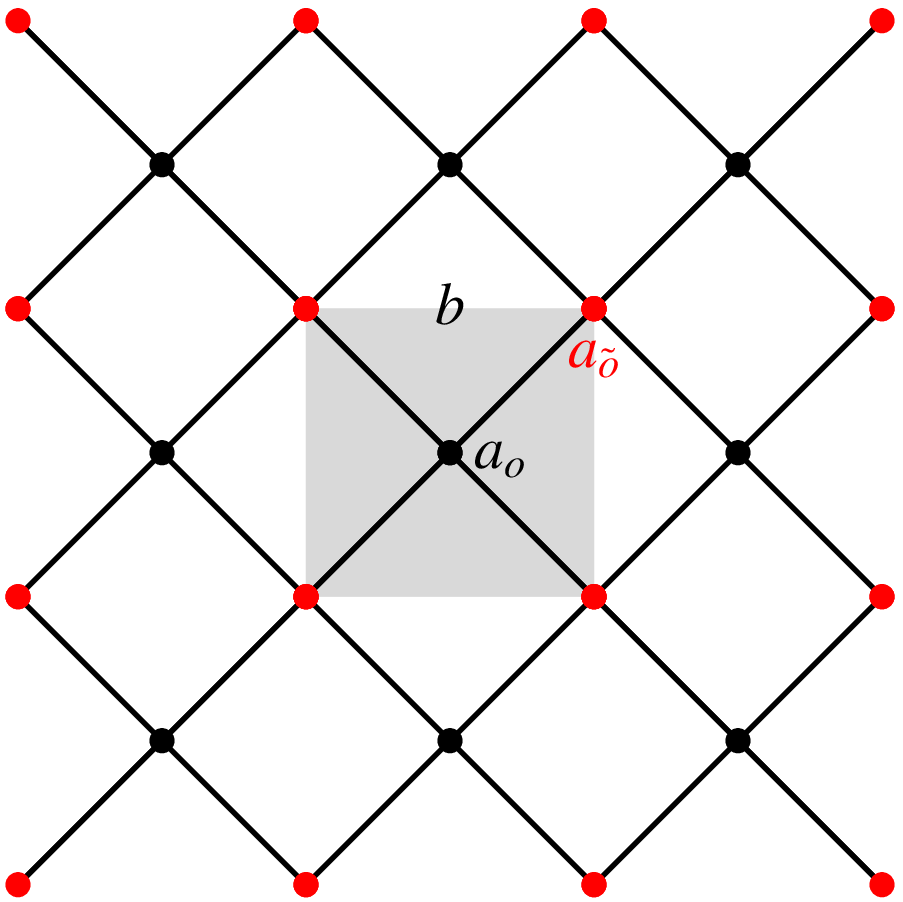}

(d) $\left(\begin{array}{cccccc}
1 & 1 & 1 & 1 & 1 & b\\
1 & 1 & 1 & a_{o} & a_{\tilde{o}} & 1
\end{array}\right)$%
\end{minipage}%
\begin{minipage}[t]{0.33\textwidth}%
\includegraphics[width=0.9\columnwidth]{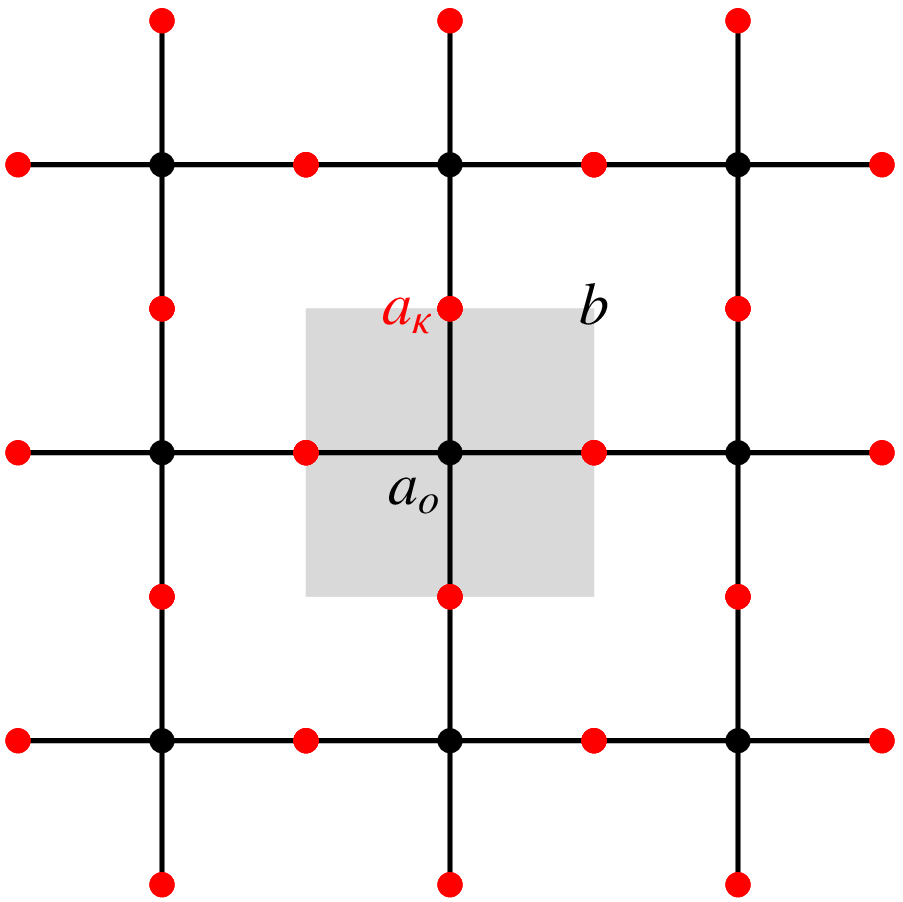}

(e) $\left(\begin{array}{cccccc}
1 & 1 & 1 & 1 & b & 1\\
1 & 1 & 1 & a_{o} & 1 & a_{\kappa}
\end{array}\right)$%
\end{minipage}%
\begin{minipage}[t]{0.33\textwidth}%
\includegraphics[width=0.9\columnwidth]{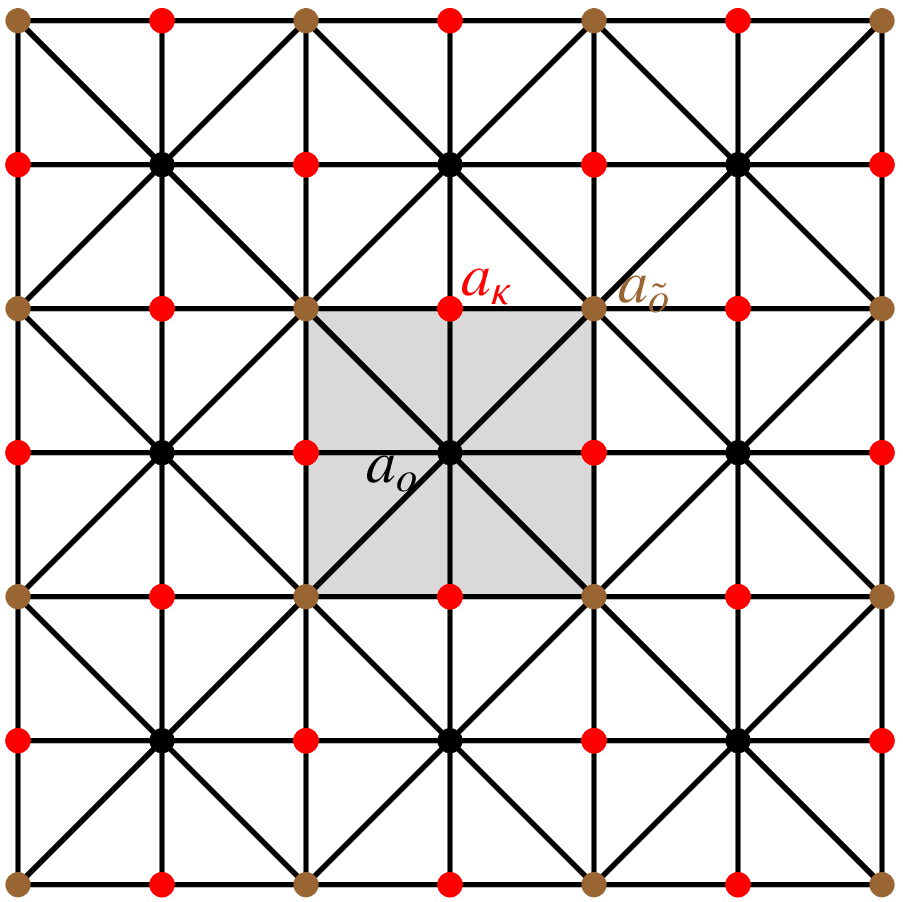}

(f) $\left(\begin{array}{cccccc}
1 & 1 & 1 & 1 & 1 & 1\\
1 & 1 & 1 & a_{o} & a_{\tilde{o}} & a_{\kappa}
\end{array}\right)$%
\end{minipage}

\caption{(Color online) $TC_{0}\left(G\right)$ models. The shaded square is a unit cell and
the origin of our coordinate system is at the center of the square.
Below each figure of lattice is the corresponding TC symmetry class
in the form \eqref{eq:sc_matrix}. Here $a_{r}$ is the ground state
eigenvalue of $A_{v}$ for $v$ at special points $r=o,\tilde{o},\kappa$;
and $b$, $b_1$, $b_2$ are the ground state eigenvalues
of $B_{p}$ for the plaquette $p$, which in these models is picked
to be the smallest cycle made with black edges where $b$, $b_1$ or $b_2$ is written,
while $b_{3}$ is for the plaquette made of a pair of black and grey edges (black and pink online). 
These edges are drawn curved to avoid overlapping and to be clear about their movement under space group operations.
The comparison between (a) and (b) gives 
an explicit example that moving the coordinate system origin by
$\left(\frac{1}{2},\frac{1}{2}\right)$ results in a transformation \eqref{eq:move_origin}: 
$P_{x}\rightarrow T_{x}P_{x}$, $\sigma_{px}\leftrightarrow\sigma_{txpx}$,
$\sigma_{pxpxy}\leftrightarrow\sigma_{pxpxy}\sigma_{txty}$. The symmetry class differs 
from (e) by such a transformation can be easily got by moving the coordinate system, 
so we do not bother drawing a separate lattice for it.}

\label{fig:tc0}
\end{figure*}

We now proceed to consider the family of models $TC_0(G)$, which includes all toric code models with square lattice space group symmetry as introduced in Sec.~\ref{sec:genlatt}, with the restriction of no spin-orbit coupling.  We remind the reader that this means, for any symmetry operation $g \in G$, we have $U_g \sigma^{\mu}_{\ell} U^{-1}_g = \sigma^{\mu}_{g \ell}$.  In words, symmetry acts simply by moving edges and vertices of the lattice, and acts trivially within the Hilbert space of each spin.

In Appendix~\ref{app:nosoc-constraints}, we obtain a number of constraints on which symmetry classes can occur for models in $TC_0(G)$.  The main result is the following theorem:
\begin{thmtext} \label{thm:nosoc_maintext}
The TC symmetry classes in $\mathsf{A}$, $\mathsf{B}$, $\mathsf{C}$,
$\mathsf{M}$, $\mathsf{M_{1}}$, $\mathsf{M_{2}}$ and $\mathsf{M_{3}}$
are not realizable in $TC_0(G)$, where
\begin{eqnarray*}
\mathsf{A} & = & \left\{ \sigma_{pxpxy}^{e}=\sigma_{pxpxy}^{m}=-1\right\} ,\\
\mathsf{B} & = & \left\{ \sigma_{pxpxy}^{e}\sigma_{txty}^{e}=\sigma_{pxpxy}^{m}\sigma_{txty}^{m}=-1\right\} ,\\
\mathsf{C} & = & \left\{ \sigma_{pxpxy}^{e}\sigma_{typx}^{e}=\sigma_{pxpxy}^{m}\sigma_{typx}^{m}=-1\right\} ,\\
\mathsf{M} & = & \left\{ \sigma_{px}^{m}=-1\vee\sigma_{pxy}^{m}=-1\vee\sigma_{txpx}^{m}=-1\right\} ,\\
\mathsf{M_{1}} & = & \left\{ \sigma_{pxpxy}^{m}=-1\wedge\left(\sigma_{px}^{e}=-1\vee\sigma_{pxy}^{e}=-1\right)\right\} ,\\
\mathsf{M_{2}} & = & \left\{ \sigma_{pxpxy}^{m}\sigma_{txty}^{m}=-1\wedge\left(\sigma_{pxy}^{e}=-1\vee\sigma_{txpx}^{e}=-1\right)\right\} ,\\
\mathsf{M_{3}} & = & \left\{ \sigma_{pxpxy}^{m}\sigma_{typx}^{m}=-1\wedge\left(\sigma_{px}^{e}=-1\vee\sigma_{txpx}^{e}=-1\right)\right\} .
\end{eqnarray*}
Here $\wedge$, $\vee$ are the logical symbols for ``and'' and
``or'' respectively.  

This leaves 95 TC symmetry classes not ruled out by the above constraints, corresponding to 82 symmetry classes under $e \leftrightarrow m$ relabeling.  In addition, all these 95 TC symmetry classes are realized by models in $TC_0(G)$. \end{thmtext}

This theorem is proved in Appendix~\ref{app:nosoc-constraints}, except for the last statements regarding counting and realization of symmetry classes, which are proved here.  In fact, we exhibit a model realizing each allowed TC symmetry class.  Before proceeding to do this, we would like to give a flavor for how the above constraints are obtained, referring the reader to Appendix~\ref{app:nosoc-constraints} for the full details.  

As an illustration, we would like to show that $\sigma^m_{px} = 1$ for any model in $TC_0(G)$.  (This is part of the fact that TC symmetry classes in $\mathsf{M}$ are not realizable in $TC_0(G)$.)  Consider a $m$ particle located at a hole $h_0 \in H$.  If $P_x h_0 = h_0$, then we can choose $U^m_{P_x}(h_0) = 1$, and therefore $\sigma^m_{px} = (U^m_{P_x})^2 (h_0) = 1$.

We then consider the case $P_x h_0 = h_1 \neq h_0$.  We can always draw a simple cut $t$ joining $h_0$ to $h_1$, so that $P_x t = t$.  We are then free to choose the $m$-localization
\begin{eqnarray}
U^m_{P_x}(h_0) &=& \cL^m_t \\
U^m_{P_x}(h_1) &=& f \cL^m_t \text{,}
\end{eqnarray}
where $f = \pm 1$ needs to be determined.  To do this, consider the state $|\psi_m(t)\rangle = \cL^m_t | \psi_{0 m} \rangle$, for which we have
\begin{equation}
U_{P_x} | \psi_m(t) \rangle = U_{P_x} \cL^m_t | \psi_{0 m} \rangle = | \psi_m(t) \rangle  \text{,}
\end{equation}
where we used the fact that $U_{P_x} \cL^m_t U^{-1}_{P_x} = \cL^m_{P_x t} = \cL^m_t$. (Note that here we use the assumption of no spin-orbit coupling.)
But we also have
\begin{equation}
U_{P_x} | \psi_m(t) \rangle = U^m_{P_x}(h_0) U^m_{P_x}(h_1) | \psi_m(t) \rangle = f | \psi_m(t) \rangle \text{.}
\end{equation}
Consistency of these two calculations requires $f=1$, and we can then calculate $P_x^2$ acting on the $m$ particle located at $h_0$, to obtain
\begin{eqnarray}
\sigma^m_{px} &=& (U^m_{P_x})^2 (h_0) =  \\
&=& U^m_{P_x}(h_1) U^m_{P_x}(h_0)  = ( \cL^m_t)^2  = 1 \text{.}
\end{eqnarray}
We have thus shown that $\sigma^m_{px} = 1$ for any model in $TC_0(G)$.  Roughly similar reasoning is followed in Appendix~\ref{app:nosoc-constraints} to establish the constraints stated in the theorem.

Now we proceed to enumerate and  count the TC symmetry classes not ruled out by the constraints of Theorem~\ref{thm:nosoc_maintext}.  At the same time, we present the explicit models realizing each class (shown in Figures~\ref{fig:octagon} and~\ref{fig:tc0}).  Here, and throughout the paper, we will find it convenient to present TC symmetry classes $\left(\left[\omega_{e}\right],\left[\omega_{m}\right]\right)$ in the matrix form
\begin{equation}
\left(\begin{array}{cccccc}
\sigma_{px}^{e} & \sigma_{pxy}^{e} & \sigma_{txpx}^{e} & \sigma_{pxpxy}^{e} & \left(\sigma_{pxpxy}^{e}\sigma_{txty}^{e}\right) & \left(\sigma_{pxpxy}^{e}\sigma_{typx}^{e}\right)\\
\sigma_{px}^{m} & \sigma_{pxy}^{m} & \sigma_{txpx}^{m} & \sigma_{pxpxy}^{m} & \left(\sigma_{pxpxy}^{m}\sigma_{txty}^{m}\right) & \left(\sigma_{pxpxy}^{m}\sigma_{typx}^{m}\right) 
\end{array}\right) \label{eq:sc_matrix} \text{,}
\end{equation}
or, equivalently,
\begin{equation}
\left(\begin{array}{cc}
\sigma_{px}^{e} & \sigma_{px}^{m}\\
\sigma_{pxy}^{e} & \sigma_{pxy}^{m}\\
\sigma_{txpx}^{e} & \sigma_{txpx}^{m}\\
\sigma_{pxpxy}^{e} & \sigma_{pxpxy}^{m}\\
\sigma_{pxpxy}^{e}\sigma_{txty}^{e} & \sigma_{pxpxy}^{m}\sigma_{txty}^{m}\\
\sigma_{pxpxy}^{e}\sigma_{typx}^{e} & \sigma_{pxpxy}^{m}\sigma_{typx}^{m}
\end{array}\right) \text{.}
\end{equation}
This form allows for simple comparison to the  constraints of Theorem~\ref{thm:nosoc_maintext}.  In addition, 
under the change of origin $o \to \left(\frac{1}{2},\frac{1}{2}\right)$, the entries of the matrix are simply permuted:
\begin{align}
 & \left(\begin{array}{cccccc}
\sigma_{1}^{e} & \sigma_{2}^{e} & \sigma_{3}^{e} & \sigma_{4}^{e} & \sigma_{5}^{e} & \sigma_{6}^{e}\\
\sigma_{1}^{m} & \sigma_{2}^{m} & \sigma_{3}^{m} & \sigma_{4}^{m} & \sigma_{5}^{m} & \sigma_{6}^{m}
\end{array}\right)\nonumber \\
\rightarrow & \left(\begin{array}{cccccc}
\sigma_{3}^{e} & \sigma_{2}^{e} & \sigma_{1}^{e} & \sigma_{5}^{e} & \sigma_{4}^{e} & \sigma_{6}^{e}\\
\sigma_{3}^{m} & \sigma_{2}^{m} & \sigma_{1}^{m} & \sigma_{5}^{m} & \sigma_{4}^{m} & \sigma_{6}^{m}
\end{array}\right). \label{eq:move_origin}
\end{align}
This holds even beyond the setting of solvable toric code models, and can be verified by replacing $P_x$ as a generator of $G$ by $P_x \to \widetilde{P_x} = T_x P_x$, which corresponds to the desired change of origin.  The $\sigma$ parameters for the new generators can then be computed in terms of those for the old generators, by noting that $U^a_{\widetilde{P_x}} = \phi^a U^a_{T_x} U^a_{P_x}$, where $a = e,m$ and $\phi^a \in \{ \pm 1 \}$.

The behavior of TC symmetry classes under a change in origin is illustrated in Fig.~\ref{fig:tc0}a and Fig.~\ref{fig:tc0}b.
Apart from this example, we do not bother to draw the same lattice twice when the only difference is a change in origin.  So, for example, the model shown in Fig.~\ref{fig:tc0}e is taken to realized both TC symmetry classes
\begin{equation}
\left(\begin{array}{cccccc}
1 & 1 & 1 & 1 & b & 1\\
1 & 1 & 1 & a_{o} & 1 & a_{\kappa}
\end{array}\right)\label{eq:symb-1}
\end{equation}
and 
\begin{equation}
\left(\begin{array}{cccccc}
1 & 1 & 1 & b & 1 & 1\\
1 & 1 & 1 & 1 & a_{o} & a_{\kappa}
\end{array}\right),\label{eq:symb-2}
\end{equation}
where the TC symmetry classes \eqref{eq:symb-1} are realized if we put the
origin at the center of the shaded square, and the TC symmetry classes
\eqref{eq:symb-2} are realized if we put the origin at the corner of
the shaded square. 

Now, we divide the TC symmetry classes not ruled out by Theorem~\ref{thm:nosoc_maintext} into four collections $\mathsf{D}_{i}$, $i=0,1,2,3$.
In $\mathsf{D}_{i}$, there are $i$ of $\sigma_{pxpxy}^{m}$, $\sigma_{txty}^{m}$
and $\sigma_{typx}^{m}$ equal to $-1$. In $\mathsf{D}_{0}$,
we have TC symmetry classes
in the form 
\[
\left(\begin{array}{cc}
\square & 1\\
\square & 1\\
\square & 1\\
\square & 1\\
\square & 1\\
\square & 1
\end{array}\right),
\]
where the symbol $\square$ means that the corresponding $\sigma$ parameter can be chosen to be $\pm 1$ independently of any other parameters.  Therefore, $\left|\mathsf{D}_{0}\right|=2^{6}$.  These TC symmetry classes are realized in the model discussed in Sec.~\ref{sec:ep-model}, and shown in Fig.~\ref{fig:octagon}.

In $\mathsf{D}_{1}$, we have TC symmetry classes $\left(\left[\omega_{e}\right],\left[\omega_{m}\right]\right)$
in the form 
\[
\left(\begin{array}{cc}
1 & 1\\
1 & 1\\
\square & 1\\
1 & -1\\
\square & 1\\
\square & 1
\end{array}\right),\left(\begin{array}{cc}
\square & 1\\
1 & 1\\
1 & 1\\
\square & 1\\
1 & -1\\
\square & 1
\end{array}\right),\text{ or }\left(\begin{array}{cc}
1 & 1\\
\square & 1\\
1 & 1\\
\square & 1\\
\square & 1\\
1 & -1
\end{array}\right),
\]
so $\left|\mathsf{D}_{1}\right|=3\times2^{3}$.  These TC symmetry classes are realized
in the models shown in Fig.~\ref{fig:tc0}(a-c).

In $\mathsf{D}_{2}$, we have TC symmetry classes
in the form 
\[
\left(\begin{array}{cc}
1 & 1\\
1 & 1\\
1 & 1\\
1 & -1\\
1 & -1\\
\square & 1
\end{array}\right),\left(\begin{array}{cc}
1 & 1\\
1 & 1\\
1 & 1\\
\square & 1\\
1 & -1\\
1 & -1
\end{array}\right),\text{ or }\left(\begin{array}{cc}
1 & 1\\
1 & 1\\
1 & 1\\
1 & -1\\
\square & 1\\
1 & -1
\end{array}\right),
\]
so $\left|\mathsf{D}_{2}\right|=3\times2$. These TC symmetry classes are realized in
Fig.~\ref{fig:tc0}(d,e).

In $\mathsf{D}_{3}$, we have only the single TC symmetry class
\begin{equation}
\left(\begin{array}{cc}
1 & 1 \\
1 & 1 \\
1 & 1 \\
1 & -1 \\
1 & -1 \\
1 & -1 \end{array}\right) \text{,}
\end{equation}
which is realized by the model of Fig.~\ref{fig:tc0}f.

In total, there are thus exactly $\sum_{i=0}^{3}\left|\mathsf{D}_{i}\right|=95$ TC symmetry
classes realized by models in $TC_{0}\left(G\right)$.  Recalling that the TC symmetry classes 
$\left(\left[\omega_{m}\right],\left[\omega_{e}\right]\right)$
and $\left(\left[\omega_{e}\right],\left[\omega_{m}\right]\right)$ correspond to the same symmetry class,
it is a straightforward but somewhat tedious exercise to show that 13 symmetry classes are double-counted among the 95 TC symmetry classes.  Therefore, the total number of symmetry classes realized by models in $TC_0(G)$ is $95-13=82$.

\subsection{General toric code models}
\label{sub:tc_so}

To consider the most general toric code models introduced in Sec.~\ref{sec:genlatt}, we must allow for spin-orbit coupling.  As discussed in Sec.~\ref{sec:genlatt}, this means, for any symmetry operation $g \in G$, we have $U_g \sigma^{\mu}_{\ell} U^{-1}_g = c^{\mu}_{\ell}(g) \sigma^{\mu}_{g \ell}$, where $c^{\mu}_{\ell}(g) \in \{ \pm 1 \}$, 
$\mu = x, z$.  The corresponding family of models is referred to as $TC(G)$.  Our results on these models are summarized in the following theorem:

\begin{thmtext} \label{thm:soc_maintext}
The TC symmetry classes in $\mathsf{P_{1}}$, $\mathsf{P_{2}}$,
$\mathsf{P_{3}}$, $\mathsf{A}$, $\mathsf{B}$ and $\mathsf{C}'$
are not realizable in $TC\left(G\right)$, where
\begin{eqnarray*}
\mathsf{P_{1}} & = & \left\{ \sigma_{px}^{e}=\sigma_{px}^{m}=-1\right\} ,\\
\mathsf{P_{2}} & = & \left\{ \sigma_{pxy}^{e}=\sigma_{pxy}^{m}=-1\right\} ,\\
\mathsf{P_{3}} & = & \left\{ \sigma_{txpx}^{e}=\sigma_{txpx}^{m}=-1\right\} ,\\
\mathsf{A} & = & \left\{ \sigma_{pxpxy}^{e}=\sigma_{pxpxy}^{m}=-1\right\} ,\\
\mathsf{B} & = & \left\{ \sigma_{pxpxy}^{e}\sigma_{txty}^{e} = \sigma_{pxpxy}^{m}\sigma_{txty}^{m} = -1\right\} ,\\
\mathsf{C'} & = & \left\{ \sigma_{px}^{e}=\sigma_{txpx}^{e}=\sigma_{pxpxy}^{e}\sigma_{typx}^{e}=\sigma_{pxpxy}^{m}\sigma_{typx}^{m}=-1\right\} .
\end{eqnarray*}

This leaves 945 TC symmetry classes not ruled out by the above constraints, corresponding to 487 symmetry classes under $e \leftrightarrow m$ relabeling.  In addition, all these 945 TC symmetry classes are realized by models in $TC(G)$.
\end{thmtext}

This theorem is proved in Appendices~\ref{app:tcsoc} and~\ref{app:models}.  The constraints ruling out some TC symmetry classes are obtained in Appendix~\ref{app:tcsoc}, while the counting of symmetry classes and the presentation of explicit models is done in Appendix~\ref{app:models}.  

\begin{figure*}
\begin{minipage}[t]{0.33\textwidth}%
\includegraphics[width=0.9\columnwidth]{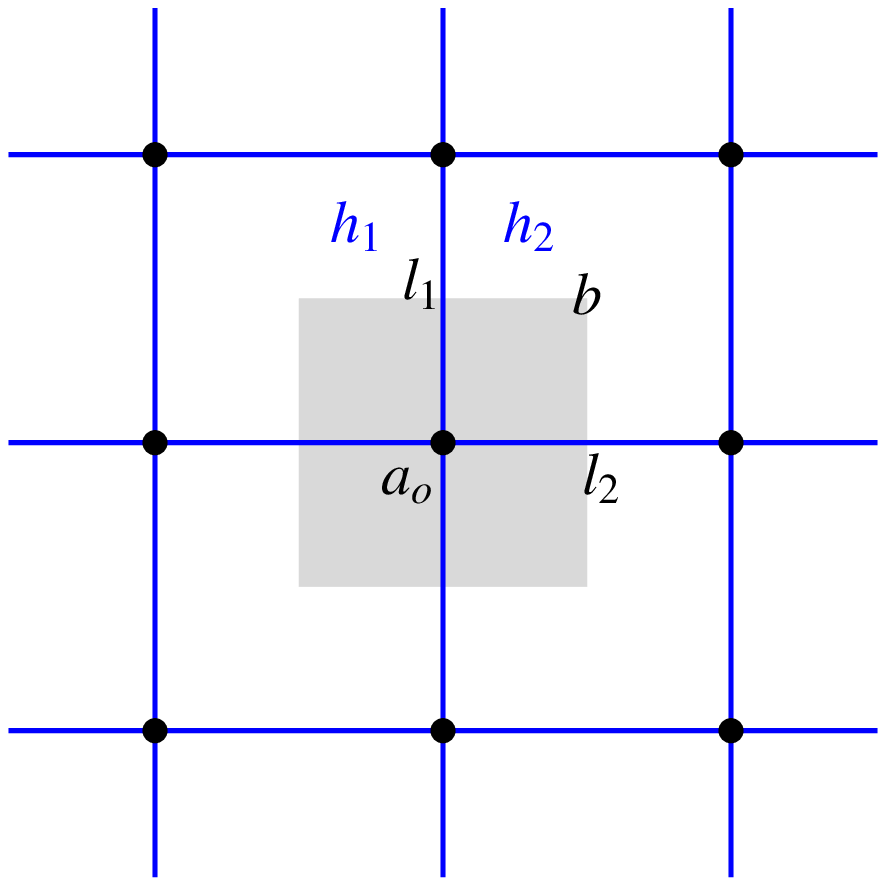}

(a) $\left(\begin{array}{cccccc}
1 & 1 & \gamma_{2} & 1 & b & \gamma_{1}\\
\alpha_{1} & 1 & 1 & a_{o} & 1 & \alpha_{2}
\end{array}\right)$%
\end{minipage}%
\begin{minipage}[t]{0.33\textwidth}%
\includegraphics[width=0.9\columnwidth]{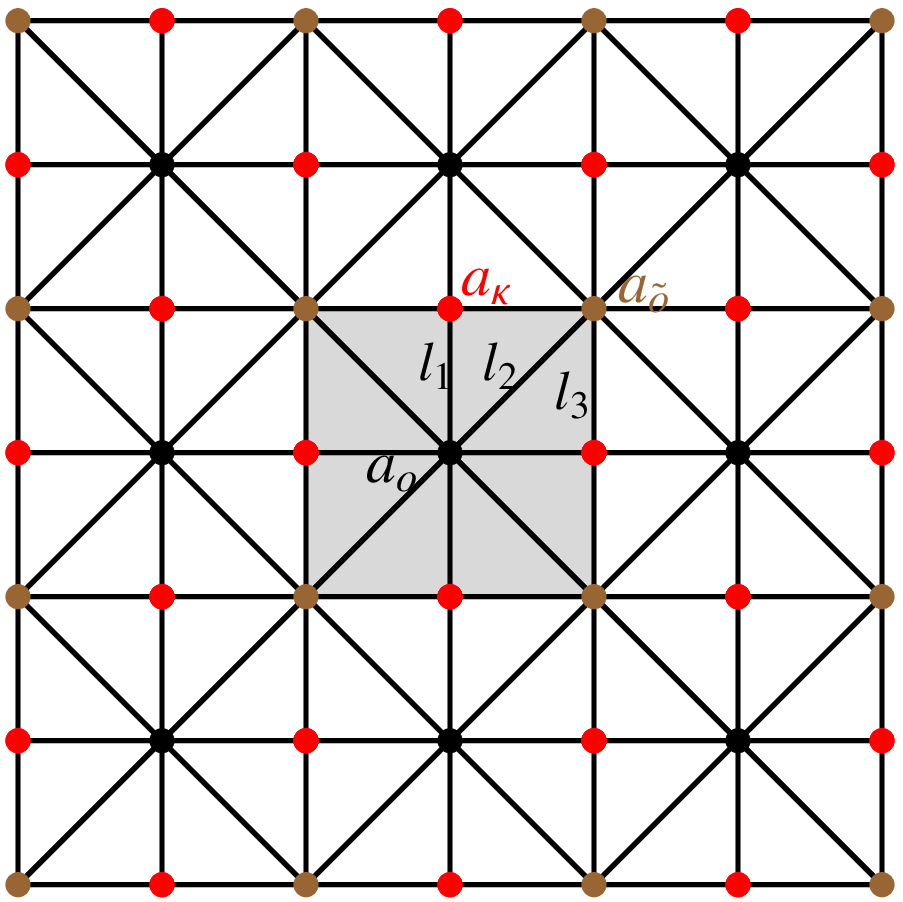}

(b) $\left(\begin{array}{cccccc}
1 & 1 & 1 & 1 & 1 & 1\\
c_{1} & c_{2} & c_{3} & a_{o} & a_{\tilde{o}} & c_{1}c_{3}a_{\kappa}
\end{array}\right)$%
\end{minipage}

\caption{(Color online) Two example models in $TC\left(G\right)$ that realize TC symmetry classes
not possible in $TC_{0}\left(G\right)$. The shaded square is a unit
cell and the origin of our coordinate system is at the center of the
square. Below each figure of lattice is the corresponding TC symmetry
class in the form \eqref{eq:sc_matrix}. Here $a_{r}$ is the ground
state eigenvalue of $A_{v}$ for $v$ at special points $r=o,\tilde{o},\kappa$
and $b$ is the ground state eigenvalue of $B_{p}$ for the plaquette
$p$, defined here to be the smallest cycle enclosing the letter ``$b$.''  We write $\alpha_{i}=c_{l_{i}}^{x}\left(P_{x}\right)$, $\beta_{i}=c_{l_{i}}^{x}\left(P_{xy}\right)$,
$\gamma_{i}=c_{l_{i}}^{z}\left(P_{x}\right)$ and $\delta_{i}=c_{l_{i}}^{z}\left(P_{xy}\right)$.
(a) A model realizing some TC symmetry classes (and symmetry classes) that cannot be realized without spin-orbit coupling.  Here ,$h_{1}$, $h_{2}$ label two positions of a $m$ particle for the calculation of $\sigma_{px}^{m}=\alpha_{1}$ in the main text. (b) A model realizing all $2^6 = 64$ possible $m$ particle fractionalization classes $\left[\omega_{m}\right]$. Here, for simplicity, we make the restriction $\gamma_i = \delta_i \equiv c_i$.}
\label{fig:tcg_maintext}
\end{figure*}

Here, we simply give an illustration how spin-orbit coupling increases the number of allowed symmetry classes.
For the model shown in Fig.~\ref{fig:tcg_maintext}a, more TC symmetry classes are possible if spin-orbit coupling is included.  
For example, take the calculation of $\sigma_{px}^{m}$. 
Suppose $U_{P_{x}}\sigma_{l_{1}}^{x}U_{P_{x}}=\alpha_{1}\sigma_{l_{1}}^{x}$, with $\alpha_1 \in \{ \pm 1 \}$. 
If we choose $U^m_{P_{x}}\left(h_{1}\right)=\sigma_{l_{1}}^{x}$, then we must have $U^m_{P_{x}}\left(h_{2}\right)=\alpha_{1}\sigma_{l_{1}}^{x}$
to ensure $U^m_{P_{x}}\left(h_{1}\right)U^m_{P_{x}}\left(h_{2}\right)\sigma_{l_{1}}^{x}\left|\psi_{m0}\right\rangle =U_{P_{x}}\sigma_{l_{1}}^{x}\left|\psi_{m0}\right\rangle $.  Therefore we have $\sigma^m_{px} = (U^m_{P_x} )^2 ( h_1) = U^m_{P_x} (h_2) U^m_{P_x} (h_1) = \alpha_1$.  Therefore we can have $\sigma^m_{px} = -1$, which is impossible without spin-orbit coupling.

Another particularly interesting example, shown in Fig.~\ref{fig:tcg_maintext}b, is a model realizing all $2^6 = 64$ $m$ particle fractionalization classes.  This model is constructed starting with the lattice of Fig.~\ref{fig:tc0}f and allowing for spin-orbit coupling.

\section{Summary and Beyond Toric Code Models}
\label{sec:summary}

To summarize, we considered the realization of distinct square lattice space group
symmetry fractionalizations in exactly solvable $\zz$ toric code models.  We obtained a complete
understanding, in the sense that every symmetry class consistent
with the fusion rules is either realized in an explicit model, or is proved rigorously to be unrealizable.  In more detail, first, we found a single model that realizes all $2^6 = 64$ $e$ particle fractionalization classes as the parameters in its Hamiltonian are varied.  Second, we considered a restricted family of models $TC_0(G)$ without spin-orbit coupling, but defined on general two-dimensional lattices.  We showed that exactly 95 TC symmetry classes $\left(\left[\omega_{e}\right],\left[\omega_{m}\right]\right)$, corresponding to 82 symmetry classes $\left\langle \left[\omega_{e}\right],\left[\omega_{m}\right]\right\rangle$, are realized by models in $TC_0(G)$.  This result was established by proving that the other TC symmetry classes cannot be realized by any model in $TC_0(G)$, and giving explicit models for  those classes not ruled out by such general arguments.  Finally, in the most general family of models considered, $TC(G)$, we allowed spin-orbit coupling in the action of symmetry.  In this case we found that exactly 945 TC symmetry classes, corresponding to 487 symmetry classes, are realized in $TC(G)$.

These main results are, of course,  confined to a special family of exactly solvable models.  Because the symmetry class is a robust characteristic of a SET phase, and thus stable to small perturbations preserving the symmetry,\cite{essin13} all the symmetry classes that we find clearly exist in more generic models.  However, there may well be symmetry classes not realized in $TC(G)$ that can occur in more generic models (this is indeed the case, as we see below).

Ideally, we would like to make statements about arbitrary local bosonic models (\emph{i.e.} those with finite-range interactions).  For example, we can ask the challenging question of which symmetry classes can be realized in the family of all local bosonic models with square lattice space group symmetry.  We do not have an answer to this question, but here we provide some partial answers.  First, we show using a parton gauge theory construction that there exist symmetry classes not realizable in $TC(G)$ that can be realized in local bosonic models.   Second, we establish a connection between symmetry classes of certain on-site symmetry groups and symmetry classes of the square lattice space group.

Our parton construction allows us to argue that if $[\omega_m]$ is a $m$ fractionalization class realized for a model in $TC_0(G)$, then the symmetry class $\langle [\omega_e], [\omega_m] \rangle$, where $[\omega_e] \in H^2(G, \zz)$ is arbitrary, can be realized in a local bosonic model.  It is easy to see that some symmetry classes obtained this way cannot be realized in $TC(G)$. For example, the symmetry classes in $\mathsf{A}$ are unrealizable in $TC(G)$ (Theorem~\ref{thm:soc_maintext}), but they are possible here.

The starting point for the construction is a Hamiltonian of the form
\begin{equation}
{\cal H} = - \sum_{v \in V} K^e_v A_v - \sum_{p \in P} B_p \text{,}
\end{equation}
where $K^e_v \in \{ \pm 1 \}$.  We take the symmetry to act without spin-orbit coupling, so this is a model in $TC_0(G)$.
We have chosen $K^m_p = 1$ for all $p \in P$, which implies the $e$ fractionalization class is trivial.  However, Hamiltonians of this form can realize any $m$ fractionalization class allowed in $TC_0(G)$, because without spin-orbit coupling the $m$ fractionalization class only depends on the lattice and on the $K^e_v$ coefficients.

We now build a $\zz$ gauge theory based on the above toric code model.  On each vertex $v$ we introduce a boson field created by $b^\dagger_{v \alpha}$, where $\alpha = 1,\dots, n$ is an internal index.  We also introduce the gauge constraint
\begin{equation}
A_v = K^e_v (-1)^{ b^\dagger_{v \alpha} b^{\vphantom\dagger}_{v \alpha} } \text{,}
\end{equation}
with sums over repeated internal indices implied.  The gauge theory Hamiltonian is taken to be
\begin{equation}
{\cal H}_{{\rm gauge}} = - \sum_{p \in P} B_p + u \sum_{v \in V} b^\dagger_{v \alpha} b^{\vphantom\dagger}_{v \alpha}
- h \sum_{\ell \in E} \sigma^x_{\ell} \text{,}
\end{equation}
with $u > 0$.
We choose symmetry to act on the boson field by
\begin{equation}
U_g b^\dagger_{v \alpha} U^{-1}_g = D_{\alpha \beta}(g) b^\dagger_{gv,  \beta} \text{,}
\end{equation}
where $D(g)$ are unitary matrices giving a $n$-dimensional projective representation of $G$.  By choosing $D(g)$, we are choosing a projective symmetry group for the parton fields.\cite{wen02}  In Ref.~\onlinecite{essin13}, it was shown that there exists a finite-dimensional projective representation for any fractionalization class $[\omega_e] \in H^2(G, \zz)$, so the bosons can be taken to transform in any desired fractionalization class.

We now discuss two limits of ${\cal H}_{{\rm gauge}}$.  First, we consider the limit $h \to +\infty$.  In this limit, we have $\sigma^x_{\ell} = 1$, and the only remaining degrees of freedom are the bosons.  The gauge constraint becomes
\begin{equation}
b^\dagger_{v \alpha} b^{\vphantom\dagger}_{v \alpha} = \left\{ \begin{array}{ll}
\text{even,} & K^e_v = 1 \\
\text{odd,} & K^e_v = -1 
\end{array}\right.  \text{.}
\end{equation}
In this Hilbert space, all local operators transform linearly under $G$, and so the model reduces to a legitimate local bosonic model in this limit.  Following the usual logic of parton constructions,\cite{senthil00,wen02,gchen12} $H_{{\rm gauge}}$ can be viewed as a low-energy effective theory for local bosonic models with the same Hilbert space and symmetry action as in the $h \to + \infty$ limit.  The expectation is that any phase realized by $H_{{\rm gauge}}$ can be realized by some such local bosonic model, although this approach does not tell us how to choose parameters of the local bosonic model to realize the corresponding phase of $H_{{\rm gauge}}$.

Now we consider the exactly solvable limit of $H_{{\rm gauge}}$ with $h = 0$. This limit is deep in the deconfined phase of the $\zz$ gauge theory, and we have
$B_p = 1$, $A_v  = K^e_v$, and $b^\dagger_{v \alpha} b^{\vphantom\dagger}_{v \alpha} = 0$ acting on ground states.  Because $A_v = K^e_v$, the $m$ particles feel the same pattern of background $\zz$ charge as in the original $TC_0(G)$ toric code model, and their fractionalization class is unchanged.  Now, however, the $\zz$-charged bosons become the $e$ particle excitations, so the $e$ particle fractionalization class $[\omega_e]$ is determined by the (arbitrarily chosen) projective representation $D(g)$.  We have thus obtained a phase with $\zz$ topological order and  symmetry class $\langle [\omega_e], [\omega_m] \rangle$, as desired.

We now present the second result of this section, namely we establish a connection between space group symmetry classes and the symmetry classes of certain on-site symmetries.  Suppose that we have a local bosonic model with symmetry $G \times G_o$, where $G$ is the space group, and $G_o$ is a finite, unitary on-site symmetry.  We do not assume square lattice symmetry here, but allow for a more general space group.  We require $G_o$ to be isomorphic to some finite quotient of the space group $G$.  For example, if $G$ is square lattice space group symmetry, we could take $G_o \simeq G / T_2$, where $T_2$ is the normal subgroup of $G$ generated by translations $T_x^2$ and $T_y^2$.  In this case, $G_o$ can be nicely described as what remains of the space group when the system is put on a $2 \times 2$ periodic torus.

Next, we suppose our model has $\zz$ topological order, and the action of symmetry is described by $e$ and $m$ fractionalization classes $[\omega_e]$ and $[\omega_m]$.  Specifying these fractionalization classes in terms of generators and relations, we further assume that the only relations with non-trivial projective phase factors (\emph{i.e.}, $\sigma$ parameters) are those involving only elements of $G_o$.  That is, space group symmetry $G$ acts linearly on $e$ and $m$ particles, and elements $g \in G$ commute with $g_o \in G_o$ when acting on $e$ and $m$ particles.  Basically, we are assuming that we have some non-trivial action of on-site symmetry, where the space group symmetry ``comes along for the ride.''  As an aside, there are some interesting open questions hidden in our assumptions.  For example, is every symmetry class of the on-site $G_o$ that can be realized in local bosonic models also compatible with an arbitrary space group symmetry $G$? Or, are there $G_o$ symmetry classes that are only compatible with a given space group $G$ if some elements of $G_o$ and $G$ are chosen not to commute acting on $e$ and/or $m$ particles?

With our assumptions specified, we proceed to break the symmetry down to the subgroup $G' \subset G \times G_o$, defined as the set of all elements of the form $(g, \phi(g) ) \in G \times G_o$, where $g \in G$ is arbitrary, and $\phi : G \to G_o$ is the quotient map.  It is easy to see that $G'$ is a subgroup, and that it is isomorphic to $G$.  We thus still have $G$ space group symmetry, but now the space group operations are combined with on-site symmetry operations.  Under the new reduced symmetry, it is easy to see that new $[\omega_e]'$ and $[\omega_m]'$ fractionalization classes for $G'$ symmetry are induced by corresponding $G_o$ fractionalization classes before breaking the symmetry.  While these remarks remain somewhat abstract at present, this discussion shows that progress in understanding symmetry classes of finite, unitary on-site symmetry\cite{Hung2013,Lu2013,chen14, lukaszpc} can potentially have direct applications to similar problems for space group symmetry.  

\acknowledgments

M.H. is grateful for related collaborations with Andrew Essin. We are also grateful for useful correspondence with Lukasz Fidkowski. H.S. thanks the hospitality of the Erwin Schr\"{o}dinger International Institute for Mathematical Physics (ESI) in Vienna during his attending the Programme on ``Topological phases of quantum matter,'' where some of this paper was written. This work was supported by the David and Lucile Packard Foundation.

\appendix

\section{Complete set of commuting observables}
\label{app:complete-set}

We show here that the operators $\{ A_v | v \in V \}$, $\{ B_p | p \in P \}$, $\cL^e_{s_x}$, $\cL^e_{s_y}$, as defined in Sec.~\ref{sec:genlatt}, form a complete set of commuting observables for any model in the family $TC(G)$.  The approach is to construct a basis that is completely labeled by the simultaneous eigenvalues of these operators.

We recall that plaquettes $P$ together with $s_x$, $s_y$, form an elementary set of cycles, so that for any $c \in C$, $\cL^e_c$ can be decomposed into a product of $\cL^e_p$'s, with the product possibly also including $\cL^e_{s_x}$ and/or $\cL^e_{s_y}$.  However, the plaquettes are in general not independent, in the sense that there may be non-trivial relations of the form $B_{p_1} \cdots B_{p_n} = 1$, for some $p_1, \dots, p_n \in P$.  For the present purpose, it will be convenient to construct an elementary and independent set of cycles.

Let ${\cal T}$ be a spanning tree of the graph ${\cal G}$.  By definition, ${\cal T}$ is a  subgraph of ${\cal G}$ containing all vertices of ${\cal G}$ (${\cal T}$ spans ${\cal G}$), so that ${\cal T}$ is connected and has no cycles (${\cal T}$ is a tree).  Any tree with $n$ vertices has $n - 1$ edges, so ${\cal T}$ has $|V| - 1$ edges.  We denote the edge set of ${\cal T}$ by $E_{{\cal T}}$, and let $E' = E - E_{\cal T}$.  For every $\ell \in E'$, there is a unique cycle $c(\ell) \in C$ containing only $\ell$ and edges in $E_{{\cal T}}$.  We claim $\{ c(\ell) | \ell \in E' \}$ is an elementary, independent set of cycles.

To show the $c(\ell)$ cycles are elementary, suppose $c$ is a cycle.  Without loss of generality, we assume $c$ has no repeated edges.  Viewing $c$ as a subset of $E$, let $c \cap E' = \{ \ell_1, \dots, \ell_n \}$.  Then we claim the desired result, namely
\begin{equation}
\cL^e_c = \cL^e_{c(\ell_1)} \cdots \cL^e_{c(\ell_n)} \text{.} \label{eqn:cycle-decomposition}
\end{equation}
To show this, consider the product $\cL^e_c \cL^e_{c(\ell_1)} \cdots \cL^e_{c(\ell_n)} = \prod_{\ell \in c'} \sigma^z_\ell$.  $c'$  lies entirely in $E_{{\cal T}}$, and must be empty or a union of disjoint cycles. These two facts are only consistent if $c'$ is  empty, and so $\cL^e_c \cL^e_{c(\ell_1)} \cdots \cL^e_{c(\ell_n)} = 1$, equivalent to Eq.~(\ref{eqn:cycle-decomposition}).  

The $c(\ell)$ cycles are also independent:  we can choose the eigenvalues of $\cL^e_{c(\ell)}$ independently for all $\ell \in E'$.  To see this, consider a reference state $| \{ \Phi_\ell \} \rangle$, defined as the eigenstate of $\sigma^z_\ell$ satisfying
\begin{equation}
\sigma^z_{\ell} | \{ \Phi_\ell \} \rangle = \left\{ \begin{array}{ll}
| \{ \Phi_\ell \} \rangle , & \ell \in E_{{\cal T}} \\
\Phi_\ell | \{ \Phi_\ell \} \rangle , & \ell \in E'  \text{,}
\end{array}\right.
\end{equation}
where $\Phi_\ell \in \{ \pm 1 \}$.  There are clearly $2^{|E| - |V| + 1}$ such reference states, which form an orthonormal set, because $E'$ contains $|E| - |V| + 1$ edges.  Also, we clearly have
\begin{equation}
\cL^e_{c(\ell)} | \{ \Phi_\ell \} \rangle = \Phi_\ell | \{ \Phi_\ell \} \rangle \text{.}
\end{equation}

From the above discussion, it is clear that for every set of $\cL^e_{c(\ell)}$ eigenvalues $\{ \Phi_\ell \}$ there is a corresponding distinct consistent choice of $\{ B_p\}$, $\cL^e_{s_x}$ and $\cL^e_{s_y}$ eigenvalues, and vice versa.  For the purpose of constructing a complete set of commuting observables, we can therefore replace $\{ B_p \}$, $\cL^e_{s_x}$ and $\cL^e_{s_y}$ by $\{ \cL^e_{c(\ell)} | \ell \in E' \}$.

We will complete the discussion by exhibiting an orthonormal basis, where the basis states are simultaneous eigenvalues of $\{ A_v \}$ and $\{ \cL^e_{c(\ell)} \}$.  We construct the basis states starting from the reference states $| \{ \Phi_\ell \} \rangle$. Let $a_v \in \{ \pm 1 \}$, subject to the constraint $\prod_v a_v = 1$, then we consider the state
\begin{equation}
| \{ a_v \} , \{ \Phi_\ell \} \rangle = \frac{1}{\sqrt{2}} \prod_{v \in V} \frac{1}{\sqrt{2}} (1 + a_v A_v) | \{ \Phi_\ell \} \rangle \text{.}
\end{equation}
These states are normalized, and satisfy
\begin{eqnarray}
A_v |  \{ a_v \} , \{ \Phi_\ell \} \rangle &=& a_v  |  \{ a_v \} , \{ \Phi_\ell \} \rangle \\
\cL^e_{c(\ell)}  |  \{ a_v \} , \{ \Phi_\ell \} \rangle &=& \Phi_\ell  |  \{ a_v \} , \{ \Phi_\ell \} \rangle \text{,}
\end{eqnarray}
thus forming an orthonormal set.  Moreover, since there are $2^{|V| - 1}$ possible choices of $\{ a_v \}$, the number of states $| \{ a_v \} , \{ \Phi_\ell \} \rangle$  is  $2^{|V| - 1} \cdot 2^{|E| - |V| + 1} = 2^{|E|}$.  This is the dimension of the Hilbert space, so we exhibited a basis completely labeled by the eigenvalues of $\{ A_v \}$ and $\cL^e_{c(\ell)}$.

\section{Symmetry-invariant ground states}
\label{app:psi0}

For an even by even lattice (\emph{i.e.} $L$ even), it is always possible to choose $U_g$ and find a ground state $| \psi_{0e} \rangle$ satisfying 
\begin{eqnarray}
U_g | \psi_{0e} \rangle &=& | \psi_{0e} \rangle \\
\cL^e_{s_x} | \psi_{0 e} \rangle &=& \cL^e_{s_y} | \psi_{0 e} \rangle = | \psi_{0 e} \rangle \text{,}
\end{eqnarray}
where $s_x$ and $s_y$ are closed paths that wind around the system once in the $x$ and $y$ directions, respectively.
From this it also follows that $U_{g_1} U_{g_2} = U_{g_1 g_2}$; this equation holds acting on $| \psi_{0 e} \rangle$, so the linear action of symmetry on local operators [Eq.~(\ref{eqn:linear-action})] implies it holds on all states.

In fact, it is also possible to find a ground state $| \psi_{0 m} \rangle$ satisfying similar properties but for $m$-string operators:
\begin{eqnarray}
U_g | \psi_{0m} \rangle &=& | \psi_{0m} \rangle \\
\cL^m_{t_x} | \psi_{0 m} \rangle &=& \cL^m_{t_y} | \psi_{0 m} \rangle = | \psi_{0 m} \rangle \text{.}
\end{eqnarray}
Here, $t_x$ and $t_y$ are closed cuts winding once around the system in $x$ and $y$ directions, respectively.  Because, for instance, $\cL^e_{s_x}$ and $\cL^m_{t_y}$ must anti-commute, $| \psi_{0 e} \rangle$ and $ | \psi_{0 m} \rangle$ cannot be the same state.

We now show the existence of $| \psi_{0 e} \rangle$; the argument for $| \psi_{0 m} \rangle$ is essentially identical, apart from one subtlety that we address at the end of this Appendix.
We define $s_x$ by first drawing a path $s^0_x$ joining  an arbitrary $v \in V$ to $T_x v$.  The path $s_x$ is then formed by joining $s^0_x$, $T_x s^0_x$, $T_x^2 s^0_x$, and so on, to obtain
\begin{equation}
s_x = (s^0_x) (T_x s^0_x ) \cdots (T^{L-1}_x s^0_x ) \text{,}
\end{equation}
a closed path winding once around the system in the $x$-direction.  We then choose $s_y = P_{xy} s_x$.  With these paths specified, we specify a unique state in the four-dimensional ground state manifold by requiring
\begin{equation}
\cL^e_{s_x} | \psi_0 \rangle = \cL^e_{s_y} | \psi_0 \rangle = | \psi_0 \rangle \text{.}
\end{equation}

By symmetry, $U_g | \psi_0 \rangle$ must also lie in the ground state manifold for all $g \in G$.  We will show that
\begin{equation}
\cL^e_{s_{\mu}} U_g | \psi_0 \rangle = U_g | \psi_0 \rangle \text{,}  \label{eqn:gs-cond}
\end{equation}
for $\mu = x,y$, which implies $U_g | \psi_0 \rangle = e^{i \phi_g} | \psi_0 \rangle$, for some phase factors $e^{i \phi_g}$.  It is enough to show this for the generators $g = T_x, P_x, P_{xy}$.  Once this is established, we can make trivial phase redefinitions $U_{T_x} \to e^{-i \phi_{T_x} } U_{T_x}$, and similarly for the other generators, thus setting $\phi_g = 0$ to obtain the desired result.

Before proceeding to show Eq.~(\ref{eqn:gs-cond}) for each generator in turn, we obtain an equivalent simpler condition.  We have
\begin{equation}
\cL^e_{s_{\mu}} U_g | \psi_0 \rangle = U_g U^{-1}_g \cL^e_{s_{\mu}} U_g | \psi_0 \rangle 
= c^z_{g^{-1}}(s_{\mu}) U_g \cL^e_{g^{-1} s_{\mu} } | \psi_0 \rangle \text{,}
\end{equation}
for $\mu = x,y$.  Now, it is clear we can break $s_{\mu}$ into two paths, $s_{\mu} = s_{\mu 1} s_{\mu 2}$, so that $s_{\mu 2} = T^{L/2}_{\mu} s_{\mu 1}$.  Then we have
\begin{equation}
c^z_{g^{-1}}(s_{\mu}) = c^z_{g^{-1}} (s_{\mu 1} )  c^z_{g^{-1}} ( T^{L/2}_{\mu} s_{\mu 1} ) \text{.}
\end{equation}
Using Lemma~\ref{lm:cT} of Appendix~\ref{app:tcsoc}, $c^z_{g^{-1}} ( T^{L/2}_{\mu} s_{\mu 1} ) = c^z_{g^{-1}} (s_{\mu 1} )$, so that
\begin{equation}
c^z_{g^{-1}}(s_{\mu}) = 1 \text{.}
\end{equation}
Therefore we have shown
\begin{equation}
\cL^e_{s_{\mu}} U_g | \psi_0 \rangle = U_g \cL^e_{g^{-1} s_{\mu} } | \psi_0 \rangle \text{.}
\end{equation}
This implies Eq.~(\ref{eqn:gs-cond}) will hold if, for each generator $g$,
\begin{equation}
\cL^e_{g^{-1} s_{\mu} } | \psi_0 \rangle  = | \psi_0 \rangle \text{.} \label{eqn:simple-gs-cond}
\end{equation}

Now we consider $g = T_x$.  Since $T_{x}^{-1} s_x = s_x$, Eq.~(\ref{eqn:simple-gs-cond}) is satisifed for $\mu = x$.  For $\mu = y$, we have
\begin{equation}
\cL^e_{T_x^{-1} s_y} | \psi_0 \rangle =\cL^e_{c} \cL^e_{s_y} | \psi_0 \rangle = \cL^e_c | \psi_0 \rangle = | \psi_0 \rangle \text{,}
\end{equation}
where $c = s_y \cup T_x^{-1} s_y$, and the last equality follows from a graphical argument in Fig.~\ref{fig:gTx}. Here and in the following, for the union $\cup$ operation to make sense, we can view paths as multisets of edges. And the meaning of $\cL^e_{c}$ is obvious; it is a product of $\sigma^{z}_l$ with multiplicities taken into account.

\begin{figure}
\includegraphics[width=0.8\columnwidth]{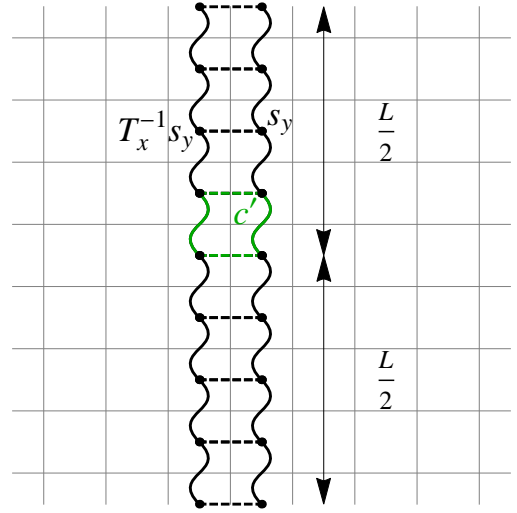}
\caption{(Color online) Graphical argument that $\cL^e_{c} | \psi_0 \rangle = | \psi_0 \rangle$, for $c = s_y \cup T^{-1}_x s_y$. The dotted lines show the $L \times L$ grid of primitive cells, and the paths $s_y$ and $T^{-1}_x s_y$ are shown.  $c$ encloses a region of area $L$, which can be broken (dashed lines) into $L$ smaller sub-regions each of unit area.  Let $c'$ be the cycle bounding one of the sub-regions, then $\cL^e_c = \prod_{n=0}^{L-1} \cL^e_{ T^n_x c' }$.  In addition, by translation symmetry $\cL^e_{c'} | \psi_0 \rangle = \cL^e_{T_y c'} | \psi_0 \rangle = \pm | \psi_0 \rangle$.  Since an even number of sub-regions appear in the decomposition of $\cL^e_c$ given above, we have $\cL^e_c | \psi_0 \rangle = | \psi_0 \rangle$.}
\label{fig:gTx}
\end{figure}

Next we consider $g = P_{xy}$, and Eq.~(\ref{eqn:simple-gs-cond}) becomes
\begin{equation}
\cL^e_{P_{xy} s_{\mu} } | \psi_0 \rangle  = | \psi_0 \rangle \text{.}
\end{equation}
This clearly holds, because $P_{xy} s_y = P_{xy}^2 s_x = s_x$, and $P_{xy} s_x = s_y$.

Finally, we consider $g = P_x$.  For $\mu = x$, we have
\begin{equation}
\cL^e_{P_x s_x} | \psi_0 \rangle = \cL^e_c \cL^e_{s_x} | \psi_0 \rangle = \cL^e_c | \psi_0 \rangle = | \psi_0 \rangle \text{,}
\end{equation}
where $c = s_x \cup P_x s_x$, and the last equality follows from an argument we now provide.  We first cut  $s_x$ into two equal-length pieces $s_{x1}$ and $T^{L/2} s_{x1}$, which meet at a vertex $v$.  We then have
\begin{eqnarray}
P_x s_x &=& (P_x s_{x1}) (P_x T^{L/2}_x s_{x1} ) 
= (P_x s_{x1}) ( T^{-L/2}_x P_x s_{x1} ) \nonumber  \\
&=& (P_x s_{x1}) ( T^{L/2}_x P_x s_{x1} ) \text{,}
\end{eqnarray}
where the last equality holds since $T^L_x = 1$.  We have thus decomposed $c = (s_{x1})( T^{L/2}_x s_{x1} )  \cup ( P_x s_{x1} ) ( T^{L/2}_x P_x s_{x1} )$.  Now we draw a path $s'$ joining $v$ to $P_x v$, and we decompose $\cL^e_c = \cL^e_{c_1} \cL^e_{c_2}$, introducing the cycles
\begin{eqnarray}
c_1 &=& (s_{x1}) (s') (T^{L/2}_x P_x s_{x1} ) (T^{L/2}_x s' ) \\
c_2 &=& (T^{L/2}_x s_{x1}) (T^{L/2}_x s') (P_x s_{x1} ) (s' ) \text{.}
\end{eqnarray}
Because $c_2 = T^{L/2}_x c_1$, it follows from symmetry that $\cL^e_{c_1} | \psi_0 \rangle = \cL^e_{c_2} | \psi_0 \rangle  = \pm | \psi_0 \rangle$, and $\cL^e_c | \psi_0 \rangle = | \psi_0 \rangle$, as desired.  This argument is illustrated graphically in Fig.~\ref{fig:gPxx}.

\begin{figure}
\includegraphics[width=0.8\columnwidth]{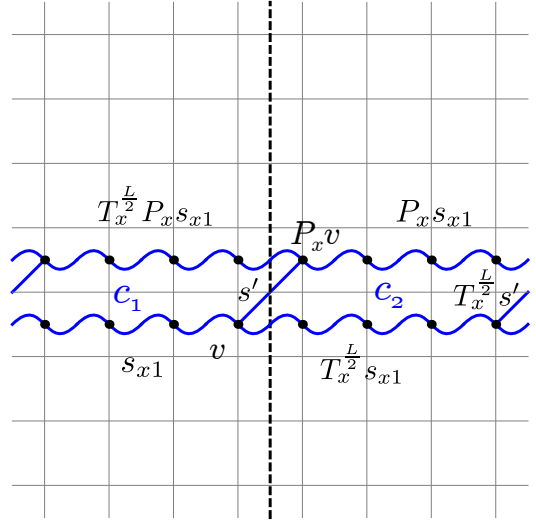}
\caption{(Color online) Graphical illustration of the argument that $\cL^e_{c} | \psi_0 \rangle = | \psi_0 \rangle$, for $c = s_x \cup P_x s_x$.  It is important to note that, in the interest of clarity, this figure is schematic in the sense that it accurately shows the connectivity of the paths involved, and their properties under translation symmetry, but \emph{not} their properties under $P_x$.  The various symbols are defined in the main text.  The vertical dashed line is the $P_x$ reflection axis, and the vertex $v$ has been chosen to lie near this axis for convenience.  $c_1$ and $c_2$ are the boundaries of the left and right shaded regions, respectively.  The most important point is that these two regions are related by $T^{L/2}_x$ translation.}
\label{fig:gPxx}
\end{figure}

For $g = P_x$ and $\mu = y$, we have
\begin{equation}
\cL^e_{P_x s_y} | \psi_0 \rangle = \cL^e_c \cL^e_{s_y} | \psi_0 \rangle = \cL^e_c | \psi_0 \rangle = | \psi_0 \rangle \text{,}
\end{equation}
where $c = s_y \cup P_x s_y$, and the last equality follows from an argument given below, which is similar to that already given in the case $\mu = x$.  We first break $s_y$ into two equal-length paths related by $T^{L/2}_y$ translation, that is
\begin{equation}
s_y = (s_{y1} ) ( T^{L/2}_y s_{y 1} ) \text{,}
\end{equation}
and let $v$ be a vertex where $s_{y 1}$ and $T^{L/2}_y s_{y 1}$ meet.  Since $P_x$ and $T_y$ commute, we have
\begin{equation}
P_x s_y = (P_x s_{y 1} ) (T^{L/2}_y P_x s_{y 1} ) \text{.}
\end{equation}
We can then proceed following the discussion for $g = P_x$, $\mu = x$ to obtain the desired result.

The argument for the existence of $| \psi_{0 m} \rangle$ is essentially identical.  However, there is one subtlety that should be addressed.  In establishing symmetry-invariance of $| \psi_{0 e} \rangle$, we had to choose the phase of $U_g$ appropriately.  The same step arises in the corresponding discussion for $| \psi_{0 m} \rangle$, and the two phase choices may not be compatible.  Fortunately, this is not an issue for our purposes, because we never need to work with $| \psi_{0 e} \rangle$ and $ | \psi_{0 m} \rangle$ at the same time.  We simply make (possibly) different phase choices for $U_g$ depending on the ground state we are working with in a given calculation.

\section{General construction of $e$ and $m$ localizations in toric code models}
\label{app:eloc}

Here, we show by explicit construction that an $e$-localization $U^e_g(v)$ always exists for the toric code models, and also that this $e$-localization is unique up to projective transformations $U^e_g(v) \to \lambda(g) U^e_g(v)$, with $\lambda(g) \in \{ \pm 1 \}$.  The explicit form for the $e$-localization we obtain is useful for obtaining general constraints on symmetry classes in Appendix~\ref{app:general-constraints}.  The corresponding results and explicit form also hold for $m$-localizations.  We focus first on $e$ particles and $e$-localizations, postponing discussion of $m$ particles to the end of this Appendix.

We fix $g \in G$, and arbitrarily single out a vertex $v_0$. $v_0$ may depend on $g$, but we do not write this explicitly.  We then  choose $U^e_g(v_0) = f^e_g(v_0) \cL^e_{s^e_g(v_0)}$, where $f^e_g(v_0) \in \{ \pm 1\}$ is arbitrary, and $s^e_g(v_0)$ is arbitrary so long as it joins $v_0$ to $g v_0$.  In addition, for each $v \neq v_0$, we choose a path $s_v$ joining $v_0$ to $v$.
We will now show that
\begin{equation}
U^e_g(v) = \cL^e_{s_v} U^e_g(v_0) g( \cL^e_{s_v} )  \label{eqn:explicit-eloc}
\end{equation}
gives an $e$-localization.  Here, we have introduced the notation $g( {\cal O} ) = U_g {\cal O} U^{-1}_g$ for any operator ${\cal O}$.  It is clear that $U^e_g(v)$ can be put into the form $U^e_g(v) = f^e_g(v) \cL^e_{s^e_g(v)}$.

To proceed, we need to show that
\begin{equation}
U_g | \psi_e (s) \rangle = U^e_g[ v_1(s) ] U^e_g [ v_2 (s) ] | \psi_e(s) \rangle \label{eqn:symloc}
\end{equation}
for all open paths $s$.  The endpoints of $s$ are denoted $v_1(s), v_2(s)$.  We first show that Eq.~(\ref{eqn:symloc}) holds for all pairs of $e$ particle positions (\emph{i.e.} all pairs of endpoints $v_1(s), v_2(s)$), using a specific choice of paths.  Then we proceed to show Eq~(\ref{eqn:symloc}) it holds for any open path $s$.  

If $s = s_v$, the endpoints of $s$ are $v_0$ and $v$, and an easy calculation shows Eq.~(\ref{eqn:symloc}) holds.  Now we consider vertices $v, v' \neq v_0$ and $v \neq v'$, which are joined by the path $s_{v v'} = s_v s_v'$, and we choose $s = s_{v v'}$.  We have
\begin{equation}
| \psi_e(s_{v v'} ) \rangle = \cL^e_{s_{v v'} } | \psi_0 \rangle = \cL^e_{s_{v} } \cL^e_{s_{v'} } | \psi_0 \rangle \text{.}
\end{equation}
Then, for the left-hand side of Eq.~(\ref{eqn:symloc}),
\begin{equation}
U_g  | \psi_e(s_{v v'} ) \rangle = g( \cL^e_{s_{v v'} } ) | \psi_0 \rangle = g( \cL^e_{s_{v} } ) g( \cL^e_{s_{v'} } ) | \psi_0 \rangle \text{.}
\end{equation}
The right-hand side of Eq.~(\ref{eqn:symloc}) can easily be verified after observing that
\begin{equation}
U^e_g (s_v) U^e_g( s_{v'} ) = U^e_g (s_v) U^e_g(v_0) U^e_g (s_{v'}) U^e_g( v_0 ) \text{,}
\end{equation}
since $ [ U^e_g(v_0) ]^2 = 1$.

Now, consider $| \psi^e(s) \rangle$, where $s$ has endpoints $v, v'$, with $v, v' \neq v_0$.  We have
\begin{equation}
| \psi^e(s) \rangle = c \cL^e_{s s_{v v'} } | \psi_e (s_{v v'} ) \rangle =  | \psi_e (s_{v v'} ) \rangle \text{,}
\end{equation}
where we used the fact that $| \psi_e(s_{v v'}) \rangle$ is an eigenstate of any closed $e$-string operator, and where $c = \pm 1$ is the eigenvalue of $\cL^e_{s s_{v v'} }$ acting on $| \psi_e(s_{v v'}) \rangle$.  The corresponding result holds when $s$ has endpoints $v_0, v$.  Therefore, Eq.~(\ref{eqn:symloc}) holds independent of the choice of $s$.

To consider uniqueness of the symmetry localization, it is convenient to use the form $U^e_g(v) = f^e_g(v) \cL^e_{s^e_g(v)}$.  The endpoints of $s^e_g(v)$ are fixed, but the path is otherwise arbitrary.  However, we are always free to deform the paths $s^e_g(v)$ to some fixed set of reference paths, since this only affects the overall phase factor $f^e_g(v)$.  Therefore it is enough to consider the redefinition  $U^e_g(v) \to \lambda(g, v) U^e_g(v)$.  We now show that Eq.~(\ref{eqn:symloc}) requires $\lambda(g,v) = \lambda(g)$, independent of $v$, which is precisely the general form of projective transformations.  Suppose there exist vertices $v_1, v_2$ with $\lambda(g,v_1) \neq \lambda(g, v_2)$.  Such a transformation changes the right-hand side of Eq.~(\ref{eqn:symloc}) by a minus sign for a state with $e$ particles at $v_1$ and $v_2$, and is not consistent.  Therefore the most general redefinition of the $e$-localization is the projective transformation $U^e_g(v) \to \lambda(g) U^e_g(v)$.

The obvious parallel discussion establishes the corresponding results for $m$-localizations.  The corresponding explicit form for the $m$-localization is
\begin{equation}
U^m_g(h) = \cL^m_{t_h} U^m_g(h_0) g( \cL^m_{t_h} ) \text{.}  \label{eqn:explicit-mloc}
\end{equation}
This is obtained following the above discussion upon replacing vertices by holes  ($v_0 \to h_0, v \to h$), and paths by cuts ($s_v \to t_h$).

\section{General constraints on symmetry classes in toric code
models}
\label{app:general-constraints}

\subsection{Toric codes without spin-orbital coupling}
\label{app:nosoc-constraints}

Here, we consider models in the family $TC_0(G)$, and  prove Theorem~\ref{thm:nosoc_maintext} stated in Section~\ref{sub:tc_woso}.

We first introduce some additional notation to be used below. Recalling that $a_{v}$
is the ground state eigenvalue of $A_{v}$, we define
$a_{X}=\prod_{v\in X}a_{v}$ for any finite subset $X \subset V$. If $t$ is
a simple closed cut (see Sec.~\ref{sec:genlatt} for a definition), we define $V_{t}=\left\{ v\in V| \mathscr{P}\left(v\right)\text{ is enclosed by }t\right\} $.
In addition, we define $\Gamma\left(g_{1},\dots,g_{n}\right)$ to be the set of vertices that are fixed under each of $g_1, \dots, g_n$.  That is,  $\Gamma(g_1, \dots, g_n) =\left\{ v\in V|g_{i}v=v,i=1,\cdots,n\right\}$.  To shorten various expressions, we write $g\left(\mathcal{O}\right)=U_{g}\mathcal{O}U_{g}^{-1}$
for the transformation of any local operator $\mathcal{O}$ under
$g\in G$, and define $R=P_{x}P_{xy}$ ($\pi/2$ counterclockwise rotation)
 and $P_{y}=P_{xy}P_{x}P_{xy}$ (reflection $y \to -y$).

In calculations below, we will use the $e$ and $m$ symmetry localizations given in Eqs.~(\ref{eqn:explicit-eloc}) and~(\ref{eqn:explicit-mloc}), and discussed in Appendix~\ref{app:eloc}.  In addition, we will often write equations like $\sigma^m_{pxpxy} = \cL^m_t$.  Such equations hold when acting on $| \psi_{0 m} \rangle$, or more generally on states created by acting on $| \psi_{0 m} \rangle$ with $m$-string operators, and should be interpreted in this way.

\begin{figure}
\includegraphics[width=0.8\columnwidth]{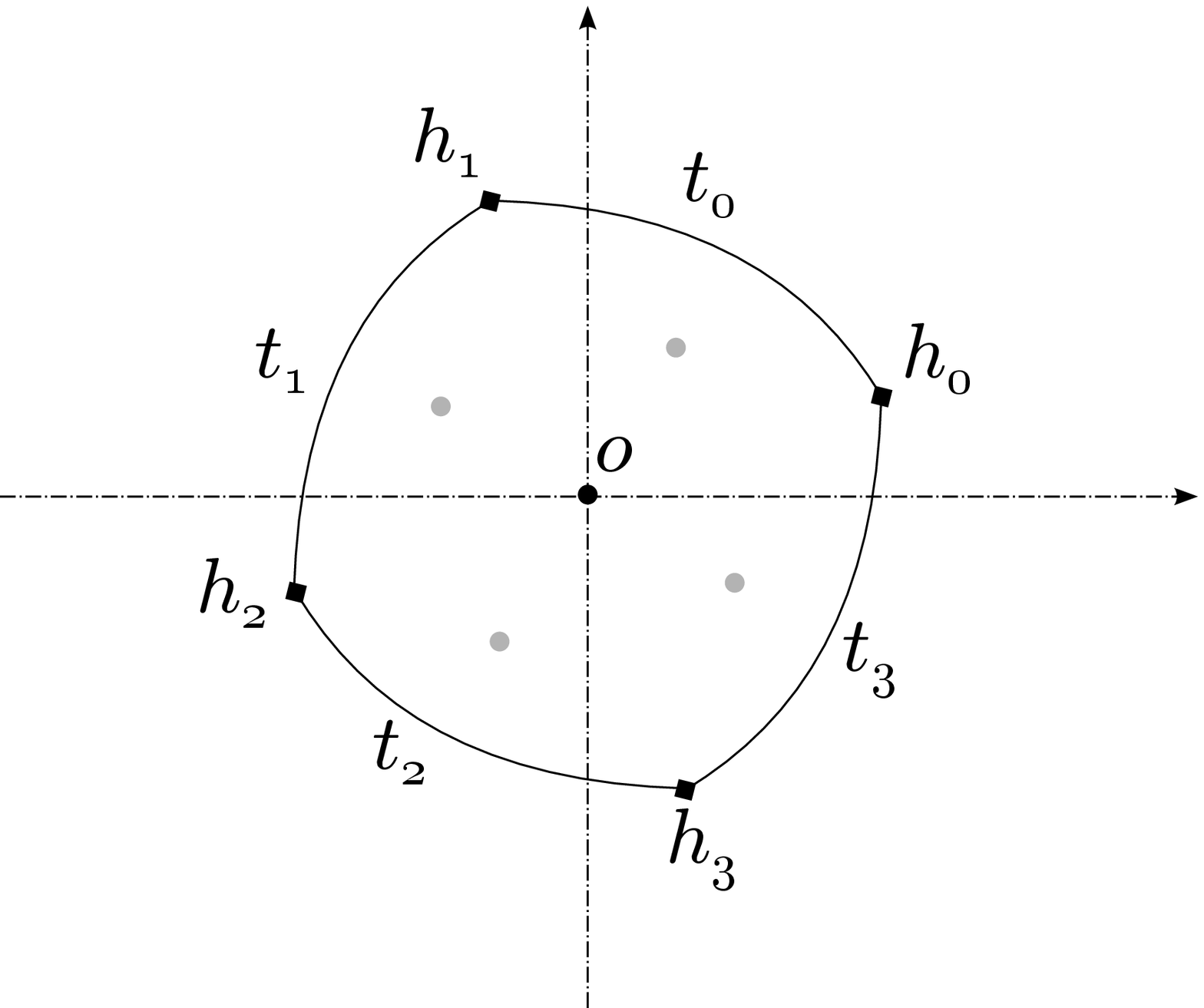}

\caption{  The calculation of $\sigma_{pxpxy}^{m}$. Put an $m$ particle at point
$h_{0}$, let $h_{j}=R^{j}\left(h_{0}\right),j=1,2,3$, let $t_{0}\in\bar{W}$
connecting $h_{0}$ to $h_{1}$ and $t_{j}=R^{j}t_{0}$. Then we have
$U_{P_{x}}^{m}U_{P_{xy}}^{m}\left(h_{0}\right)=f_{0}^{m}\mathcal{L}_{t_{0}}^{m}$
with $f_{0}^{m}\in\left\{ \pm1\right\} $ and $U_{P_{x}}^{m}U_{P_{xy}}^{m}\left(h_{j}\right)=\mathcal{L}_{t_{0}\cdots t_{j-1}}^{m}U_{P_{x}}^{m}U_{P_{xy}}^{m}\left(h_{0}\right)R\left(\mathcal{L}_{t_{0}\cdots t_{j-1}}^{m}\right)$
for $j=1,2,3$. With some calculation, $\left(U_{P_{x}}^{m}U_{P_{xy}}^{m}\right)^{4}\left(h_{0}\right)=\mathcal{L}_{t}^{m}$
with $t=t_{0}t_{1}t_{2}t_{3}$. Thus, $\sigma_{pxpxy}^{e}=a_{V_{t}}$.
If $\mathscr{P}\left(v\right)$ is enclosed by $t$ with $R^{2}v\neq v$,
then $v$, $Rv$, $R^{2}v$, $R^{3}v$ are four different vertices
enclosed by $t$ such as the grey vertices shown here. Since $\prod_{i=0}^{3}a_{R^{i}v}=1$,
we have $\sigma_{pxpxy}^{e}=a_{V_{t}}=a_{\mathscr{P}^{-1}\left(o\right)}=a_{\Gamma\left(R^{2}\right)}$.
The above statements are also true in the cases with spin-orbital coupling using the gauge choice described in Appendix~\ref{app:tcsoc}.
In the case without spin-orbital coupling, since $a_{v}=a_{P_{x}v}=a_{P_{xy}v}$,
we have $\sigma_{pxpxy}^{m}=a_{\Gamma\left(P_{x},P_{xy}\right)}$. }

\label{fig:m_rotation}
\end{figure}

\begin{lem}
\label{lm:mpxpxy}For any model in $TC_0(G)$,
we have $\sigma_{pxpxy}^{m}=a_{\mathscr{P}^{-1}\left(o\right)}=a_{\Gamma\left(P_{x},P_{xy}\right)}$,
where $o=\left(0,0\right)$. \end{lem}
\begin{proof}
As shown in Fig.~\ref{fig:m_rotation},  consider an $m$ particle
located at an arbitrary hole $h_{0}$ and let $h_{j}=R^{j} h_{0} ,j=1,2,3$.
Then draw a cut $t\in\bar{W}$ connecting $h_{0}$ and $h_{1}$. Let
$t_{j}=R^{j} t_{0} ,j=1,2,3$ and $t= t_0 t_1 t_2 t_3$. The cuts are chosen so that $t$ is simple.
Then we choose $U_{R}^{e} \left(h_{0}\right)= \mathcal{L}_{t_{0}}^{m}$, and using the results of Appendix~\ref{app:eloc},
\begin{eqnarray}
U_{R}^{m}\left(h_{1}\right) & = & \mathcal{L}_{t_{0}}^{m} \mathcal{L}_{t_{0}}^{m }R\left(\mathcal{L}_{t_{0}}^{m}\right),\label{eq:r1}\\
U_{R}^{m}\left(h_{2}\right) & = & \mathcal{L}_{t_{0}t_{1}}^{m} \mathcal{L}_{t_{0}}^{m }R\left(\mathcal{L}_{t_{0}t_{1}}^{m}\right),\label{eq:r2}\\
U_{R}^{m}\left(h_{3}\right) & = & \mathcal{L}_{t_{0}t_{1}t_{2}}^{m} \mathcal{L}_{t_{0}}^{m}R\left(\mathcal{L}_{t_{0}t_{1}t_{2}}^{m}\right),\label{eq:r3}\\
\left(U_{R}^{m}\right)^{4}\left(h_{0}\right) & = & \mathcal{L}_{t_{0}}^{m}\mathcal{L}_{t_{2}}^{m}R\left(\mathcal{L}_{t_{0}}^{m}\mathcal{L}_{t_{2}}^{m}\right)=\mathcal{L}_{t}^{m}.\label{eq:mrotation}
\end{eqnarray}
Therefore, $\sigma_{pxpxy}^{m}=a_{V_{t}}$. For $v\in V_{t}$, if
$\mathscr{P}\left(v\right)\neq o$ or $R^{2}v\neq v$, then $v$,
$Rv$, $R^{2}v$, $R^{3}v$ are four different vertices in $V_{t}$, with $a_v = a_{R v} = \cdots$.
Then $\prod_{i=0}^{3}a_{R^{i}v}=1$, and these vertices do not contribute to
$a_{V_{t}}$. We have thus shown $\sigma_{pxpxy}^{m} = a_{V_{t}}=a_{\mathscr{P}^{-1}\left(o\right)}=a_{\Gamma\left(R^{2}\right)}$, part of the desired result.

For $v\in\Gamma\left(R^{2}\right)$, we have $P_{x}v,P_{xy}v,Rv\in\Gamma\left(R^{2}\right)$
since $R^{2}$ commutes with $P_{x}$, $P_{xy}$. Let $G_{o}$ denote
the subgroup generated by $P_{x}$, $P_{xy}$, which is the same as the subgroup fixing the origin $o$. Let $G_{o}v$ be
the orbit of $v$ under $G_{o}$ and $G_{v}=\left\{ g\in G|gv=v\right\} $.
Then $G_{o}v\subseteq\Gamma\left(R^{2}\right)$ and $G_{v}$ is a
subgroup of $G_{o}$. Because $\left|G_{o}\right|=8$ and $R^{2}v=v$,
we have $\left|G_{o}v\right|=\left|G_{o}/G_{v}\right|=1,2,4$. Now
with the assumption that there is no spin-orbital coupling, then $a_{v}=a_{P_{x}v}=a_{P_{xy}v}=a_{Rv}$
and hence $a_{G_{o}v}=1$ unless $\left|G_{o}v\right|=1$. Therefore,
$\sigma_{pxpxy}^{m}=a_{\Gamma\left(R^{2}\right)}=a_{\Gamma\left(P_{x},P_{xy}\right)}$.\end{proof}

\begin{lem}
\label{lm:o_corner}Let $\widetilde{P_{x}}=T_{x}P_{x}$ and $\sigma_{\widetilde{px}pxy}^{a}=\left(U_{\widetilde{P_{x}}}^{a}U_{P_{xy}}^{a}\right)^{4}$
for $a=e,m$. Then we have $\sigma_{\widetilde{px}pxy}^{a}=\sigma_{pxpxy}^{a}\sigma_{txty}^{a}$,
for $a=e,m$.\end{lem}
\begin{proof}
Since $U^a_{\widetilde{P_x}} = \pm U^a_{T_x} U^a_{P_x}$, we have $\left(U_{\widetilde{P_{x}}}^{a}U_{P_{xy}}^{a}\right)^{4}
= \left( U^a_{T_x} U^a_{P_x} U_{P_{xy}} \right)^4$.  It is then straightforward to bring the $U^a_{T_x}$ operators to the left side of this product, using the relations of Eqs.~(\ref{eqn:projpx}-\ref{eqn:projtypx}), and the result follows.
\end{proof}

\begin{rem*}
This lemma is valid even with spin-orbital coupling allowed.
\end{rem*}

\begin{lem}
\label{lm:mpxpxymtxty}For any model in $TC_0(G)$,
we have $\sigma_{pxpxy}^{m}\sigma_{txty}^{m}=a_{\mathscr{P}^{-1}\left(\widetilde{o}\right)}=a_{\Gamma\left(P_{xy},T_{x}P_{x}\right)}$,
where $\widetilde{o}=\left(\frac{1}{2},\frac{1}{2}\right)$.\end{lem}
\begin{proof}
Repeat the proof to Lemma~\ref{lm:mpxpxy}, replacing $P_{x} \to \widetilde{P_x} = T_x P_x$, $o \to \widetilde{o} = \left( \frac{1}{2} ,\frac{1}{2} \right)$, and $\sigma_{pxpxy}^{m} \to \sigma_{\widetilde{px}pxy}^{m}$, obtaining the result $\sigma_{\widetilde{px}pxy}^{m}= a_{ {\mathscr P}^{-1} (\widetilde{o} ) } = a_{\Gamma\left(\widetilde{P_{x}},P_{xy}\right)}$.  The desired result then follows immediately from Lemma~\ref{lm:o_corner}.
\end{proof}

\begin{lem}
\label{lm:mpxpxytypx}For any model in $TC_0(G)$,
we have $\sigma_{pxpxy}^{m}\sigma_{typx}^{m}=a_{\mathscr{P}^{-1}\left( (0,\frac{1}{2} ) \right)}=a_{\Gamma\left(P_{x},T_{y}P_{y}\right)}$.\end{lem}
\begin{proof}
As shown in Fig.~\ref{fig:m_pxty}, pick a hole $h_{0}$ near the
y-axis, let $h_{1}=P_{x} h_{0}, h_{2}=T_{y} h_{0}, h_{3}=T_{y} h_{1}$.
Draw a cut $t_{0}$ connecting $h_{0}$, $h_{1}$ and a cut $t_{1}$
joining $h_{0}$ and $h_{2}$. Let $t_{2}=P_{x} t_{1}$,
$t_{3}=T_{y} t_{0}$ and $t=t_{0}t_{1}t_{3}t_{2}$.  The cuts, some of which may contain no edges, are chosen so that $t$ is simple.
We choose $h_{0}$ and $t$ such that all vertices enclosed
by $t$ are located on the y-axis and no vertex located above $\left(0,\frac{1}{2}\right)$
is enclosed by $t$. 

We choose $U_{P_{x}}^{m}\left(h_{0}\right)= \mathcal{L}_{t_{0}}^{m}$,
$U_{T_{y}}^{m}\left(h_{0}\right)=\mathcal{L}_{t_{1}}^{m}$, then following Appendix~\ref{app:eloc} we can choose 
\begin{eqnarray*}
U_{T_{y}}^{m}\left(h_{1}\right) & = & \mathcal{L}_{t_{0}}^{m}\mathcal{L}_{t_{1}}^{m} \, T_{y} (\mathcal{L}_{t_{0}}^{m} ),\\
U_{P_{x}}^{m}\left(h_{2}\right) & = & \mathcal{L}_{t_{1}}^{m}\mathcal{L}_{t_{0}}^{m} \, P_{x}(\mathcal{L}_{t_{1}}^{m}) .
\end{eqnarray*}
These results can be used to evaluate the product $U_{T_{y}}^{m}U_{P_{x}}^{m}\left(U_{T_{y}}^{m}\right)^{-1}\left(U_{P_{x}}^{m}\right)^{-1}$ acting on a $m$ particle initially located at $h_3$.  We obtain
\begin{equation}
\sigma_{typx}^m = \mathcal{L}_{t_{0}}^{m}\mathcal{L}_{t_{1}}^{m}T_{y}\left(\mathcal{L}_{t_{0}}^{m}\right)P_{x}\left(\mathcal{L}_{t_{1}}^{m}\right).\label{eq:mtypx}
\end{equation}

So far we have not assumed the absence of spin-orbit coupling.  Now making this assumption, we have $P_{x}\left(\mathcal{L}_{t_{1}}^{m}\right)=\mathcal{L}_{t_{2}}^{m}$
and $T_{y}\left(\mathcal{L}_{t_{0}}^{m}\right)=\mathcal{L}_{t_{3}}^{m}$.
Thus, $\sigma_{typx}^{m}=\mathcal{L}_{t}^{m}=a_{V_{t}}$. If $v\in V_{t}$
and $\mathscr{P}\left(v\right) \neq o,\left(0,\frac{1}{2}\right)$, then
$v$, $P_{y}v$ are two different vertices in $V_{t}$ by construction.
Since $a_{v}=a_{P_{y}v}$, their product does not contribute to $a_{V_{t}}$.
So $\sigma_{typx}^{m}=a_{\mathscr{P}^{-1}\left(o\right)}a_{\mathscr{P}^{-1}\left(\left(0,\frac{1}{2}\right)\right)}$.
Lemma~\ref{lm:mpxpxy} says $\sigma_{pxpxy}^{m}=a_{\mathscr{P}^{-1}\left(o\right)}$.
Thus, $\sigma_{pxpxy}^{m}\sigma_{typx}^{m}=a_{\mathscr{P}^{-1}\left(\left(0,\frac{1}{2}\right)\right)}$, part of the result to be shown.

Let $\kappa=\left(0,\frac{1}{2}\right)$, $G_{\kappa}=\left\{ g\in G|g\kappa=\kappa\right\} $
and $G_{\kappa}v$ the orbit of $v$ under $G_{\kappa}$. Then $G_{\kappa}v\subseteq\mathscr{P}^{-1}\left(\kappa\right)$
if $\mathscr{P}\left(v\right)=\kappa$, and $G_{\kappa}$ is generated
by $P_{x}$, $T_{y}P_{y}$. In addition, $\left|G_{\kappa}\right|=4$
and hence $\left|G_{\kappa}v\right|=1,2,4$. Since $a_{v'}=a_{v}$
for $v'\in G_{\kappa}v$, we have $a_{G_{\kappa}v}=1$ unless $\left|G_{\kappa}v\right|=1$.
Thus, $\sigma_{pxpxy}^{m}\sigma_{typx}^{m}=a_{\mathscr{P}^{-1}\left(\kappa\right)}=a_{\Gamma\left(P_{x},T_{y}P_{y}\right)}$.
\end{proof}
\begin{figure}
\includegraphics[width=0.8\columnwidth]{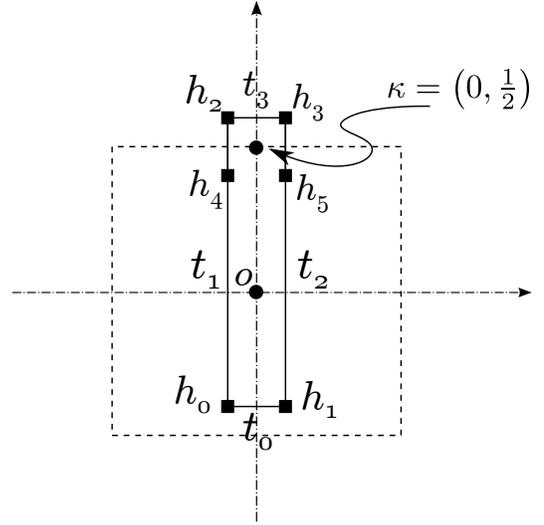}
\caption{Illustration of the calculation of $\sigma_{typx}^{m}$ in Lemmas~\ref{lm:mpxpxytypx} and~\ref{lm:mtypx_soc}. Solid squares denote the locations of holes $h \in H$, which are chosen so that $h_{0}$ is arbitrary (but near the y-axis), and $h_{1}=P_{x} h_{0} $, $h_{2}=T_{y} h_{0} $,
$h_{3}=T_{y} h_{1}$, $h_{4}=P_{y} h_{0}$,
$h_{5}=P_{y} h_{1}$.  $h_4$ and $h_5$ are not used in Lemma~\ref{lm:mpxpxytypx}.
Cuts are represented by solid lines, and $\widetilde{h_0 h_1}$ denotes, for example, a cut joining $h_0$ to $h_1$.  The cuts
$t_0 = \widetilde{h_{0}h_{1}}$, $t_1 = \widetilde{h_{0}h_{2}}$, $t_2 = \widetilde{h_{1}h_{3}}$, and $t_3 = \widetilde{h_{2}h_{3}}$ are labeled, and are chosen to have properties described in the text.  The points $o = (0,0)$ and $\kappa = (0, \frac{1}{2})$ are shown.}
\label{fig:m_pxty}
\end{figure}

\begin{lem}
\label{clm:mpx}For any model in $TC_0(G)$, we
have $\sigma_{px}^{m}=\sigma_{pxy}^{m}=\sigma_{txpx}^{m}=1$.\end{lem}
\begin{proof}
Given a hole $h_0$, let $h_1 = g h_0$, where $g=P_{x},P_{xy}$, or $T_{x}P_{x}$.
We can always draw a simple cut joining $h_0$ to $h_1$ so that $g t = t$.  We choose
$U_{g}^{m}\left(h_{0}\right)= \mathcal{L}_{t}^{m}$, and by Appendix~\ref{app:eloc} we can choose
\[
U_{g}^{m}\left(h_{1}\right)=\mathcal{L}_{t}^{m} \mathcal{L}_{t}^{m}g\left(\mathcal{L}_{t}^{m}\right)= \mathcal{L}_{t}^{m},
\]
where we used the assumption of no spin-orbit coupling.
Then $U_{g}^{m}\left(h_{1}\right)U_{g}^{m}\left(h_{0}\right)=1$.  We place a $m$ particle at $h_0$, and compute $(U^m_g)^2$ acting on this $m$ particle, finding $(U^m_g)^2(h_0) = U^m_g(h_1) U^m_g(h_0) = 1$.  Thus, 
$\sigma_{px}^{m}=\sigma_{pxy}^{m}=\sigma_{txpx}^{m}=1$.\end{proof}

\begin{lem}
\label{lm:e_inv}Suppose $g\in G$ such that $g^{2}=1$ and there
is $v\in V$ such that $gv=v$. Then $\left(U_{g}^{e}\right)^{2}=1$. \end{lem}
\begin{proof}
Because $gv=v$, we have $U_{g}^{e}\left(v\right)=1$ or $-1$. So
\[
\left(U_{g}^{e}\right)^{2}\left(v\right)=\left(U_{g}^{e}\left(v\right)\right)^{2}=1.
\]
Therefore, $\left(U_{g}^{e}\right)^{2}=1$.\end{proof}

\begin{rem*}
This lemma is valid even with spin-orbital coupling allowed.
\end{rem*}

\begin{thm}
The TC symmetry classes in $\mathsf{A}$, $\mathsf{B}$, $\mathsf{C}$,
$\mathsf{M}$, $\mathsf{M_{1}}$, $\mathsf{M_{2}}$ and $\mathsf{M_{3}}$
are not realizable in $TC_0(G)$, where
\begin{eqnarray*}
\mathsf{A} & = & \left\{ \sigma_{pxpxy}^{e}=\sigma_{pxpxy}^{m}=-1\right\} ,\\
\mathsf{B} & = & \left\{ \sigma_{pxpxy}^{e}\sigma_{txty}^{e}=\sigma_{pxpxy}^{m}\sigma_{txty}^{m}=-1\right\} ,\\
\mathsf{C} & = & \left\{ \sigma_{pxpxy}^{e}\sigma_{typx}^{e}=\sigma_{pxpxy}^{m}\sigma_{typx}^{m}=-1\right\} ,\\
\mathsf{M} & = & \left\{ \sigma_{px}^{m}=-1\vee\sigma_{pxy}^{m}=-1\vee\sigma_{txpx}^{m}=-1\right\} ,\\
\mathsf{M_{1}} & = & \left\{ \sigma_{pxpxy}^{m}=-1\wedge\left(\sigma_{px}^{e}=-1\vee\sigma_{pxy}^{e}=-1\right)\right\} ,\\
\mathsf{M_{2}} & = & \left\{ \sigma_{pxpxy}^{m}\sigma_{txty}^{m}=-1\wedge\left(\sigma_{pxy}^{e}=-1\vee\sigma_{txpx}^{e}=-1\right)\right\} ,\\
\mathsf{M_{3}} & = & \left\{ \sigma_{pxpxy}^{m}\sigma_{typx}^{m}=-1\wedge\left(\sigma_{px}^{e}=-1\vee\sigma_{txpx}^{e}=-1\right)\right\} .
\end{eqnarray*}
Here $\wedge$, $\vee$ are the logical symbols for ``and'' and
``or'' respectively.  

This leaves 95 TC symmetry classes not ruled out by the above constraints, corresponding to 82 symmetry classes under $e \leftrightarrow m$ relabeling.  In addition, all these 95 TC symmetry classes are realized by models in $TC_0(G)$. \end{thm}
\begin{proof}
The unrealizability of $\mathsf{M}$ is a restatement of Lemma~\ref{clm:mpx}. 

To prove the unrealizability of $\mathsf{A}$ and $\mathsf{M_{1}}$,
 suppose $\sigma_{pxpxy}^{m}=-1$, then Lemma~\ref{lm:mpxpxy}
implies that there is $v\in V$  fixed under $P_{x}$, $P_{xy}$ and
hence fixed under $R$. So $\sigma_{pxpxy}^{e}=\sigma_{px}^{e}=\sigma_{pxy}^{e}=1$
by Lemma~\ref{lm:e_inv} and hence the TC symmetry classes in $\mathsf{A}$
and $\mathsf{M_{1}}$ are unrealizable. The unrealizability of 
$\mathsf{B}$ and $\mathsf{M_{2}}$ follows from an almost identical argument, using
Lemma~\ref{lm:mpxpxymtxty} and Lemma~\ref{lm:e_inv}.

To prove the unrealizability of $\mathsf{C}$ and $\mathsf{M_{3}}$,
suppose $\sigma_{pxpxy}^{m}\sigma_{typx}^{m}=-1$. Then Lemma~\ref{lm:mpxpxytypx}
implies there exists $v_{0}\in V$  fixed under $P_{x}$, $T_{y}P_{y}$.  It follows that $P_{xy} v_0$ is a vertex
fixed under $T_x P_x$.  Therefore, $\sigma_{px}^{e}=\sigma_{txpx}^{e}=1$
by Lemma~\ref{lm:e_inv}. So the TC symmetry classes in $\mathsf{M_{3}}$
are not realizable in $TC_0(G)$. Further, for
the unrealizability of $\mathsf{C}$, let $v_{j}=R^{j} v_{0} $
for $j=1,2,3$, pick $s_{0}\in W$ joining $v_{0}$ to $v_{1}$,
and let $s_{j}=R^{j} s_{0} $ for $j=1,2,3$. We choose $U_{P_{x}}^{e}\left(v_{0}\right)= 1$,
$U_{R}^{e}\left(v_{0}\right)= \mathcal{L}_{s_{0}}^{e}$ and $U_{T_{y}}^{e}\left(v_{2}\right)=\mathcal{L}_{s_{0}s_{1}}^{e}$.
Using Appendix~\ref{app:eloc}, we have $U_{R}^{e}\left(v_{j}\right)=  \mathcal{L}_{s_{j}}^{e}$
for $j=1,2,3$, and $U_{P_{x}}^{e}\left(v_{2}\right)=\mathcal{L}_{s_{0}s_{1}P_{x}\left(s_{1}s_{0}\right)}^{e}$.
Thus, 
\begin{eqnarray*}
\sigma^e_{pxpxy}=\left(U_{R}^{e}\right)^{4}\left(v_{0}\right) & = & \mathcal{L}_{s_{0}s_{1}s_{2}s_{3}}^{e},\\
\sigma^e_{typx} = U_{T_{y}}^{e}U_{P_{x}}^{e}\left(U_{T_{y}}^{e}\right)^{-1}\left(U_{P_{x}}^{e}\right)^{-1}\left(v_{2}\right) & = & \mathcal{L}_{s_{0}s_{1}P_{x}\left(s_{1}s_{0}\right)}^{e}.
\end{eqnarray*}
Therefore, 
\begin{eqnarray*}
\sigma_{typx}^{e}\sigma_{pxpxy}^{e} & = & \mathcal{L}_{s_{0}s_{1}P_{x}\left(s_{1}s_{0}\right)}^{e}\mathcal{L}_{s_{0}s_{1}s_{2}s_{3}}^{e}\\
 & = & \mathcal{L}_{s_{2}P_{x}s_{1}}^{e}\mathcal{L}_{s_{3}P_{x}s_{0}}^{e}\\
 & = & \mathcal{L}_{s_{2}P_{x}s_{1}}^{e}R\left(\mathcal{L}_{s_{2}P_{x}s_{1}}^{e}\right)=1 ,
\end{eqnarray*}
where the last equality holds because $s_2 P_x s_1$ (and hence also $s_3 P_x s_0$) is a closed path.
 In short, $\sigma_{typx}^{m}\sigma_{pxpxy}^{m}=-1$ implies $\sigma_{typx}^{e}\sigma_{pxpxy}^{e}=1$.
So the TC symmetry classes in $\mathsf{C}$ are not realizable in $TC_0(G)$.

The statements about counting and realization of symmetry classes are proved in Sec.~\ref{sub:tc_woso}.
\end{proof}

\subsection{Toric codes with spin-orbit coupling}
\label{app:tcsoc}

We now allow for spin-orbit coupling and consider models in the family $TC(G)$.  We prove the following Theorem, which was also stated in Sec.~\ref{sub:tc_so}:
 
\begin{thm}
\label{thm:sc_soc}The TC symmetry classes in $\mathsf{P_{1}}$, $\mathsf{P_{2}}$,
$\mathsf{P_{3}}$, $\mathsf{A}$, $\mathsf{B}$ and $\mathsf{C}'$
are not realizable in $TC\left(G\right)$, where
\begin{eqnarray*}
\mathsf{P_{1}} & = & \left\{ \sigma_{px}^{e}=\sigma_{px}^{m}=-1\right\} ,\\
\mathsf{P_{2}} & = & \left\{ \sigma_{pxy}^{e}=\sigma_{pxy}^{m}=-1\right\} ,\\
\mathsf{P_{3}} & = & \left\{ \sigma_{txpx}^{e}=\sigma_{txpx}^{m}=-1\right\} ,\\
\mathsf{A} & = & \left\{ \sigma_{pxpxy}^{e}=\sigma_{pxpxy}^{m}=-1\right\} ,\\
\mathsf{B} & = & \left\{ \sigma_{pxpxy}^{e}\sigma_{txty}^{e} = \sigma_{pxpxy}^{m}\sigma_{txty}^{m} = -1\right\} ,\\
\mathsf{C'} & = & \left\{ \sigma_{px}^{e}=\sigma_{txpx}^{e}=\sigma_{pxpxy}^{e}\sigma_{typx}^{e}=\sigma_{pxpxy}^{m}\sigma_{typx}^{m}=-1\right\} .
\end{eqnarray*}

This leaves 945 TC symmetry classes not ruled out by the above constraints, corresponding to 487 symmetry classes under $e \leftrightarrow m$ relabeling.  In addition, all these 945 TC symmetry classes are realized by models in $TC(G)$.
\end{thm}

The unrealizability in $TC(G)$ of the above TC symmetry classes is proved below in this Appendix.  Appendix~\ref{app:models} describes the counting of TC symmetry classes not ruled out by the Theorem, and gives  explicit examples of models for these classes.

\begin{lem}
\label{lm:cp}  If $g^2 = 1$, then $c_{\ell}^{\mu}\left(g\right)=c_{g \ell}^{\mu}\left(g\right),\forall \ell\in E$.\end{lem}
\begin{proof}
By Eq.~(\ref{eqn:c-restriction}), if $g^{2}=1$, then for all $\ell\in E$ 
\[
c_{\ell}^{\mu}\left(g\right)c_{g \ell}^{\mu}\left(g\right)=c_{\ell}^{\mu}\left(g^{2}\right)=c_{\ell}^{\mu}\left(1\right)=1.
\]
Hence $c_{\ell}^{\mu}\left(g\right)=c_{g\ell}^{\mu}\left(g\right)$.
\end{proof}

As discussed in Sec.~\ref{sec:genlatt}, we can redefine the local axes for each spin such that $c_{\ell}^{\mu}\left(T\right)=1$, for $T \in G$ any translation.  We always work in such a gauge.

\begin{lem}
\label{lm:cT}For any translations $T,T_{1},T_{2}\in G$, and for all $\ell \in E$, $g \in G$, we have $c_{T \ell}^{\mu}\left(T_{1}gT_{2}\right)=c_{\mu}^{x,z}\left(g\right)$.\end{lem}
\begin{proof}
We have
\begin{eqnarray*}
c_{T\ell}^{\mu}\left(T_{1}gT_{2}\right) &=& c_{\ell}^{\mu}\left(T\right)c_{T\ell}^{\mu}\left(T_{1}gT_{2}\right)=c_{\ell}^{\mu}\left(T_{1}gT_{2}T\right) \\ &=& c_{\ell}^{\mu}\left(T'g\right)=c_{\ell}^{\mu}\left(g\right)c_{g\ell}^{\mu}\left(T'\right)=c_{\ell}^{\mu}\left(g\right) \text{.}
\end{eqnarray*}
Here, we have used the fact that there is a translation
$T'\in G$ such that $T_{1}gT_{2}T=T'g$, which follows from the fact that
translations are a normal subgroup of $G$ (so, in particular, $g^{-1} T_1 g$ is a translation).
\end{proof}

To proceed, we need to consider further gauge fixing of $c_{\ell}^{z,x}\left(g\right)$ by choosing the local 
frame of spins.  By Lemma~\ref{lm:cT}, it is sufficient to restrict to $g \in G_o$, where $G_{o}=\left\{ g\in G|go=o\right\} $ with $o=\left(0,0\right)$
the origin. Let $\boell$ be the orbit of some $\ell \in E$ under
translations.  By Lemma~\ref{lm:cT}, we can write $c_{\boell}^{\mu}\left(g\right)=c_{\ell}^{\mu}\left(g\right)$, for all $g\in G_{o}$.
Gauge transformations $\gamma^{\mu}_{\ell}$ that are constant on translation orbits $\boell$ do not affect the choice $c^{\mu}_\ell(T) = 1$ for translations $T$.  Therefore, it is natural to think of the allowed gauge transformations as functions of $\boell$, and write $\gamma^{\mu}_{\boell}$ instead of $\gamma^{\mu}_{\ell}$.  The gauge transformation Eq.~(\ref{eqn:gauge-transformation}) then becomes
\begin{equation}
c^{\mu}_{\boell}(g) \to \gamma^{\mu}_{\boell} \gamma^{\mu}_{g \boell} c^{\mu}_{\boell} (g) \text{.}
\end{equation}

Now, consider some fixed translation orbit $\boell_0$. Let $G_{oR}$ be the rotation subgroup of $G_o$. Denote the orbit of $\boell_0$ under $G_o$ by $G_o \boell_0$, and the orbit of $\boell_0$ under rotations by $G_{oR} \boell_0$.  Then $\left| G_{oR} \boell_0 \right|=4,2,1$
and $\left| G_o \boell_0 \right|=8,4,2,1$.  We have the following possibilities:
\begin{enumerate}
\item $\left| G_o \boell_0 \right|=8$.  In this case, elements $\boell \in G_o \boell_0$ are in one-to-one correspondence with group elements $g \in G_o$.  That is, for each $\boell \in G_o \boell_0$, we can write uniquely $\boell = g \boell_0$ for some $g \in G_o$.  We make a gauge transformation by choosing
\begin{equation}
\gamma^{\mu}_{\boell} = \gamma^{\mu}_{g \boell_0} = c^{\mu}_{\boell_0}(g) \text{.}
\end{equation}
Then, in the transformed gauge, $c^{\mu}_{\boell_0}(g) \to 1$ for all $g \in G_o$, by construction.  We now consider $c^{\mu}_{\boell}(g)$ for arbitrary $\boell \in G_o \boell_0$, $g \in G_o$, in the transformed gauge.  We can write $\boell = g_1 \boell_0$ for some unique $g_1 \in G_o$, and
\begin{eqnarray*}
c^{\mu}_{\boell}(g) &=& c^{\mu}_{g_1 \boell_0}(g) = c^{\mu}_{g_1 \boell_0}(g) c^{\mu}_{\boell_0}(g_1)  \\
&=& c^{\mu}_{\boell_0}(g g_1) = 1 \text{.}
\end{eqnarray*}
Therefore, we are free to choose a gauge where $c^{\mu}_{\boell}(g) = 1$.  In particular, we have shown $c^{\mu}_{R^2 \boell}(g) = c^{\mu}_{\boell}(g)$.

\item $\left| G_o \boell_0 \right|=4$ and $\left| G_{oR} \boell_0 \right|=4$.  In this case, elements $\boell \in G_o \boell_0 = G_{oR} \boell_0$ are in one-to-one correspondence with $g_R \in G_{oR}$.  Therefore, the same argument given in the previous case implies we can choose a gauge so that $c^z_{\boell}(g_R) = 1$ for all $\boell \in G_{o} \boell_0$ and all $g_R \in G_{oR}$.  Now, for arbitrary $g \in G_o$, we consider
\begin{eqnarray*}
c^{\mu}_{R^2 \boell}(g) &=& c^{\mu}_{R^2 \boell}(g) c^{\mu}_{\boell}(R^2) 
= c^{\mu}_{\boell}(g R^2) \\
&=& c^{\mu}_{\boell}(R^2 g)
= c^{\mu}_{\boell}(g) c^{\mu}_{g \boell}(R^2) = c^{\mu}_{\boell}(g) \text{.}
\end{eqnarray*}
Therefore, we have also chosen a gauge in this case where $c^{\mu}_{R^2 \boell}(g) = c^{\mu}_{\boell}(g)$.

\item $|G_{oR} \boell_0| < 4$.  In this case, $R^2 \boell = \boell$ for all $\boell \in L$.  Therefore it holds trivially that $c^{\mu}_{R^2 \boell}(g) = c^{\mu}_{\boell}(g)$.
\end{enumerate}

We have thus shown the following fact, which will be useful in later calculations:
\begin{lem}
\label{lm:R2_gauge}It is possible to choose a local spin frame so that $c_{R^{2}\ell}^{\mu}\left(g\right)=c_{\ell}^{\mu}\left(g\right)$, for all $\ell \in E$ and $g \in G$, with $\mu = x,z$.
\end{lem}
To be more concrete, below, we always work in a local spin frame such that $\forall \mu=x,z$,
\begin{eqnarray}
c_{\ell}^{\mu}\left(T\right) & = & 1,\text{ for any translation }T,\label{eq:gauge1}\\
c_{\ell}^{\mu}\left(g\right) & = & 1,\forall g\in G_{o},\text{ if }\left|G_{o}\boldsymbol{\ell}\right|=8,\label{eq:gauge2}\\
c_{\ell}^{\mu}\left(g\right) & = & 1,\forall g\in G_{oR},\text{ if }\left|G_{oR}\boldsymbol{\ell}\right|=4,\label{eq:gauge3}
\end{eqnarray}
and hence Lemma~\ref{lm:R2_gauge} can be applied.

\begin{prop}
No TC symmetry classes in $\mathsf{P_1}$, $\mathsf{P}_2$ or $\mathsf{P}_3$ are realizable in $TC(G)$.
\end{prop}
\begin{proof}
We define
\begin{eqnarray*}
E_{0}^{px}&=&\left\{ \ell \in E   |   P_{x}\ell = \ell \text{ with ends of } \ell \text{ fixed}\right\} \text{,} \\
E_{1}^{px}&=& \left\{ \ell \in E | P_{x}\ell =\ell \text{ with ends of } \ell \text{ interchanged}\right\} \text{.}
\end{eqnarray*}
Then $E-\left(E_{0}^{px}\cup E_{1}^{px}\right)$ can be partitioned
into pairs $\left\{ \ell,P_{x} \ell\right\} $. Let $E_{2}^{px}$ be a set
formed by selecting one edge from each such pair. Now, we put
a $m$ particle at $h_{0}$ and draw a cut $t\in\bar{W}$ joining
$h_{0}$ with $h_{1}=P_{x}h_{0}$ such that $P_{x}t=t$. Then we choose
$U_{P_{x}}^{m}\left(h_{0}\right)  =  \mathcal{L}_{t}^{m}$, and by Appendix~\ref{app:eloc} we
may further choose
\begin{equation*}
U_{P_{x}}^{e}\left(h_{1}\right)  =  \mathcal{L}_{t}^{m} \mathcal{L}_{t}^{m}P_{x}\left(\mathcal{L}_{t}^{m}\right)=P_{x}\left(\mathcal{L}_{t}^{m}\right) \text{.}
\end{equation*}
Therefore,
\begin{eqnarray*}
\left(U_{P_{x}}^{m}\right)^{2}\left(h_{0}\right)  &=&  U_{P_{x}}^{m}\left(h_{1}\right)U_{P_{x}}^{m}\left(h_{0}\right) = P_{x}\left(\mathcal{L}_{t}^{m}\right)\mathcal{L}_{t}^{m} \\
&=& \prod_{\ell \in E} [c^x_{\ell}(P_x)]^{| \ell \cap t | } = \sigma^m_{px} \text{.}
\end{eqnarray*}

Since $\left|P_{x}\ell \cap t\right|=\left|\ell\cap P_{x}t\right|=\left|\ell\cap t\right|$,
we have
\begin{eqnarray*}
\sigma^m_{px} &=& \prod_{\ell\in E_{0}^{px}\cup E_{1}^{px}}\left[c_{\ell}^{x}(P_x) \right]^{\left| \ell\cap t\right|}\prod_{\ell\in E_{2}^{px}}\left[c_{\ell}^{x}(P_x) c_{P_{x}\ell}^{x}(P_x) \right]^{\left|\ell\cap t\right|} \\
&=& \prod_{\ell\in E_{0}^{px}}\left[c_{\ell}^{x}(P_x) \right]^{\left| \ell\cap t\right|} \text{,}
\end{eqnarray*}
where we used the fact that $\left|\ell\cap t\right|$ is even for $\ell\in E_{1}^{px}$, and also $c_{\ell}^{x}(P_x) c_{P_{x}\ell}^{x}(P_x) = c^x_{\ell}(P_x^2) = c^x_{\ell}(1) = 1$.

So $\sigma_{px}^{m}=-1$ implies that there is $\ell\in E$ such that
$P_{x}\ell=\ell$ with its ends fixed, and therefore there is a vertex $v$ with $P_x v = v$. Hence $\sigma_{px}^{e}=1$ by Lemma~\ref{lm:e_inv}.

In short, $\sigma_{px}^{m}=-1$ implies $\sigma_{px}^{e}=1$ and hence
$\mathsf{P_{1}}$ is not realizable in $TC(G)$.  The same arguement applies to $\mathsf{P}_{2}$ and $\mathsf{P}_{3}$.
\end{proof}

\begin{lem}
\label{lm:r4}In the chosen gauge, for any $v \in V$, $a_{v}a_{Rv}a_{R^{2}v}a_{R^{3}v}=1$
and $a_{v}a_{P_{x}v}a_{P_{y}v}a_{R^{2}v}=1$.\end{lem}
\begin{proof}
First we show that $a_{v}a_{R^{2}v}=a_{Rv}a_{R^{3}v}$.
We have
\begin{eqnarray*}
R\left(A_{v}A_{R^{2}v}\right) & = & \Big[ \prod_{\ell\ni v}c_{\ell}^{x}\left(R\right) \Big] \Big[ \prod_{\ell\ni R^{2}v}c_{\ell}^{x}\left(R\right) \Big] A_{R v} A_{R^3 v} \\
 & = &
  \Big[ \prod_{\ell\ni v}c_{\ell}^{x}\left(R\right) \Big] \Big[ \prod_{\ell\ni v}c_{R^2 \ell}^{x}\left(R\right) \Big] A_{R v} A_{R^3 v}\\
 & = & A_{Rv}A_{R^{3}v} \text{,}
\end{eqnarray*}
where the last equality follows from Lemma~\ref{lm:R2_gauge}.
Because $R$ is a symmetry, this implies $a_{v}a_{R^{2}v}=a_{Rv}a_{R^{3}v}$, and hence $a_{v}a_{Rv}a_{R^{2}v}a_{R^{3}v}=1$.

Similarly, 
\begin{eqnarray*}
P_x \left(A_{v}A_{R^{2}v}\right) & = & \Big[ \prod_{\ell\ni v}c_{\ell}^{x}\left(P_x \right) \Big] \Big[ \prod_{\ell\ni R^{2}v}c_{\ell}^{x}\left(P_x \right) \Big] A_{P_x v} A_{P_y v} \\
 & = &
  \Big[ \prod_{\ell\ni v}c_{\ell}^{x}\left(P_x \right) \Big] \Big[ \prod_{\ell\ni v}c_{R^2 \ell}^{x}\left(P_x \right) \Big] A_{P_x v} A_{P_y v}\\
 & = & A_{P_x v}A_{P_y v} \text{.}
\end{eqnarray*}
Therefore, $a_{v}a_{R^{2}v}=a_{P_{x}v}a_{P_{y}v}$, and hence $a_{v}a_{P_{x}v}a_{P_{y}v}a_{R^{2}v}=1$.\end{proof}

\begin{lem}
\label{lm:mpxpxy_soc}For any model in $TC\left(G\right)$, we have
$\sigma_{pxpxy}^{m}=a_{\mathscr{P}^{-1}\left(o\right)}=a_{\Gamma\left(R^{2}\right)}$,
where $o=\left(0,0\right)$.\end{lem}
\begin{proof}
We repeat the first paragraph of the proof to Lemma~\ref{lm:mpxpxy}.
The last equality in Eq.~\ref{eq:mrotation} is no longer obvious; it still holds
 because $R\left(\mathcal{L}_{t_{0}}^{e}\mathcal{L}_{t_{2}}^{e}\right)=R\left(\mathcal{L}_{t_{0}}^{e}\mathcal{L}_{R^{2}t_{0}}^{e}\right)=\mathcal{L}_{t_{1}}^{e}\mathcal{L}_{t_{3}}^{e}$,
by Lemma~\ref{lm:R2_gauge}.  In addition, the argument given in the proof of Lemma~\ref{lm:mpxpxy} that $\prod_{i=0}^3 a_{R^i v} = 1$ for $v \in V_t$ (with $\mathscr{P}\left(v\right)\neq o$ or $R^{2}v\neq v$) is no longer correct.  Instead, this fact follows directly from Lemma~\ref{lm:r4}.\end{proof}

\begin{prop}
\label{propn:TCG_A} No TC symmetry classes in
$\mathsf{A}$ are realizable in $TC(G)$.\end{prop}
\begin{proof}
Suppose that $\sigma_{pxpxy}^{m}=-1$. Then Lemma~\ref{lm:mpxpxy_soc}
tells us that there exists $v$ such that $R^{2}v=v$.  But then we can choose $U^e_{R^2}(v) = 1$, implying
$\sigma_{pxpxy}^{e}= \left( U_{R^{2}}^{e}\right)^{2}(v)=1$.
\end{proof}

\begin{lem}
\label{lm:mpxtpxy_soc}
For any model in $TC\left(G\right)$, we have $\sigma_{pxpxy}^{m}\sigma_{txty}^{m}=a_{\mathscr{P}^{-1}\left(\widetilde{o}\right)}=a_{\Gamma\left(T_{x} T_y R^{2}\right)}$,
where $\widetilde{o}=\left(\frac{1}{2},\frac{1}{2}\right)$.\end{lem}
\begin{proof}
We repeat the proof of Lemma~\ref{lm:mpxpxy_soc}, replacing $P_x \to \widetilde{P_x}$ and using Lemma~\ref{lm:o_corner}.  We also use the fact that $(\widetilde{P_x} P_{xy})^2 = T_x T_y R^2$  It should be noted that Lemma~\ref{lm:R2_gauge} still holds upon replacing $R^2 \to T_x T_y R^2$, by Lemma~\ref{lm:cT}.
\end{proof}

\begin{prop}
No TC symmetry classes in $\mathsf{B}$ are realizable in $TC\left(G\right)$.\end{prop}
\begin{proof}
This follows by the same argument used to prove Prop.~\ref{propn:TCG_A}, using Lemma~\ref{lm:mpxtpxy_soc} in place of Lemma~\ref{lm:mpxpxy_soc}.
\end{proof}

\begin{lem}
\label{lm:mtypx_soc}For any model in $TC\left(G\right)$, if $\sigma_{px}^{e}=\sigma_{txpx}^{e}=-1$, then
$\sigma_{pxpxy}^{m}\sigma_{typx}^{m}=a_{\mathscr{P}^{-1}\left(\kappa\right)}=a_{\Gamma\left(T_{y}R^{2}\right)}$, where $\kappa=\left(0,\frac{1}{2}\right)$.\end{lem}
\begin{proof}
As shown in Fig.~\ref{fig:m_pxty}, choose $h_{0} \in H$ near the
$y$-axis, let $h_{1}=P_{x} h_{0}$, $h_{2}=T_{y} h_{0}$,
$h_{3}=T_{y} h_{1}$, $h_{4}=P_{y} h_{0}$,
$h_{5}=P_{y} h_{1}$.   Draw a simple cut $t\in\bar{W}$ joining $h_{0}$,
$h_{1}$, $h_{5}$, $h_{3}$, $h_{2}$, $h_{4}$, $h_{0}$ in turn.  We denote the part of $t$
joining two successive holes by, for example, $\widetilde{h_0 h_1}$, and that joining three successive holes by,
for example, $\widetilde{h_0 h_1 h_5} = \widetilde{h_0 h_1} \widetilde{h_1 h_5}$.  We let $t_0 = \widetilde{h_{0}h_{1}}$,
$t_1 = \widetilde{h_{0} h_4 h_{2}}$, $t_2 \widetilde{h_{1} h_5 h_{3}}$ and $t_3 = \widetilde{h_{2}h_{3}}$.
We choose $h_{0}$ and $t$ such that any vertices enclosed
by $t$ are located on the $y$-axis, and no vertex with $y$-coordinate greater than $1/2$ is enclosed by $t$.
Moreover, $t$ is constructed so that $T_{y}t_{0}=t_{3}$, $P_{x}t_{1}=t_{2}$,
$P_{y}\widetilde{h_{0}h_{4}}=\widetilde{h_{0}h_{4}}$, $T_{y}P_{y}\widetilde{h_{2}h_{4}}=\widetilde{h_{2}h_{4}}$.
Then we have 
\begin{equation}
\sigma_{typx}^{m}=\mathcal{L}_{t_{0}}^{m}\mathcal{L}_{t_{1}}^{m}T_{y}\left(\mathcal{L}_{t_{0}}^{m}\right)P_{x}\left(\mathcal{L}_{t_{1}}^{m}\right)\label{eq:mtypx-1} \text{,}
\end{equation}
using the same argument leading to Eq.~\ref{eq:mtypx} in the
proof of Lemma~\ref{lm:mpxpxytypx}. 

In our chosen gauge, $T_{y}\left(\mathcal{L}_{t_{0}}^{m}\right)=\mathcal{L}_{t_{3}}^{m}$.
We now prove $P_{x}\left(\mathcal{L}_{t_{1}}^{m}\right)=\mathcal{L}_{t_{2}}^{m}$
by showing $P_{x}\left(\mathcal{L}_{\widetilde{h_{0}h_{4}}}^{m}\right)=\mathcal{L}_{\widetilde{h_{1}h_{5}}}^{m}$
and $P_{x}\left(\mathcal{L}_{\widetilde{h_{4}h_{2}}}^{m}\right)=\mathcal{L}_{\widetilde{h_{5}h_{3}}}^{m}$. 

First, to show $P_{x}\left(\mathcal{L}_{\widetilde{h_{0}h_{4}}}^{m}\right)=\mathcal{L}_{\widetilde{h_{1}h_{5}}}^{m}$,
let 
\begin{eqnarray*}
E_{0}^{py} & = & \left\{ \ell \in E | P_{y} \ell=\ell\mbox{ with ends of } \ell \text{ fixed}\right\} ,\\
E_{1}^{py} & = & \left\{ \ell \in E | P_{y}\ell=\ell \mbox{ with ends of } \ell \text{ interchanged}\right\} .
\end{eqnarray*}
Then $E-\left(E_{0}^{py}\cup E_{1}^{py}\right)$ can be divided into
pairs $\left\{ \ell,P_{y}\ell\right\} $. Let $E_{2}^{py}$ be a set formed
by picking one edge from each such pair. Since $\left|P_{y}\ell \cap\widetilde{h_{0}h_{4}}\right|=\left|\ell\cap P_{y}\widetilde{h_{0}h_{4}}\right|=\left|\ell\cap\widetilde{h_{0}h_{4}}\right|$,
we have 
\[
\mathcal{L}_{\widetilde{h_{0}h_{4}}}^{m}=\prod_{\ell\in E_{0}^{py}\cup E_{1}^{py}}\left(\sigma_{\ell}^{x}\right)^{\left| \ell \cap\widetilde{h_{0}h_{4}}\right|}\prod_{\ell\in E_{2}^{py}}\left(\sigma_{\ell}^{x}\sigma_{P_{y}\ell}^{x}\right)^{\left| \ell\cap\widetilde{h_{0}h_{4}}\right|} \text{.}
\]
We notice $\sigma_{px}^{e}=-1$ implies there is no $v$ such that
$P_{x}v=v$, by Lemma~\ref{lm:e_inv}.
Hence there is no $v$ such that $P_{y}v=v$; otherwise $P_{xy}v$
is fixed under $P_{x}$. Thus, $E_{0}^{py}$ is empty. In addition,
$\left|\ell \cap\widetilde{h_{0}h_{4}}\right|$ is even for $\ell \in E_{1}^{py}$.  Finally,
\begin{eqnarray}
 &  & c_{\ell}^{x}\left(P_{x}\right)c_{P_{y}\ell}^{x}\left(P_{x}\right)=c_{\ell}^{x}\left(P_{x}\right)c_{R^{2}P_{y}\ell}^{x}\left(P_{x}\right)\nonumber \\
 & = & c_{\ell}^{x}\left(P_{x}\right)c_{P_{x}\ell}^{x}\left(P_{x}\right)=c_{\ell}^{x}\left(P_{x}^{2}\right)=c_{\ell}^{x}\left(1\right)=1,\label{eq:cc}
\end{eqnarray}
so we have 
\[
P_{x}\left(\mathcal{L}_{\widetilde{h_{0}h_{4}}}^{m}\right)=\mathcal{L}_{\widetilde{h_{1}h_{5}}}^{m}.
\]

Similarly, to show $P_{x}\left(\mathcal{L}_{\widetilde{h_{4}h_{2}}}^{m}\right)=\mathcal{L}_{\widetilde{h_{5}h_{3}}}^{m}$,
let 
\begin{eqnarray*}
E_{0}^{typy} & = & \left\{ \ell \in E | T_{y}P_{y}\ell=\ell \text{ with ends of } \ell \text{ fixed}\right\} ,\\
E_{1}^{typy} & = & \left\{ \ell \in E | P_{y}\ell=\ell \text{ with ends of } \ell \text{ interchanged}\right\} .
\end{eqnarray*}
Then $E-\left(E_{0}^{typy}\cup E_{1}^{typy}\right)$ can be divided
into pairs $\left\{ \ell,T_{y}P_{y}\ell\right\} $. Let $E_{2}^{typy}$
be a set formed by choosing one edge from each such pair, and let $E_{01}^{typy} =E_{0}^{typy}\cup E_{1}^{typy}$.  Then,
\[
\mathcal{L}_{\widetilde{h_{4}h_{2}}}^{m}=\prod_{\ell\in E_{01}^{typy}}\left(\sigma_{\ell}^{x}\right)^{\left|\ell\cap\widetilde{h_{4}h_{2}}\right|}\prod_{\ell\in E_{2}^{typy}}\left(\sigma_{\ell}^{x}\sigma_{T_{y}P_{y}\ell}^{x}\right)^{\left|\ell\cap\widetilde{h_{4}h_{2}}\right|}.
\]
We notice $\sigma_{txpx}^{e}=-1$ implies there is no $v$ such that
$T_{x}P_{x}v=v$, by Lemma~\ref{lm:e_inv}.
Hence there is no $v$
such that $T_{y}P_{y}v=v$; otherwise $P_{xy}v$ is fixed under $T_{x}P_{x}$.
Thus, $E_{0}^{typy}$ is empty. In addition, $\left|l\cap\widetilde{h_{4}h_{2}}\right|$ is even for $l\in E_{1}^{typy}$.
Finally,
\begin{equation*}
c_{\ell}^x(P_x) c^x_{T_y P_y \ell}(P_x) = c_{\ell}^x(P_x) c^x_{P_y \ell}(P_x) = 1 \text{,}
\end{equation*}
where the last equality was shown in Eq.~(\ref{eq:cc}).  
Therefore, we have 
\[
P_{x}\left(\mathcal{L}_{\widetilde{h_{4}h_{2}}}^{m}\right)=\mathcal{L}_{\widetilde{h_{5}h_{3}}}^{m}.
\]

Therefore, $P_{x}\left(\mathcal{L}_{t_{1}}^{m}\right)=\mathcal{L}_{t_{2}}^{m}$,
and hence $\sigma_{typx}^{m}=\mathcal{L}_{t}^{m}=a_{V_{t}}$. For $v \in V_t$, if $\mathscr{P}\left(v\right)\ne o, \kappa$,
then $v$, $P_{x}v$, $P_{y}v$, $P_{x}P_{y}v$ are four different
vertices in $V_t$.  This holds because $\sigma_{px}^{e}=-1$ requires $v\neq P_{x}v$ and
$P_{y}v\neq P_{x}P_{y}v$. Since $a_{v}a_{P_{x}v}a_{P_{y}v}a_{P_{x}P_{y}v}=1$
by Lemma~\ref{lm:r4}, we have $\sigma_{typx}^{m}=a_{V_{t}}=a_{\mathscr{P}^{-1}\left(o\right)}a_{\mathscr{P}^{-1}\left(\kappa\right)}.$
Hence, using Lemma~\ref{lm:mpxpxy_soc}, $\sigma_{pxpxy}^{m}\sigma_{typx}^{m}=a_{\mathscr{P}^{-1}\left(\kappa\right)}$.

Further, if $v\in\mathscr{P}^{-1}\left(\kappa\right)$ but $T_{y}R^{2}v\ne v$,
then $v$, $P_{x}v$, $T_{y}P_{y}v$ and $P_{x}T_{y}P_{y}v$ are distinct
vertices in $\mathscr{P}^{-1}\left(\kappa\right)$; $\sigma_{px}^{e}=-1$
requires that $v\neq P_{x}v$ and $\sigma_{txpx}^{e}=-1$ requires
that $v\neq T_{y}P_{y}v$. Using Lemma~\ref{lm:r4}, and the fact that $a_{T v} = a_v$ for any translation $T$, we have 
\[
a_{v}a_{P_{x}v}a_{T_{y}P_{y}v}a_{P_{x}T_{y}P_{y}v}=a_{v}a_{P_{x}v}a_{P_{y}v}a_{P_{x}P_{y}v}=1 \text{.}
\]
This implies that only those vertices $v \in \mathscr{P}^{-1}(\kappa)$ satisfying $v = T_y R^2 v$ give non-trivial contributions to $\sigma^m_{pxpxy} \sigma^m_{typx}$, and we have shown
\[
\sigma_{pxpxy}^{m}\sigma_{typx}^{m}=a_{\mathscr{P}^{-1}\left(\kappa\right)}=a_{\Gamma\left(T_{y}R^{2}\right)}.
\]
\end{proof}

\begin{lem}
\label{lm:invs}If $\sigma_{px}^{e}=\sigma_{txpx}^{e}=\sigma_{pxpxy}^{m}\sigma_{typx}^{m}=-1$,
then there exists $v\in V$ and $s=l_{1}l_{2}\cdots l_{q}\in W$ connecting
$v$, $v'=P_{x}v$ such that $T_{y}R^{2}v=v$ and $T_{y}R^{2}l=l$
with ends fixed for each edge $l$ in $s$.\end{lem}
\begin{proof}
By Lemma~\ref{lm:mtypx_soc}, $\sigma_{pxpxy}^{m}\sigma_{typx}^{m}=-1$
implies $\Gamma\left(T_{y}R^{2}\right)$ is non-empty. In addition,
$\sigma_{px}^{e}=-1$ implies that there is no $v\in V$ such that
$P_{x}v=v$. Let $J=\left\{ v\in\Gamma\left(T_{y}R^{2}\right)|a_{v}a_{P_{x}v}=-1\right\} $.
Then $J=\left\{ v_{1},v_{1}',v_{2},v_{2}',\cdots,v_{n},v_{n}'\right\} $
with $v_{i}'=P_{x}v_{i}$  for $i=1,2,\cdots,n$.
Here $n$ must be odd, since $-1=\sigma_{pxpxy}^{m}\sigma_{typx}^{m}=a_{\Gamma\left(T_{y}R^{2}\right)}=a_{J}=\left(-1\right)^{n}$.  In addition, $v'_i = T_y P_y v_i$, because $v'_i = P_x v_i = P_x T_y R^2 v_i = T_y P_y v_i$.

We consider the graph $\mathcal{G}_{0}=\left(Y,E_{0}\right)$, where
\begin{eqnarray*}
E_{0} & = & \left\{ \ell\in E|T_{y}R^{2}\ell=\ell\text{ with ends fixed},c^x_{\ell}\left(P_{x}\right)=-1\right\} ,\\
Y & = & \Gamma\left(T_{y}R^{2}\right).
\end{eqnarray*}
Let $E_{0}^{v}=\left\{ \ell\in E_{0}| \ell\ni v\right\} $. Now we show
that in $\mathcal{G}_{0}$, the degree of each vertex $v\in J$ is
odd, while the degree of $v\in Y-J$ is even. That is, $\left|E_{0}^{v}\right|$ is odd for $v\in J$ and
$\left|E_{0}^{v}\right|$ is even for $v\in Y-J$. To show this, we consider $v \in Y$ and
notice the following partition 
\[
\text{star}\left(v\right)=\Big[ \cup_{j}\left\{ \ell_{j},T_{y}R^{2}\ell_{j}\right\} \Big] \cup\text{star}_{T_{y}R^{2}} \left(v\right),
\]
where $\text{star}\left(v\right)=\left\{ \ell\in E| \ell\ni v\right\} $,
$j$ labels all distinct pairs $\left\{ \ell_{j},T_{y}R^{2}\ell_{j}\right\} $ for $\ell_j \in \operatorname{star}(v)$ with
$\ell_{j}\neq T_{y}R^{2}\ell_{j}$, and $\text{star}_{T_{y}R^{2}} \left(v\right)=\left\{ \ell\in\text{star}\left(v\right)|T_{y}R^{2}\ell=\ell\right\} $.
Then we have
\begin{eqnarray*}
 &  & P_{x}\left(A_{v}\right)=P_{x}\left(\prod_{\ell\in\text{star}\left(v\right)}\sigma_{\ell}^{x}\right)\\
 & = & \prod_{\ell\in\text{star}_{T_{y}R^{2}}\left(v\right)}c_{\ell}^{x}\left(P_{x}\right)\prod_{j}\left(c_{\ell_{j}}^{x}\left(P_{x}\right)c_{T_{y}R^{2}\ell_{j}}^{x}\left(P_{x}\right)\right)A_{P_{x}v}\\
 & = & \left(\prod_{\ell\in\text{star}_{T_{y}R^{2}}\left(v\right)}c_{\ell}^{x}\left(P_{x}\right)\right)A_{P_{x}v}\\
 & = & \left(-1\right)^{\left|E_{0}^{v}\right|}A_{P_{x}v} \text{,}
\end{eqnarray*}
where we used the fact that $c^x_{T_y R^2 \ell}(P_x) = c^x_{R^2 \ell}(P_x) = c^x_{\ell}(P_x)$.
It follows that $a_{v}a_{P_{x}v}=\left(-1\right)^{\left|E_{0}^{v}\right|}$.
So $\left|E_{0}^{v}\right|$ is odd for $v\in J$ and $\left|E_{0}^{v}\right|$
is even for $v\in Y-J$. 

Now we claim that there exists $v\in J$ and a path $s=\ell_{1}\ell_{2}\cdots \ell_{q}$
in $\mathcal{G}_{0}$ connecting $v$ with $P_{x}v$, which is a more detailed version of the result to be shown.
We prove this claim by contradiction, and assume there is no $v\in J$ such that $v$ and $P_{x}v$ are in the
same connected component of $\mathcal{G}_{0}$. Then, without loss of generality, we relabel pairs $v_i \leftrightarrow v'_i$, so that
each component has empty intersection with at least one of the sets $\{ v_1, \dots, v_n \}$ or $\{ v'_1, \dots, v'_n \}$.
Since $n$ is odd, there must then be at least one component of ${\cal G}_0$ containing an odd number of vertices in $J$.  This is a contradiction, since the number of vertices of odd degree is even in any graph.
\end{proof}
\begin{prop}
No TC symmetry classes in $\mathsf{C}'$
are realizable in $TC(G)$.\end{prop}
\begin{proof}
Assume $\sigma_{px}^{e}=\sigma_{txpx}^{e}=\sigma_{pxpxy}^{m}\sigma_{typx}^{m}=-1$.
Lemma~\ref{lm:invs} tells us that there exists $v\in V$ and $s_{\kappa}\in W$
connecting $v$, $v'=P_{x}v$ such that $T_{y}R^{2}v=v$ and $T_{y}R^{2}s_{\kappa}=s_{\kappa}$,
where the subscript $\kappa$ indicates that $s_{\kappa}$ connects
 vertices with $\mathscr{P}\left(v\right)=\mathscr{P}\left(v' \right) =\kappa \equiv\left(0,\frac{1}{2}\right)$.
Choose $s\in W$ joining $P_{x}v$ to $R v $. Let $v_{j}=R^{j}v$,
$v_{j}'=R^{j}P_{x}v$ for $j=1,2,3$. 

In order to compute $\sigma_{pxpxy}^{e}$, we follow Appendix~\ref{app:eloc} to choose
\begin{eqnarray*}
U_{R}^{e}\left(v\right) & = &  \mathcal{L}_{s_{\kappa}s}^{e},\\
U_{R}^{e}\left(v_{1}\right) & = & \mathcal{L}_{s_{\kappa}s}^{e}\mathcal{L}_{s_{\kappa}s}^{e}R\left(\mathcal{L}_{s_{\kappa}s}^{e}\right),\\
U_{R}^{e}\left(v_{2}\right) & = & \mathcal{L}_{s_{\kappa}sR\left(s_{\kappa}s\right)}^{e}\mathcal{L}_{s_{\kappa}s}^{e}R\left(\mathcal{L}_{s_{\kappa}sR\left(s_{\kappa}s\right)}^{e}\right),\\
U_{R}^{e}\left(v_{3}\right) & = & \mathcal{L}_{s_{\kappa}sR\left(s_{\kappa}s\right)R^{2}\left(s_{\kappa}s\right)}^{e}\mathcal{L}_{s_{\kappa}s}^{e}R\left(\mathcal{L}_{s_{\kappa}sR\left(s_{\kappa}s\right)R^{2}\left(s_{\kappa}s\right)}^{e}\right).
\end{eqnarray*}
We have
\begin{equation*}
\sigma^e_{pxpxy} = (U^e_R)^4(v) = \cL^e_{s_\kappa s R^2 ( s_\kappa s ) } R ( \cL^e_{s_\kappa s R^2 ( s_\kappa s ) } ) \text{.}
\end{equation*}
Noticing $R\left(\mathcal{L}_{s_{\kappa}s}^{e}\mathcal{L}_{R^{2}\left(s_{\kappa}s\right)}^{e}\right)=\mathcal{L}_{R\left(s_{\kappa}s\right)}^{e}\mathcal{L}_{R^{3}\left(s_{\kappa}s\right)}^{e}$
by Lemma~\ref{lm:R2_gauge}, we then have 
\[
\sigma_{pxpxy}^{e}=\mathcal{L}_{s_{\kappa}sR\left(s_{\kappa}s\right)R^{2}\left(s_{\kappa}s\right)R^{3}\left(s_{\kappa}s\right)}^{e}.
\]

To calculate $\sigma_{typx}^{e}$,
we choose 
\begin{eqnarray*}
U_{P_{x}}^{e}\left(v_{2}\right) & = & \mathcal{L}_{T_{y}^{-1}s_{\kappa}}^{e}= \mathcal{L}_{R^{2}s_{\kappa}}^{e},\\
U_{T_{y}}^{e}\left(v_{2}\right) & = & \mathcal{L}_{s_{\kappa}sR\left(s_{\kappa}s\right)}^{e}.
\end{eqnarray*}
Following Appendix~\ref{app:eloc}, we further choose
\begin{eqnarray*}
U_{P_{x}}^{e}\left(v\right) & = & \mathcal{L}_{s_{\kappa}sR\left(s_{\kappa}s\right)}^{e} \mathcal{L}_{R^{2}s_{\kappa}}^{e}P_{x}\left(\mathcal{L}_{s_{\kappa}sR\left(s_{\kappa}s\right)}^{e}\right),\\
U_{T_{y}}^{e}\left(v_{2}'\right) & = & \mathcal{L}_{R^{2}s_{\kappa}}^{e} \mathcal{L}_{s_{\kappa}sR\left(s_{\kappa}s\right)}^{e}T_{y}\left(\mathcal{L}_{R^{2}s_{\kappa}}^{e}\right)\\
 &  & \mathcal{L}_{R^{2}s_{\kappa}}^{e}\mathcal{L}_{s_{\kappa}sR\left(s_{\kappa}s\right)}^{e}\mathcal{L}_{s_{\kappa}}^{e}.
\end{eqnarray*}
Thus, 
\begin{eqnarray*}
\sigma_{typx}^{e} & = & U_{T_{y}}^{e}U_{P_{x}}^{e}\left(U_{T_{y}}^{e}\right)^{-1}\left(U_{P_{x}}^{e}\right)^{-1}\left(v'\right)\\
 & = & U_{T_{y}}^{e}\left(v_{2}'\right)U_{P_{x}}^{e}\left(v_{2}\right)\left(U_{T_{y}}^{e}\left(v_{2}\right)\right)^{-1}\left(U_{P_{x}}^{e}\left(v\right)\right)^{-1}\\
 & = & \mathcal{L}_{s_{\kappa}sR\left(s_{\kappa}s\right)}^{e}\mathcal{L}_{R^{2}s_{\kappa}}^{e}P_{x}\left(\mathcal{L}_{s_{\kappa}sR\left(s_{\kappa}s\right)}^{e}\right)\mathcal{L}_{s_{\kappa}}^{e}.
\end{eqnarray*}

Finally, we have 
\begin{equation}
\sigma_{pxpxy}^{e}\sigma_{typx}^{e}  =  \mathcal{L}_{R^{2}\left(s\right)R^{3}\left(s_{\kappa}s\right)}^{e}\mathcal{L}_{s_{\kappa}}^{e}P_{x}\left(\mathcal{L}_{s_{\kappa}sR\left(s_{\kappa}s\right)}^{e}\right). \label{eqn:tosimp}
\end{equation}
This can be simplified, first noting that $U^e_{P_x}(v') = \cL^e_{s_\kappa} U^e_{P_x}(v) P_x ( \cL^e_{s_\kappa} )$, and therefore
\begin{equation}
-1 = \sigma^e_{px} = U^e_{P_x}(v) U^e_{P_x}(v') = \mathcal{L}_{s_\kappa}^{e}P_{x}\left(\mathcal{L}_{s_\kappa}^{e}\right) \text{.} \label{eqn:simp1}
\end{equation}
In addition, we have $T_x P_x v_3 = v'_3$, so we choose
\begin{eqnarray*}
U^e_{T_x P_x}(v_3) &=& \cL^e_{R^3 s_{\kappa}} \\
U^e_{T_x P_x}(v_3) &=& \cL^e_{R^3 s_{\kappa}} U^e_{T_x P_x}(v_3) (T_x P_x)( \cL^e_{R^3 s_{\kappa}} ) \\
&=& \cL^e_{R^3 s_{\kappa}} U^e_{T_x P_x}(v_3) P_x ( \cL^e_{R s_{\kappa}} ) \text{.}
\end{eqnarray*}
Therefore,
\begin{equation}
-1 = \sigma^e_{txpx} = U^e_{T_x P_x}(v_3) U^e_{T_x P_x} (v'_3) = \cL^e_{R^3 s_{\kappa}} P_x ( \cL^e_{R s_{\kappa}} ) \text{.}
\label{eqn:simp2}
\end{equation}

Substituting  Eqs.~(\ref{eqn:simp1}) and~\ref{eqn:simp2}) into Eq.~(\ref{eqn:tosimp}), we have
\begin{eqnarray*}
 &  & \sigma_{pxpxy}^{e}\sigma_{typx}^{e}\\
 & = & c_{s}^{z}\left(P_{x}\right)c_{Rs}^{z}\left(P_{x}\right)\mathcal{L}_{R^{2}s}^{e}\mathcal{L}_{R^{3}s}^{e}\mathcal{L}_{P_{x}s}^{e}\mathcal{L}_{P_{x}Rs}^{e}\\
 & = & c_{s}^{z}\left(P_{x}\right)c_{Rs}^{z}\left(P_{x}\right)c_{P_{x}sR^{3}s}^{z}\left(R^{-1}\right)\mathcal{L}_{P_{x}sR^{3}s}^{e}R^{-1}\left(\mathcal{L}_{P_{x}sR^{3}s}^{e}\right)\\
 & = & c_{s}^{z}\left(P_{x}\right)c_{Rs}^{z}\left(P_{x}\right)c_{P_{x}sR^{3}s}^{z}\left(R^{-1}\right),
\end{eqnarray*}
where in the last line we used the fact that $P_x s R^3 s$ is a closed path.
Since $c^z_{s_1 s_2}(g) = c^z_{s_1}(g) c^z_{s_2}(g)$, we have
\begin{eqnarray*}
 & & \sigma_{pxpxy}^{e}\sigma_{typx}^{e} =  c_{s}^{z}\left(P_{x}\right)c_{Rs}^{z}\left(P_{x}\right)c_{P_{x}sR^{3}s}^{z}\left(R^{-1}\right) \\
 &=& c_{s}^{z}\left(P_{x}\right) c^z_{P_x s} (R^{-1}) c_{Rs}^{z}\left(P_{x}\right) c_{R^{3}s}^{z}\left(R^{-1}\right) \text{.}
 \end{eqnarray*}
 To simplify this further, we make repeated use of Eq.~(\ref{eqn:c-restriction}).  First, we note that
 \begin{equation*}
 c_{s}^{z}\left(P_{x}\right) c^z_{P_x s} (R^{-1}) = c_s^z (P_{xy}) \text{,}
 \end{equation*}
 and so
 \begin{equation}
 \sigma_{pxpxy}^{e}\sigma_{typx}^{e} = c_s^z (P_{xy}) c_{Rs}^{z}\left(P_{x}\right) c_{R^{3}s}^{z}\left(R^{-1}\right) \text{.}
 \label{eqn:int}
 \end{equation}
 Next, $c^z_{P_{xy} s} (P_x) c^z_{R s}(P_x) = c^z_{P_{xy} s} (P_x^2) = 1$, and so
 \begin{equation}
 c^z_{R s}(P_x) = c^z_{P_{xy} s}(P_x) \text{.} \label{eqn:int_sub1} \text{.}
 \end{equation}
 Moreover, $c^z_{R^2 s}(R) c^z_{R^3 s} (R^{-1}) = c^z_{R^3 s}(R^{-1} R) = 1$, so
 \begin{equation}
 c^z_{R^3 s}(R^{-1}) = c^z_{R^2 s}(R) \label{eqn:int_sub2} \text{.}
 \end{equation}
 Substituting Eqs.~(\ref{eqn:int_sub1}, \ref{eqn:int_sub2}) into Eq.~(\ref{eqn:int}), we have
 \begin{eqnarray*}
&&   \sigma_{pxpxy}^{e}\sigma_{typx}^{e} = c_s^z (P_{xy}) c^z_{P_{xy} s}(P_x)c^z_{R^2 s}(R) \\
&=& c^z_{s}(R) c^z_{R^2 s}(R) = 1 \text{,}
\end{eqnarray*}
where the last equality follows from Lemma~\ref{lm:R2_gauge}.
 
In short, $\sigma_{px}^{e}=\sigma_{txpx}^{e}=\sigma_{pxpxy}^{m}\sigma_{typx}^{m}=-1$
implies $\sigma_{pxpxy}^{e}\sigma_{typx}^{e}=1$. Therefore, no TC symmetry
classes in $\mathsf{C}'$ are realizable. 
\end{proof}

\section{Models in $TC\left(G\right)$}
\label{app:models}

\begin{figure*}[t]
\includegraphics[width=0.6\textwidth]{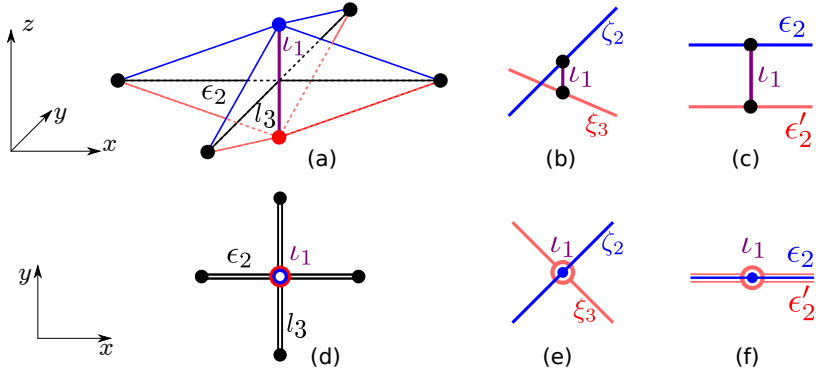}

\caption{(Color online) Depiction of the graphical notation used to represent stacking of vertices and edges.
The first row shows the connectivity of vertices and edges, and the second row gives the corresponding two-dimensional
presentation. It is convenient to imagine the graph of the lattice as first being embedded in three-dimensional space, and then projected into the two-dimensional plane. When these structures are present, we always assume top edges (blue online) are transformed to bottom edges (red online) under improper space group operations (\emph{i.e.} reflections), while translations do not swap edges with different colors.
Edges parallel to the $x$-axis, $y$-axis, $z$-axis
are labeled by symbols $\epsilon$, $l$, $\iota$, respectively.
We use $\zeta$ and $\xi$ to label diagonal edges.  For a diagonal edge, we can associate a unit vector $\hat{e}$ running along the direction of the edge, always choosing $\hat{e}_x > 0$.  Then $\zeta$ ($\xi$) is used to label edges with $\hat{e}_y > 0$ ($\hat{e}_y < 0$).
\emph{Panels (a,d).} This configuration  is only used in Fig.~\ref{fig:tcg1}c.
The two stacking vertices (blue and red online) together with
edge $\iota_{1}$ connecting them are projected into a point, presented
as a ring (blue and red online). Edges $\epsilon_{2}$, $l_{3}$
 pass through the ring but do not end on it. The triple-stacking edges
are presented as double lines. \emph{Panels (b,e).} A configuration with double-stacking
vertices and no stacking edges. We use a darker point
(blue online) to represent the upper vertex, and a lighter ring (red
online) to represent the lower vertex. The edges linked to the upper
vertex are darker (blue online) and the edges linked to
the lower vertex are lighter (red online). \emph{Panels (c,f).} A situation with
double-stacking vertices and edges. The vertices are represented as in (b,e).  The lower edge is represented by a lighter double line (red online),
and the upper edge is a single darker line (blue online) drawn in the center of the
double line.}

\label{fig:tcg-caption}
\end{figure*}

To complete the proof of Theorem~\ref{thm:sc_soc}, we need to give
explicit models in $TC\left(G\right)$ for the TC symmetry classes
that are not excluded by the theorem. These models are summarized
in Fig.~\ref{fig:tcg1}, Fig.~\ref{fig:tcg2} and Fig.~\ref{fig:tcg3}.
In some of the models, we use lattices with stacking of vertices and/or edges; that is, there
can be distinct edges or vertices with the same image under ${\mathscr P}$.
We use single
solid lines and points to present edges and vertices that do not
stack, while the meaning of other line and point types used is illustrated
in Fig.~\ref{fig:tcg-caption}. We use different letters
$l$, $\epsilon$, $\iota$, $\xi$, $\zeta$ to label edges with different direction,
as illustrated in Fig.~\ref{fig:tcg-caption}.   In particular, $l$ labels vertical edges, and $\epsilon$ horizontal edges, with $\xi$ and $\zeta$ indicating diagonal edges.  The symbol $\iota$ is reserved for edges that project to a single point under ${\mathscr P}$.

Following this discussion, it is easy but tedious to verify that all TC symmetry classes
not excluded by Theorem~\ref{thm:sc_soc} are realized by
the models in Fig.~\ref{fig:tcg1}, Fig.~\ref{fig:tcg2} and Fig.~\ref{fig:tcg3}.  

Finally, let's compute the total number of realizable symmetry classes.
Let $\mathsf{D}=\mathsf{P_{1}}\cup\mathsf{P_{2}}\cup\mathsf{P_{3}}\cup\mathsf{A}\cup\mathsf{B}$
be a subset of unrealizable TC symmetry classes, and let $\mathsf{T}$ be the set of all
TC symmetry classes.  From the form Eq.~(\ref{eq:sc_matrix}) and the definition of $\mathsf{D}$, it is apparent that 
\begin{eqnarray*}
\left| \mathsf{T} - \mathsf{D} \right| &=& 3^5 \times 4 = 972 .
\end{eqnarray*}
The TC symmetry classes in $\mathsf{C}' - \mathsf{C}' \cap \mathsf{D}$ are of the form
\[
\left(\begin{array}{cc}
\sigma_{px}^{e} & \sigma_{px}^{m}\\
\sigma_{pxy}^{e} & \sigma_{pxy}^{m}\\
\sigma_{txpx}^{e} & \sigma_{txpx}^{m}\\
\sigma_{pxpxy}^{e} & \sigma_{pxpxy}^{m}\\
\sigma_{txty}^{e}\sigma_{pxpxy}^{e} & \sigma_{txty}^{m}\sigma_{pxpxy}^{m}\\
\sigma_{pxpxy}^{e}\sigma_{typx}^{e} & \sigma_{pxpxy}^{m}\sigma_{typx}^{m}
\end{array}\right)=\left(\begin{array}{cc}
-1 & 1\\
\square & \square\\
-1 & 1\\
\square & \square\\
\square & \square\\
-1 & -1
\end{array}\right),
\]
with classes  in $\mathsf{P}_2,\mathsf{A},\mathsf{B}$ excluded, so
$\left|\mathsf{C}- \mathsf{C}' \cap \mathsf{D}\right|=3^{3}=27$. Thus, the total number
of TC symmetry classes realizable in $TC\left(G\right)$ is $\left|(\mathsf{T}-\mathsf{D})-(\mathsf{C}' - \mathsf{C}' \cap \mathsf{D} )\right|=972 - 27 =945$. 

To count realizable symmetry classes (as opposed to TC symmetry classes), we first count the number of symmetry classes obtained from $\mathsf{T} - \mathsf{D}$.  Under $e \leftrightarrow m$ relabeling, every TC symmetry class in $\mathsf{T} - \mathsf{D}$ either goes into itself, or goes into another TC symmetry class in $\mathsf{T} - \mathsf{D}$.  It is easy to see that only 2 TC symmetry classes in $\mathsf{T} - \mathsf{D}$ are invariant under $e \leftrightarrow m$ relabeling, so the number of distinct symmetry classes obtained from $\mathsf{T} - \mathsf{D}$ is $\frac{1}{2} ( 972 - 2) + 2 = 487$.

Now consider a TC symmetry class in $\mathsf{C}- \mathsf{C}' \cap \mathsf{D}$.  Under $e \leftrightarrow m$, we obtain   a TC symmetry class not contained in $\mathsf{C}- \mathsf{C}' \cap \mathsf{D}$, but which is contained in $\mathsf{T} - \mathsf{D}$.  Therefore the resulting TC symmetry class is realizable.  This means that removing $\mathsf{C}- \mathsf{C}' \cap \mathsf{D}$ from  $\mathsf{T} - \mathsf{D}$ does not reduce the number of symmetry classes, even though the number of TC symmetry classes is reduced.  The total number of realizable symmetry classes is thus 487.  This completes the proof of Theorem~\ref{thm:sc_soc}.

\begin{figure*}
\begin{minipage}[t]{0.33\textwidth}%
\includegraphics[width=0.9\columnwidth]{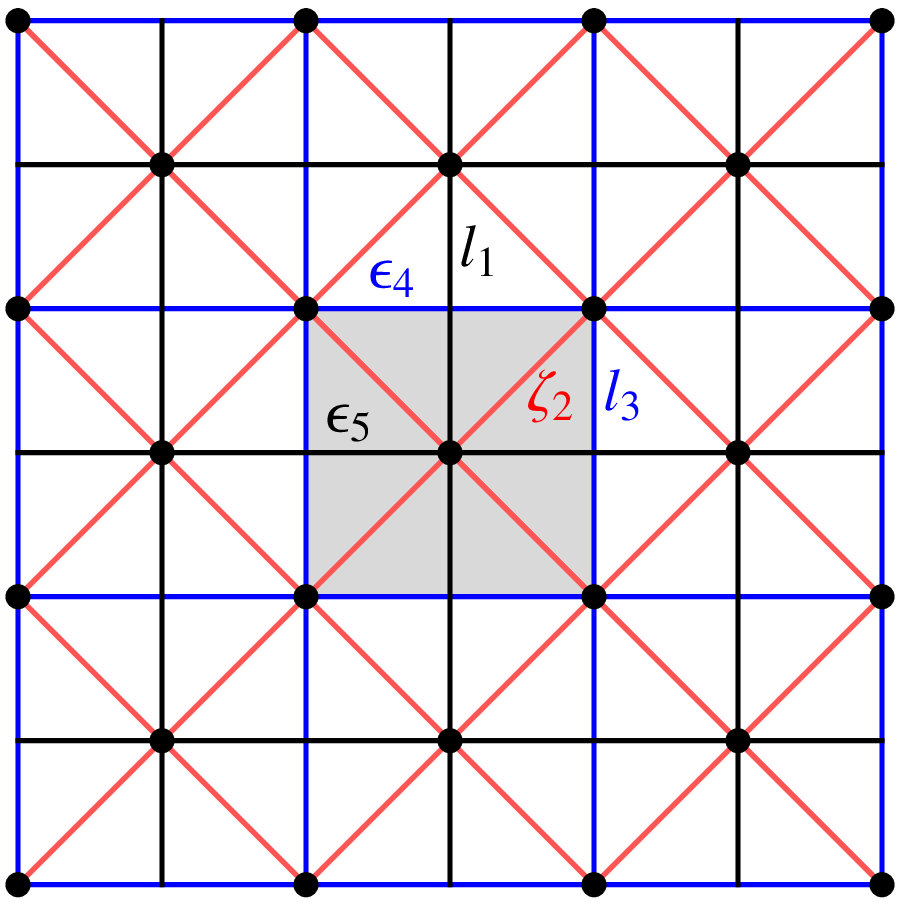}

(a) $\left(\begin{array}{cccccc}
1 & 1 & 1 & 1 & 1 & \gamma_{1}\\
\alpha_{1} & \beta_{2} & \alpha_{3} & e_{o} & e_{\tilde{o}} & \alpha_{4}\alpha_{5}
\end{array}\right)$%
\end{minipage}%
\begin{minipage}[t]{0.33\textwidth}%
\includegraphics[width=0.9\columnwidth]{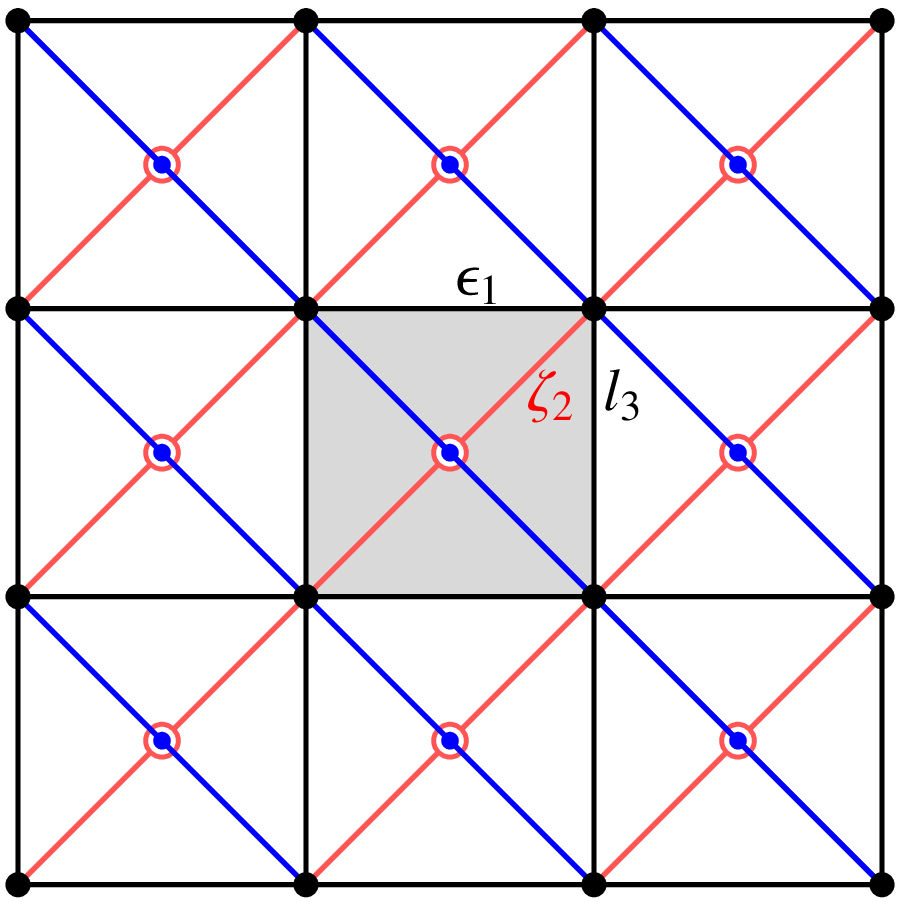}

(b) $\left(\begin{array}{cccccc}
\gamma_{o} & 1 & 1 & 1 & 1 & \gamma_{3}\\
1 & \alpha_{2} & \alpha_{3} & e_{o} & e_{\tilde{o}} & \alpha_{1}
\end{array}\right)$%
\end{minipage}%
\begin{minipage}[t]{0.33\textwidth}%
\includegraphics[width=0.9\columnwidth]{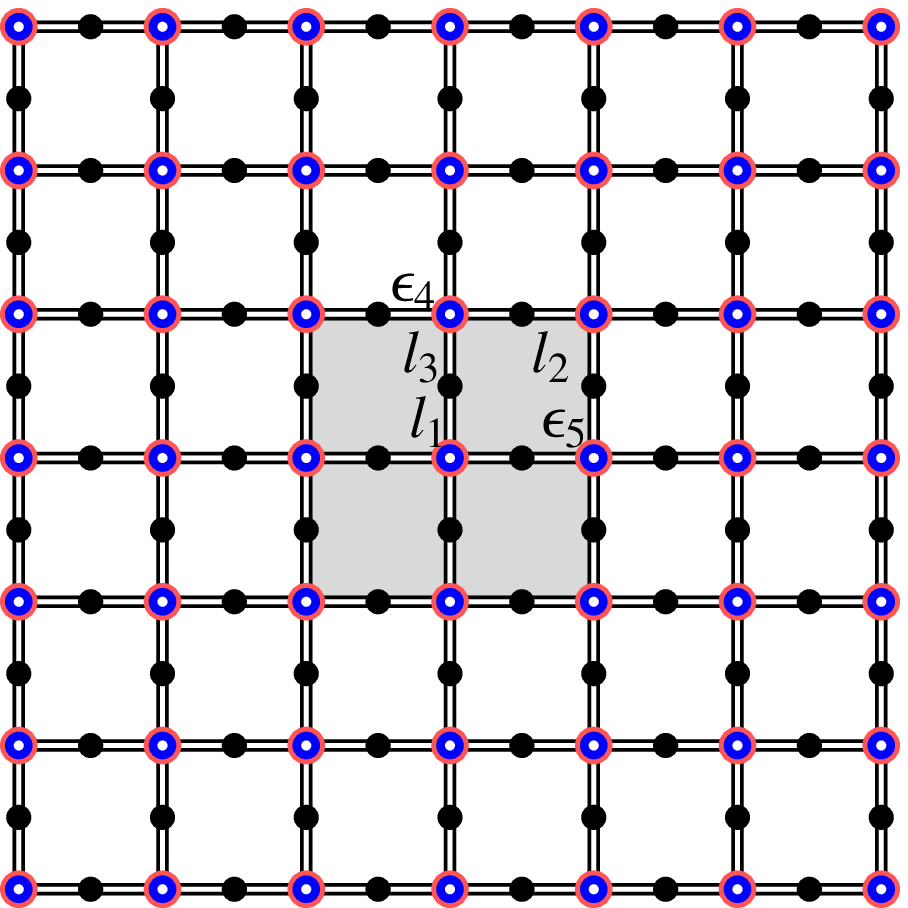}

(c) $\left(\begin{array}{cccccc}
1 & \delta_{o} & 1 & 1 & 1 & \gamma_{1}\gamma_{3}\\
\alpha_{1} & 1 & \alpha_{2} & e_{o} & e_{\tilde{o}} & \alpha_{4}\alpha_{5}e_{\kappa}
\end{array}\right)$%
\end{minipage}

\bigskip{}

\begin{minipage}[t]{0.33\textwidth}%
\includegraphics[width=0.9\columnwidth]{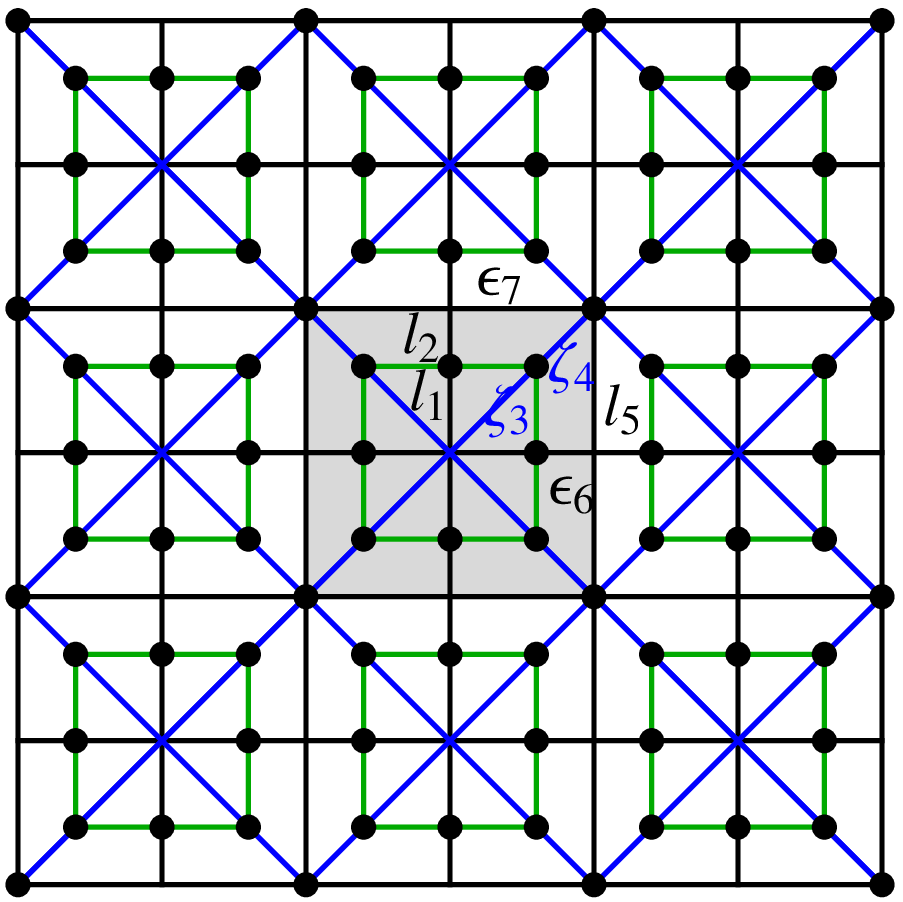}

(d) $\left(\begin{array}{cccccc}
1 & 1 & 1 & \gamma_{1} & 1 & \gamma_{2}\\
\alpha_{1} & \beta_{3} & \alpha_{5} & 1 & e_{\tilde{o}} & \alpha_{6}\alpha_{7}
\end{array}\right)$%
\end{minipage}%
\begin{minipage}[t]{0.33\textwidth}%
\includegraphics[width=0.9\columnwidth]{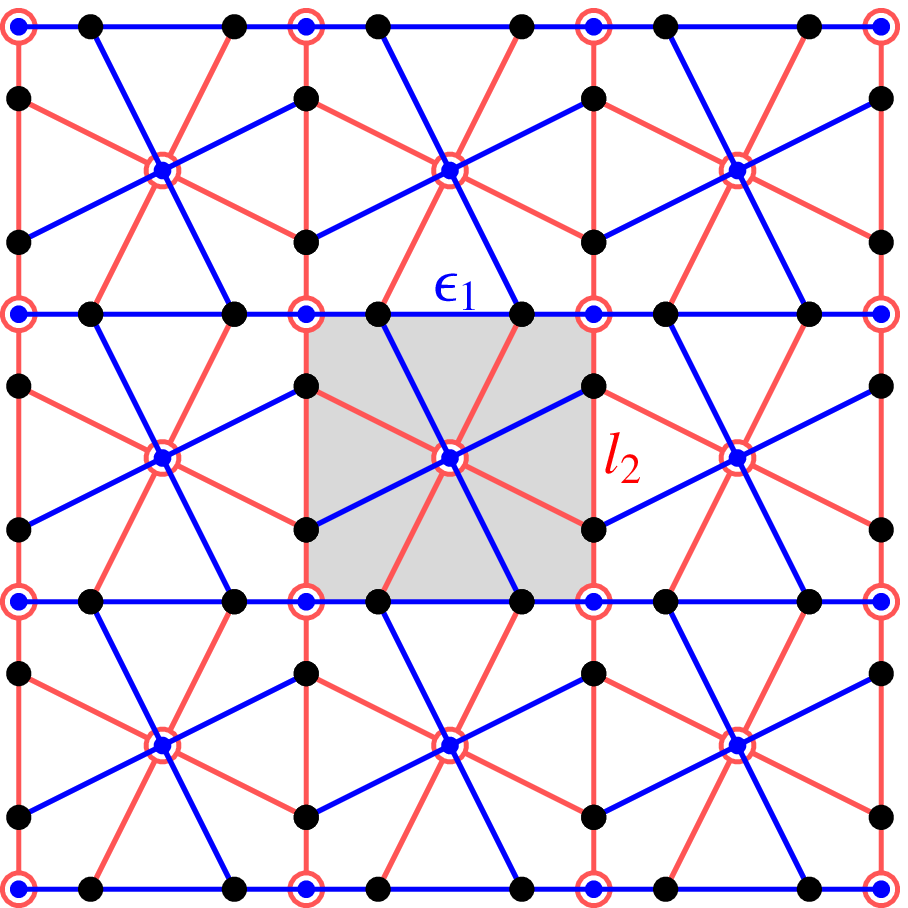}

(e) $\left(\begin{array}{cccccc}
\gamma_{o} & \delta_{o} & 1 & 1 & 1 & \gamma_{2}\\
1 & 1 & \alpha_{2} & e_{o} & e_{\tilde{o}} & \alpha_{1}
\end{array}\right)$%
\end{minipage}%
\begin{minipage}[t]{0.33\textwidth}%
\includegraphics[width=0.9\columnwidth]{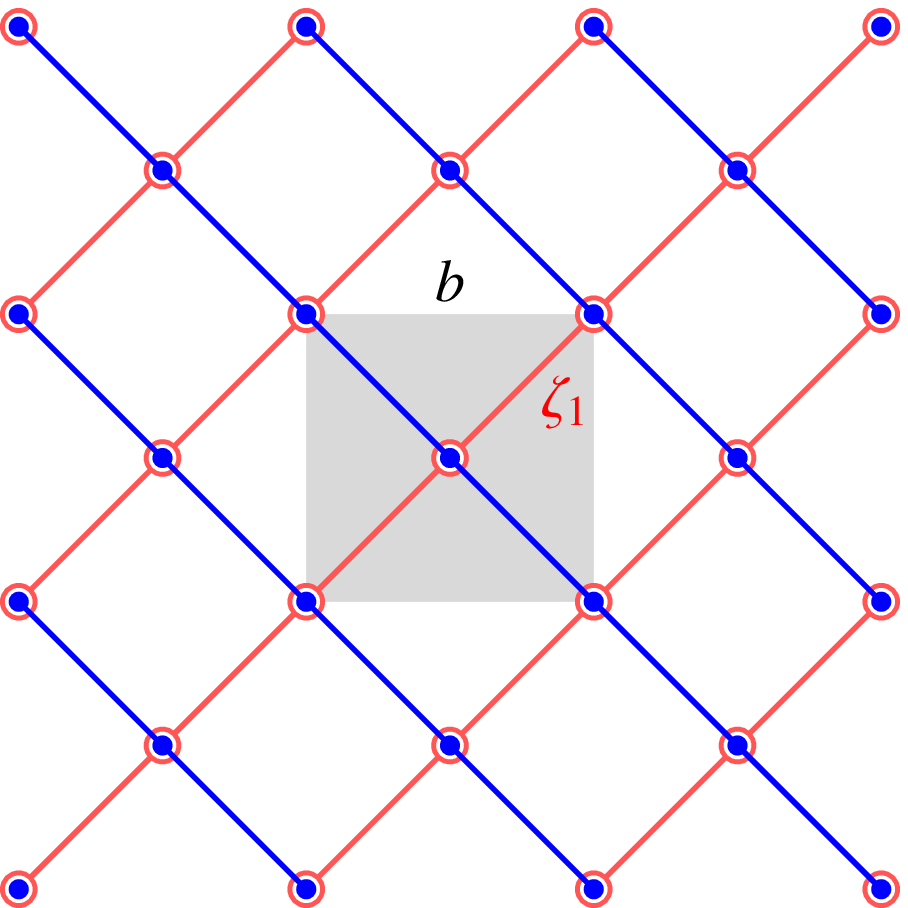}

(f) $\left(\begin{array}{cccccc}
\gamma_{o} & 1 & \gamma_{\tilde{o}} & 1 & 1 & b\gamma_{o}\gamma_{\tilde{o}}\\
1 & \beta_{1} & 1 & e_{o} & e_{\tilde{o}} & 1
\end{array}\right)$%
\end{minipage}\bigskip{}

\begin{minipage}[t]{0.33\textwidth}%
\includegraphics[width=0.9\columnwidth]{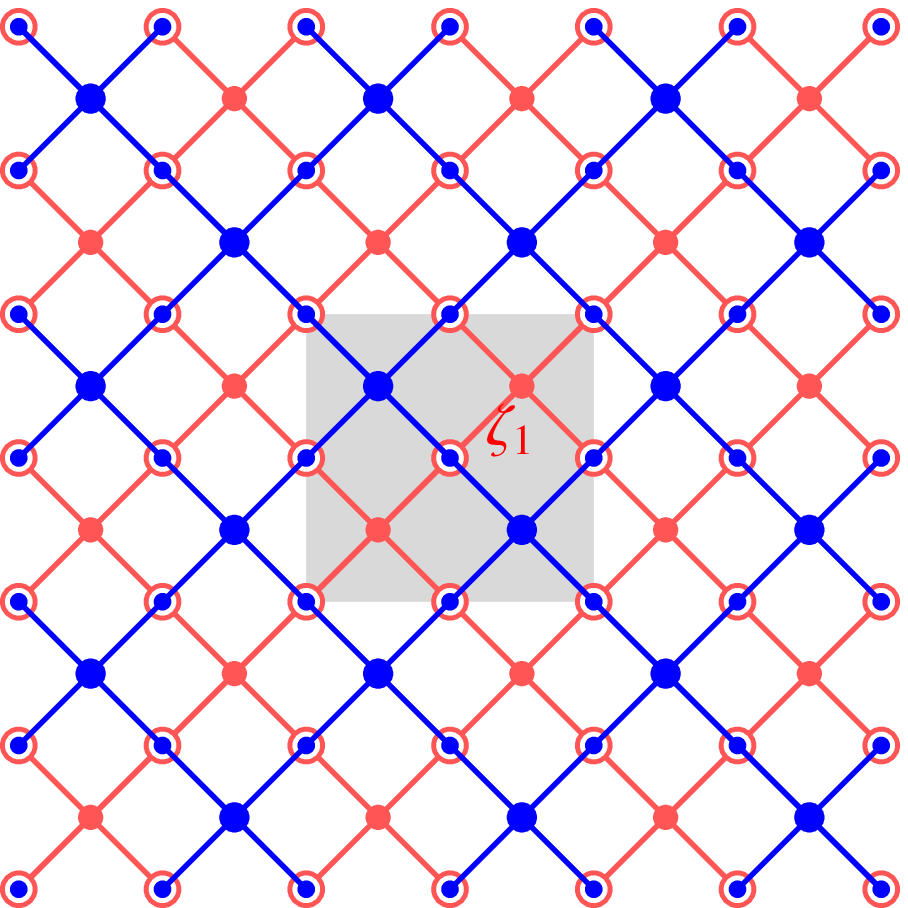}

(g) $\left(\begin{array}{cccccc}
\gamma_{o} & 1 & \gamma_{\tilde{o}} & 1 & 1 & \gamma_{o}\gamma_{\tilde{o}}\\
1 & \beta_{1} & 1 & e_{o} & e_{\tilde{o}} & e_{\kappa}
\end{array}\right)$%
\end{minipage}%
\begin{minipage}[t]{0.33\textwidth}%
\includegraphics[width=0.9\columnwidth]{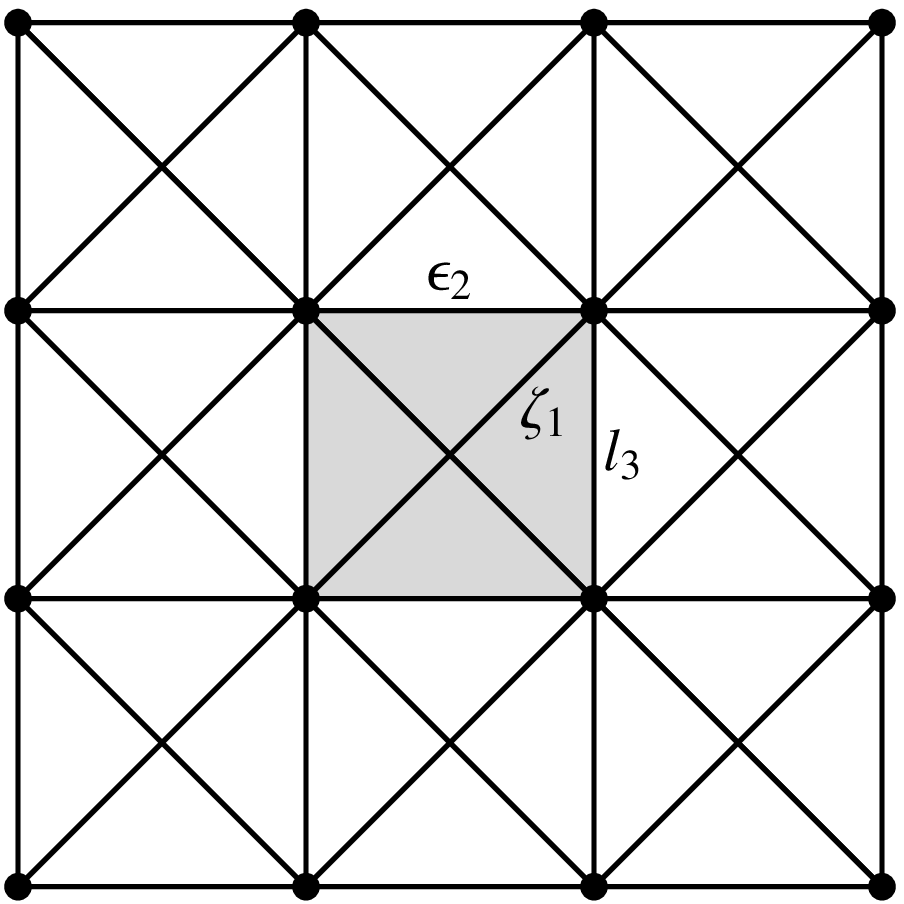}

(h) $\left(\begin{array}{cccccc}
\gamma_{2} & 1 & 1 & \delta_{1} & 1 & \gamma_{3}\\
1 & \beta_{1} & \alpha_{3} & 1 & e_{\tilde{o}} & \alpha_{2}
\end{array}\right)$%
\end{minipage}%
\begin{minipage}[t]{0.33\textwidth}%
\includegraphics[width=0.9\columnwidth]{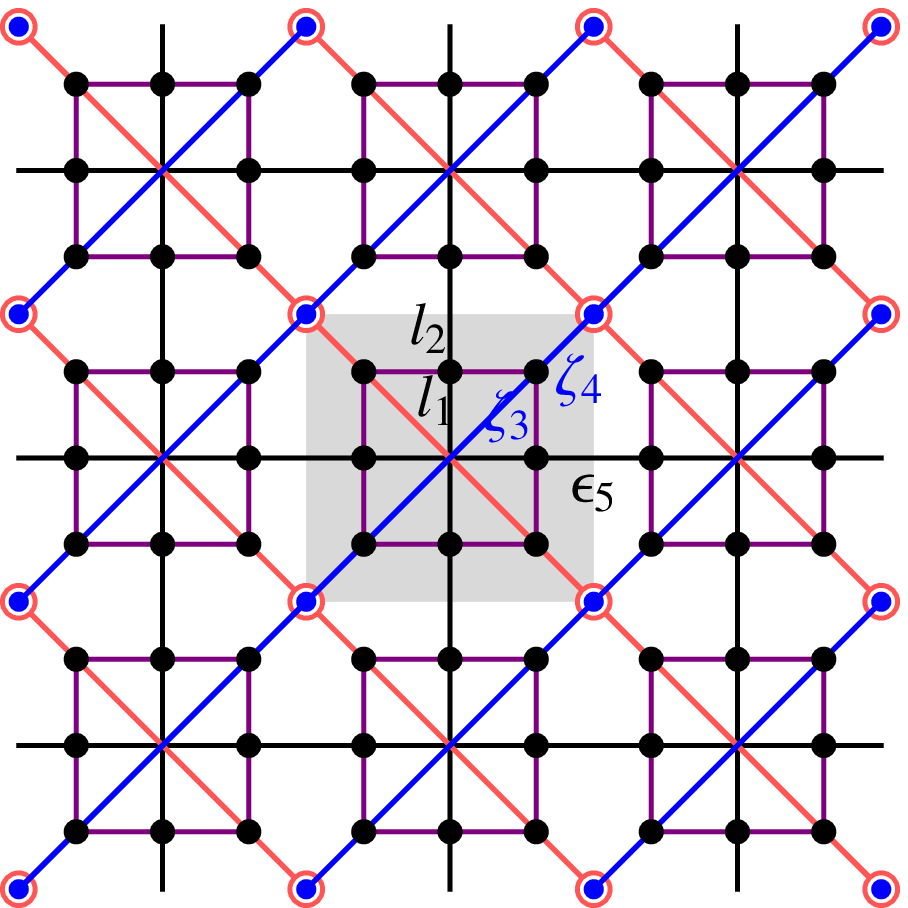}

(i) $\left(\begin{array}{cccccc}
1 & 1 & \gamma_{5} & \gamma_{1} & 1 & \gamma_{2}\\
\alpha_{1} & \beta_{3} & 1 & 1 & e_{\tilde{o}} & \alpha_{5}
\end{array}\right)$%
\end{minipage}

\caption{(Color online) $TC\left(G\right)$ models (Part I). The shaded square is a unit cell
and the TC symmetry classes are calculated with the origin $o$ at
the center of the shaded square. Below each  lattice is the
corresponding TC symmetry class in the form \eqref{eq:sc_matrix}.
The edges are labeled by different letters according to their directions
as described in the text and in Fig.~\ref{fig:tcg-caption}. Edges that map to a single 
point under ${\mathscr P}$ are labeled by $\iota_{o}$, $\iota_{\tilde{o}}$,
$\iota_{\kappa}$, $\iota_{\tilde{\kappa}}$ with the subscript indicating their position,
and $\tilde{o}=\left(\frac{1}{2},\frac{1}{2}\right)$, $\kappa=\left(0,\frac{1}{2}\right)$,
$\tilde{\kappa}=\left(\frac{1}{2},0\right)$, in units such that
the size of the unit cell is $1\times1$. For short, we define
$\alpha_{i}=c_{\varepsilon_{i}}^{x}\left(P_{x}\right)$, $\beta_{i}=c_{\varepsilon_{i}}^{x}\left(P_{xy}\right)$,
$\gamma_{i}=c_{\varepsilon_{i}}^{z}\left(P_{x}\right)$ and $\delta_{i}=c_{\varepsilon_{i}}^{z}\left(P_{xy}\right)$,
where $\varepsilon=l,\epsilon,\xi,\zeta,\iota$ stands for a generic
edge. In addition, $e_{r}=a_{\mathscr{P}^{-1}\left(r\right)}$, and
$b$ is the eigenvalue of $B_{p}$ for the plaquette (here meaning
smallest cycle) $p$ within which $b$ is written. The values of $e_{r}$ and $b$ are well-defined with respect to any local spin frame system satisfying Eqs.~(\ref{eq:gauge1}-\ref{eq:gauge3}). }

\label{fig:tcg1}
\end{figure*}

\begin{figure*}
\begin{minipage}[t]{0.33\textwidth}%
\includegraphics[width=0.9\columnwidth]{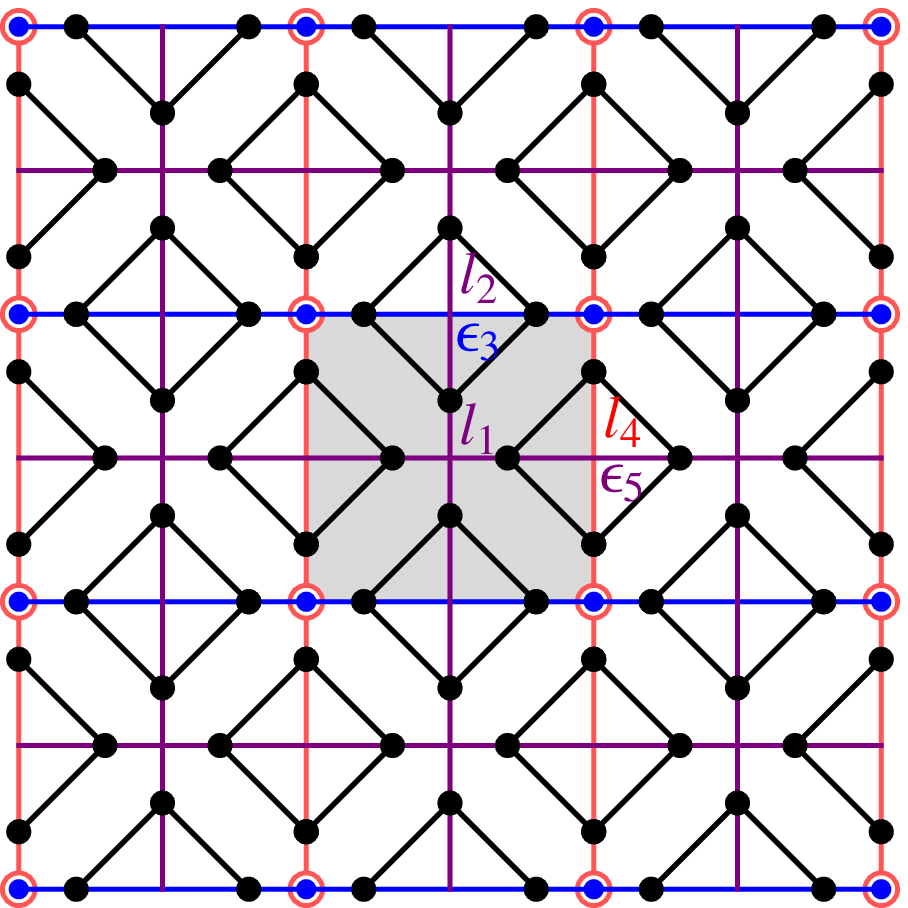}

(j) $\left(\begin{array}{cccccc}
1 & \delta_{\tilde{o}} & 1 & \gamma_{1} & 1 & \gamma_{2}\\
\alpha_{1} & 1 & \alpha_{4} & 1 & e_{\tilde{o}} & \alpha_{3}\alpha_{5}
\end{array}\right)$%
\end{minipage}%
\begin{minipage}[t]{0.33\textwidth}%
\includegraphics[width=0.9\columnwidth]{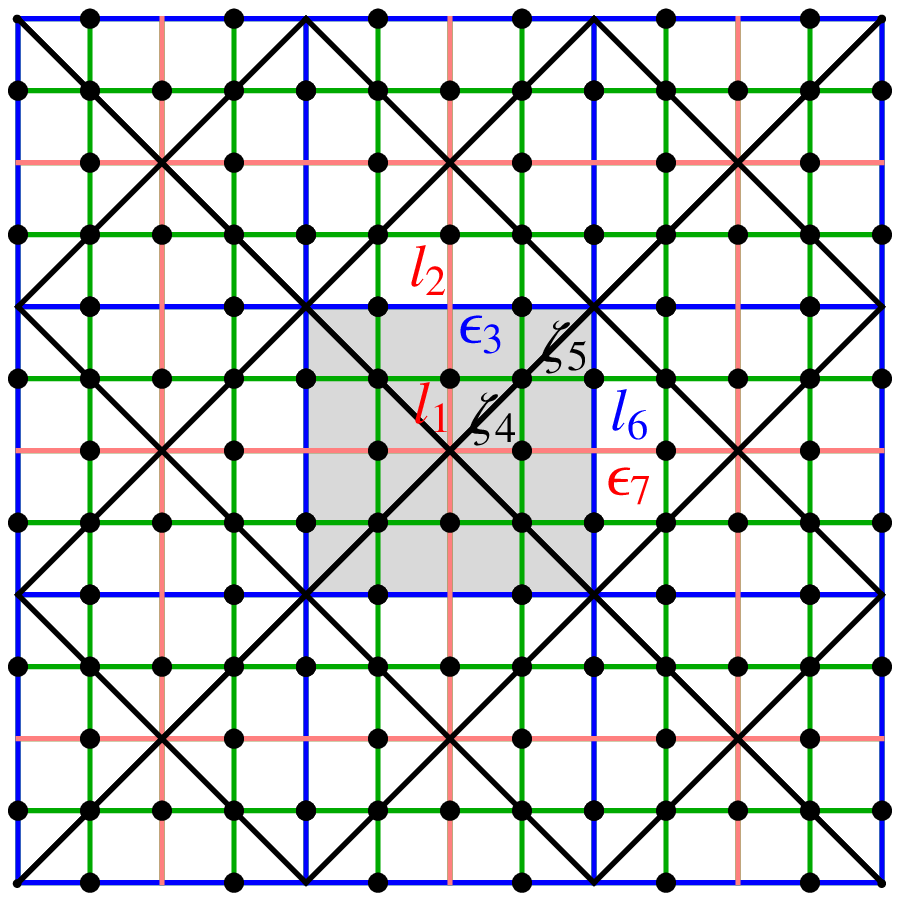}

(k) $\left(\begin{array}{cccccc}
1 & 1 & 1 & \gamma_{1} & \delta_{5} & \gamma_{2}\\
\alpha_{1} & \beta_{4} & \alpha_{6} & 1 & 1 & \alpha_{3}\alpha_{7}
\end{array}\right)$%
\end{minipage}%
\begin{minipage}[t]{0.33\textwidth}%
\includegraphics[width=0.9\columnwidth]{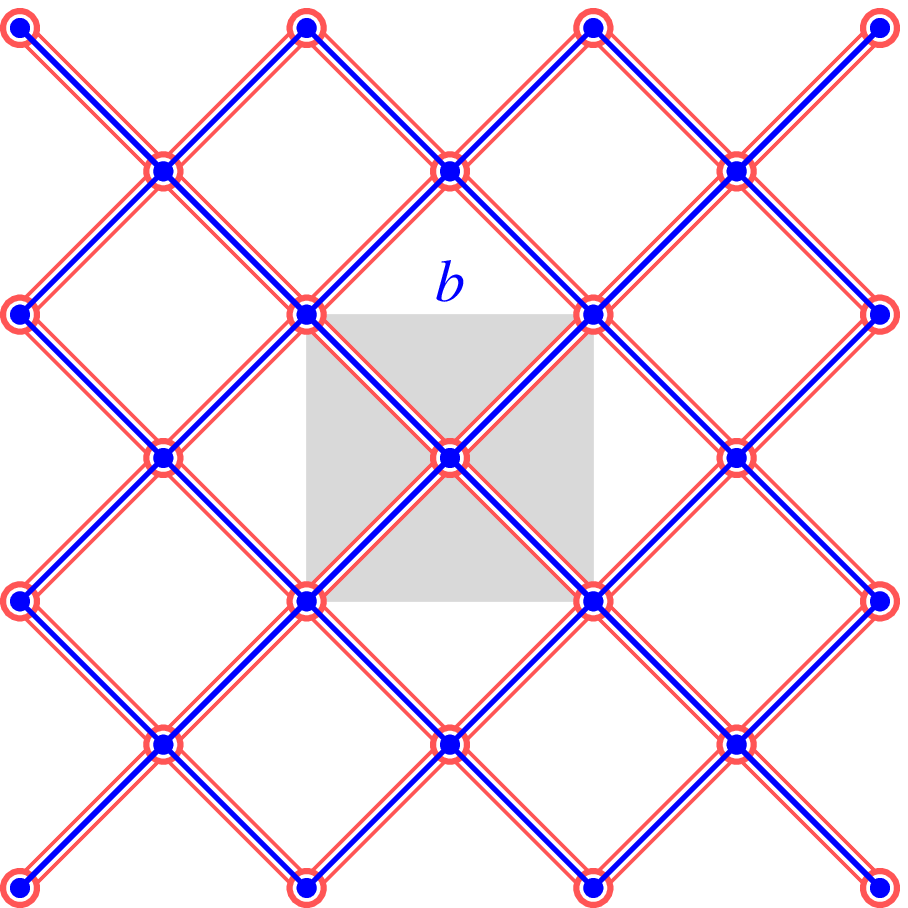}

(l) $\left(\begin{array}{cccccc}
\gamma_{o} & \delta_{o} & \gamma_{\tilde{o}} & 1 & 1 & b\gamma_{o}\gamma_{\tilde{o}}\\
1 & 1 & 1 & e_{o} & e_{\tilde{o}} & 1
\end{array}\right)$%
\end{minipage}

\bigskip{}

\begin{minipage}[t]{0.33\textwidth}%
\includegraphics[width=0.9\columnwidth]{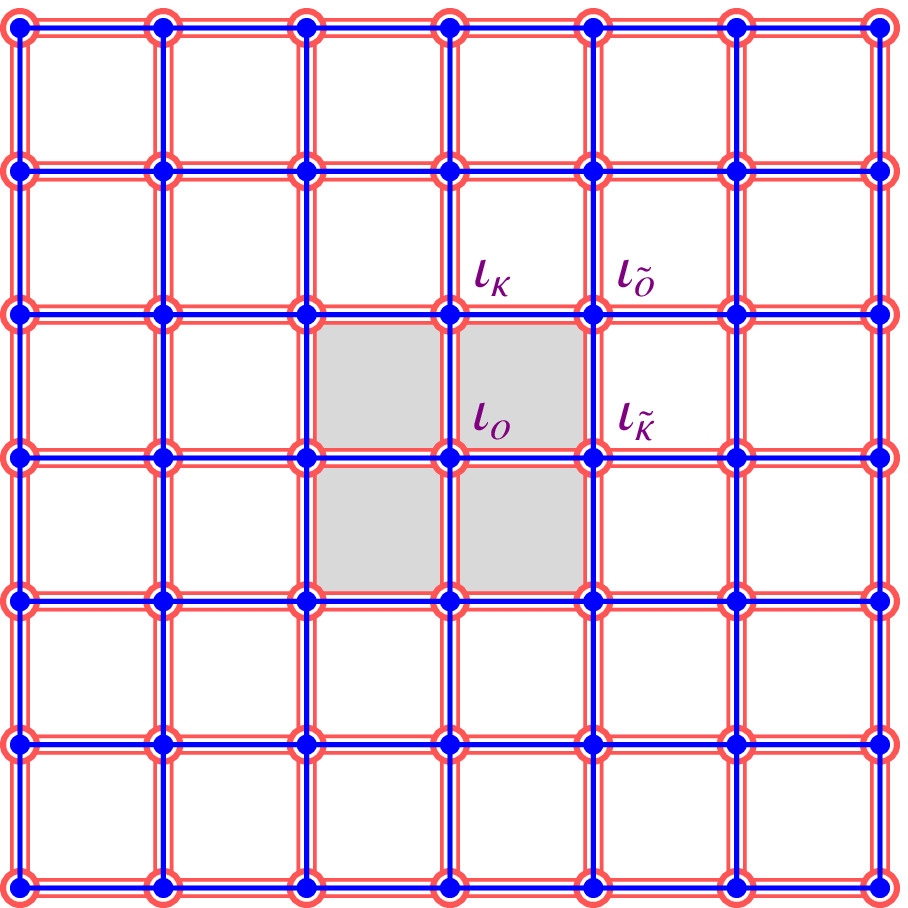}

(m) $\left(\begin{array}{cccccc}
\gamma_{o} & \delta_{o} & \gamma_{\tilde{o}} & 1 & 1 & \gamma_{o}\gamma_{\tilde{o}}\\
1 & 1 & 1 & e_{o} & e_{\tilde{o}} & e_{\kappa}  %e_{\gamma}
\end{array}\right)$%
\end{minipage}%
\begin{minipage}[t]{0.33\textwidth}%
\includegraphics[width=0.9\columnwidth]{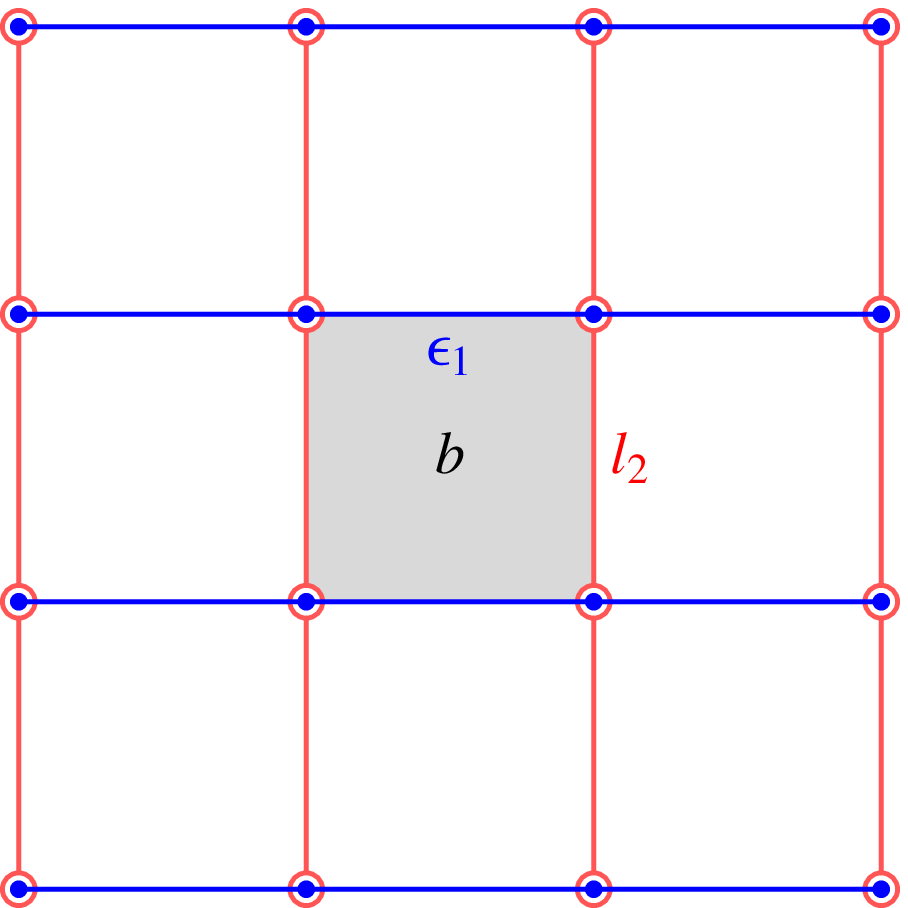}

(n) $\left(\begin{array}{cccccc}
\gamma_{1} & \delta_{\tilde{o}} & 1 & b & 1 & \gamma_{2}\\
1 & 1 & \alpha_{2} & 1 & e_{\tilde{o}} & \alpha_{1}
\end{array}\right)$%
\end{minipage}%
\begin{minipage}[t]{0.33\textwidth}%
\includegraphics[width=0.9\columnwidth]{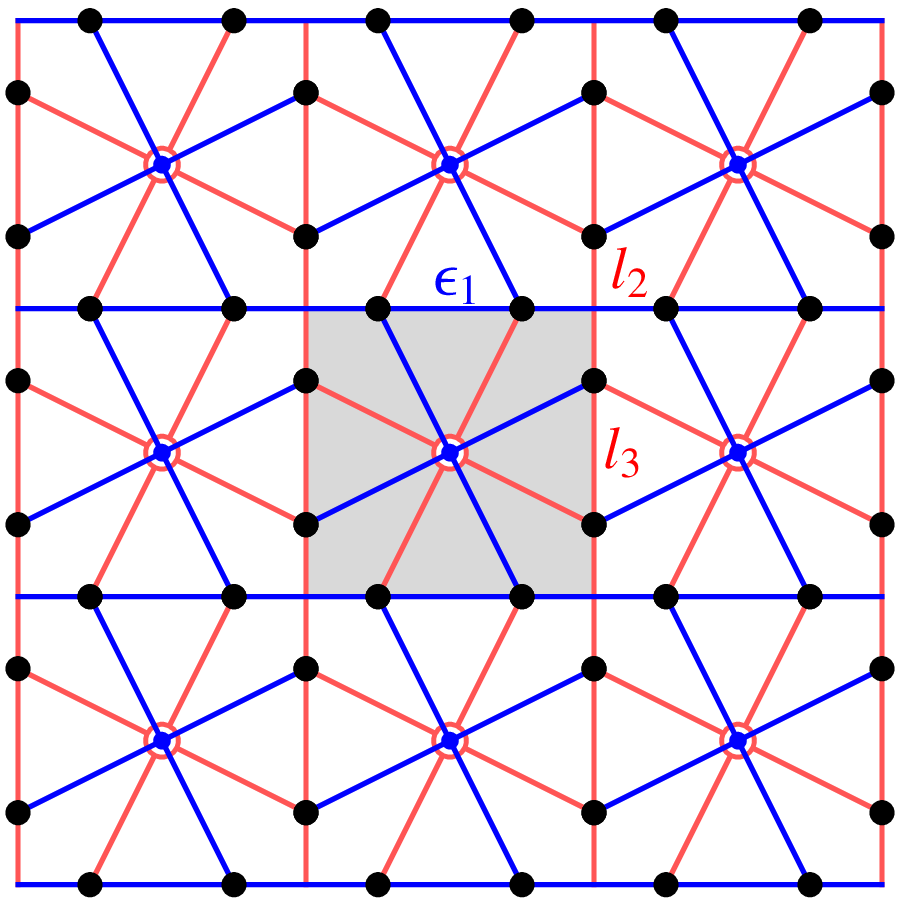}

(o) $\left(\begin{array}{cccccc}
\gamma_{o} & \delta_{o} & 1 & 1 & \gamma_{2} & \gamma_{3}\\
1 & 1 & \alpha_{3} & e_{o} & 1 & \alpha_{1}
\end{array}\right)$%
\end{minipage}

\bigskip{}

\begin{minipage}[t]{0.33\textwidth}%
\includegraphics[width=0.9\columnwidth]{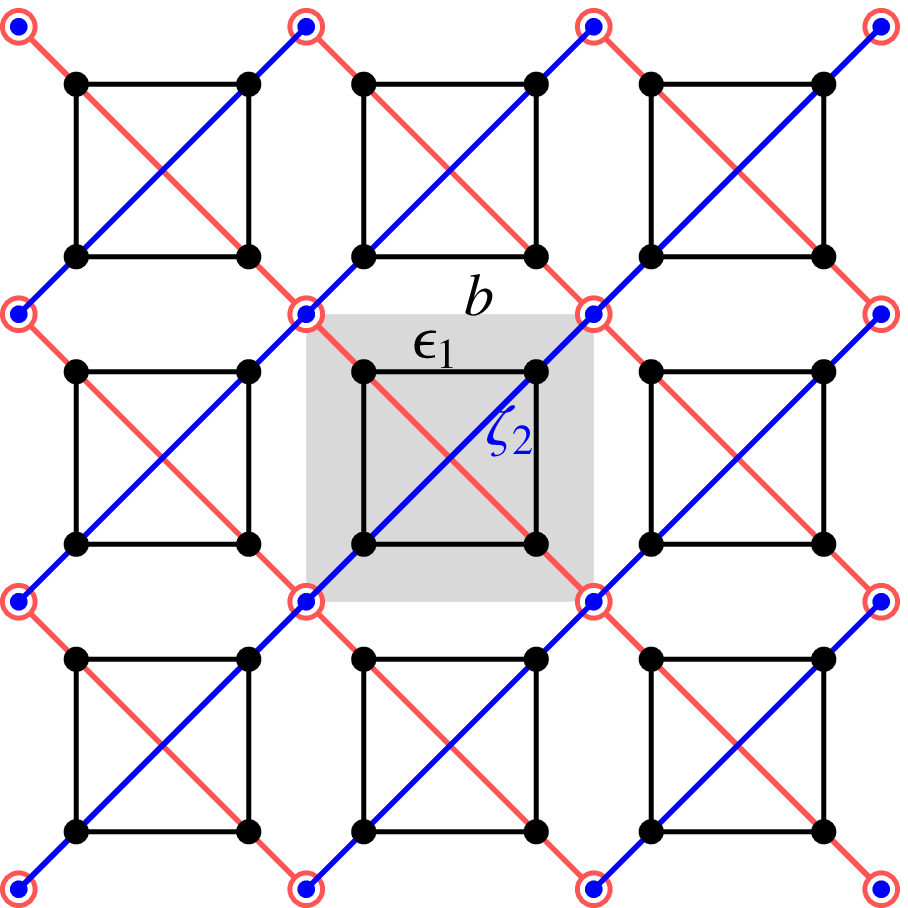}

(p) $\left(\begin{array}{cccccc}
\gamma_{1} & 1 & \gamma_{\tilde{o}} & \delta_{2} & 1 & b\gamma_{1}\gamma_{\tilde{o}}\\
1 & \beta_{2} & 1 & 1 & e_{\tilde{o}} & 1
\end{array}\right)$%
\end{minipage}%
\begin{minipage}[t]{0.33\textwidth}%
\includegraphics[width=0.9\columnwidth]{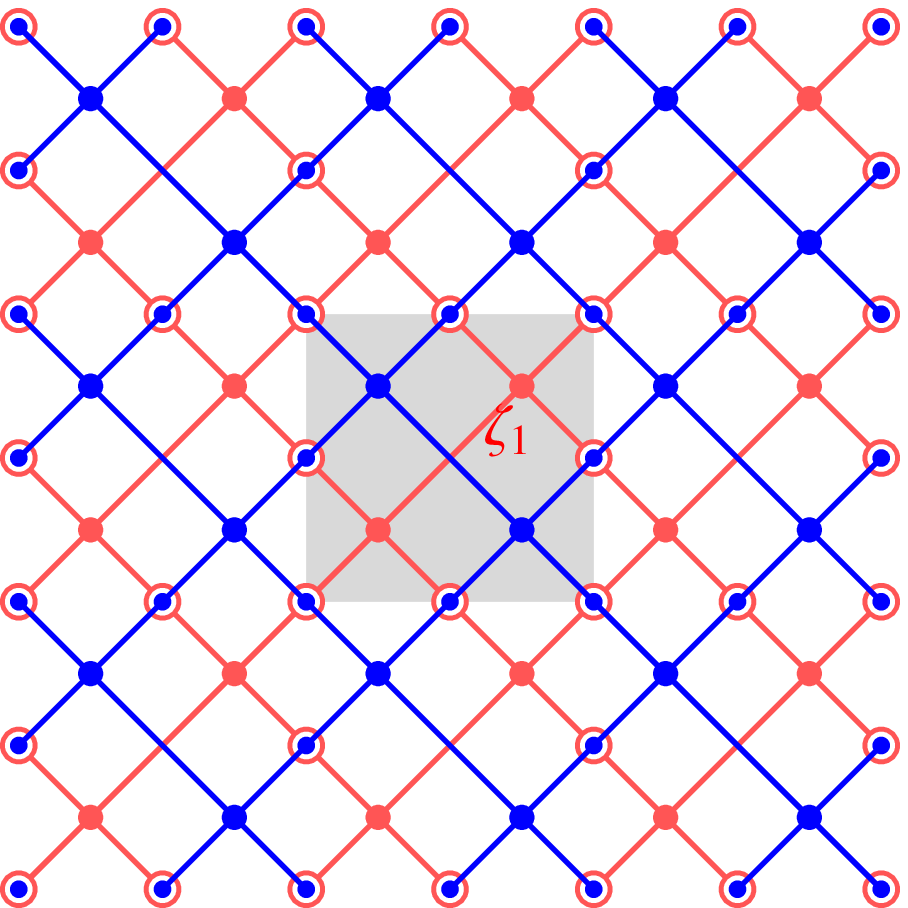}

(q) $\left(\begin{array}{cccccc}
\gamma_{\kappa} & 1 & \gamma_{\tilde{o}} & \delta_{1} & 1 & \gamma_{\kappa}\gamma_{\tilde{o}}\\
1 & \beta_{1} & 1 & 1 & e_{\tilde{o}} & e_{\kappa}
\end{array}\right)$%
\end{minipage}%
\begin{minipage}[t]{0.33\textwidth}%
\includegraphics[width=0.9\columnwidth]{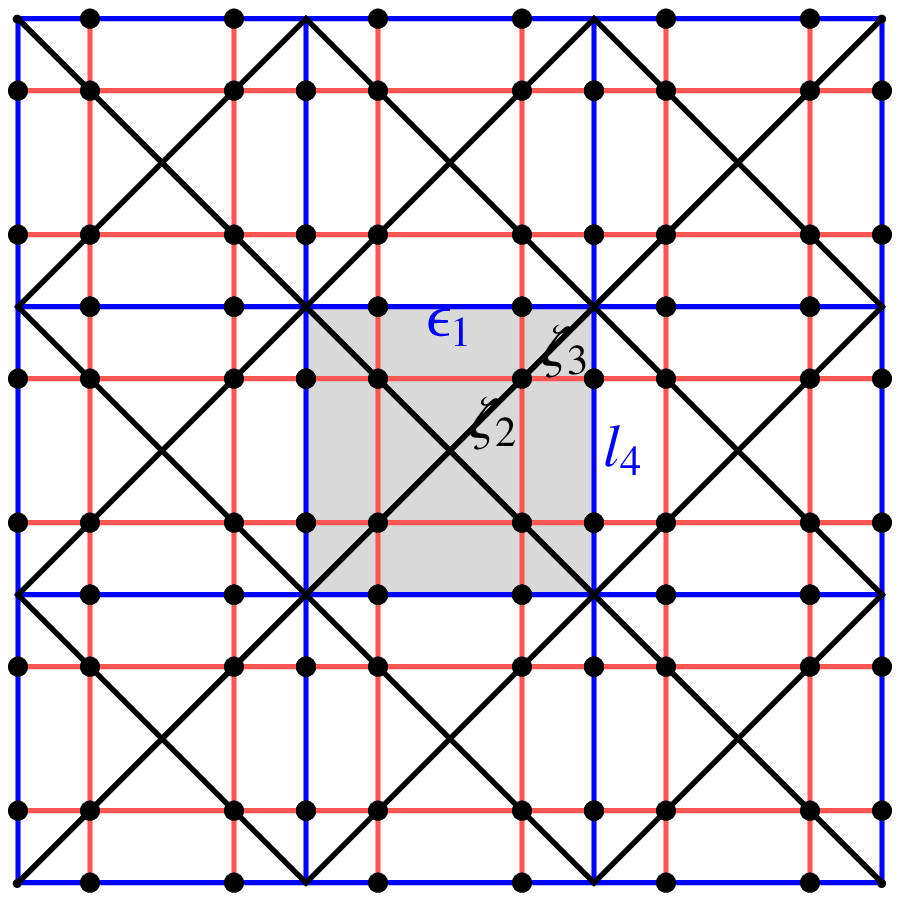}

(r) $\left(\begin{array}{cccccc}
\gamma_{1} & 1 & 1 & \delta_{2} & \delta_{3} & \gamma_{4}\\
1 & \beta_{2} & \alpha_{4} & 1 & 1 & \alpha_{1}
\end{array}\right)$%
\end{minipage}

\caption{(Color online) $TC\left(G\right)$ models (Part II). 
The shaded square is a unit cell
and the TC symmetry classes are calculated with the origin $o$ at
the center of the shaded square. Below each  lattice is the
corresponding TC symmetry class in the form \eqref{eq:sc_matrix}.
The edges are labeled by different letters according to their directions
as described in the text and in Fig.~\ref{fig:tcg-caption}. Edges that map to a single 
point under ${\mathscr P}$ are labeled by $\iota_{o}$, $\iota_{\tilde{o}}$,
$\iota_{\kappa}$, $\iota_{\tilde{\kappa}}$ with the subscript indicating their position,
and $\tilde{o}=\left(\frac{1}{2},\frac{1}{2}\right)$, $\kappa=\left(0,\frac{1}{2}\right)$,
$\tilde{\kappa}=\left(\frac{1}{2},0\right)$, in units such that
the size of the unit cell is $1\times1$. For short, we define
$\alpha_{i}=c_{\varepsilon_{i}}^{x}\left(P_{x}\right)$, $\beta_{i}=c_{\varepsilon_{i}}^{x}\left(P_{xy}\right)$,
$\gamma_{i}=c_{\varepsilon_{i}}^{z}\left(P_{x}\right)$ and $\delta_{i}=c_{\varepsilon_{i}}^{z}\left(P_{xy}\right)$,
where $\varepsilon=l,\epsilon,\xi,\zeta,\iota$ stands for a generic
edge. In addition, $e_{r}=a_{\mathscr{P}^{-1}\left(r\right)}$, and
$b$ is the eigenvalue of $B_{p}$ for the plaquette (here meaning
smallest cycle) $p$ within which $b$ is written.
 In panel (l),
$b$ is the eigenvalue of $B_{p}$ for the top plaquette.
The values of $e_{r}$ and $b$ are well-defined with respect to any local spin frame system satisfying Eqs.~(\ref{eq:gauge1}-\ref{eq:gauge3}). }

\label{fig:tcg2}
\end{figure*}

\begin{figure*}
\begin{minipage}[t]{0.33\textwidth}%
\includegraphics[width=0.9\columnwidth]{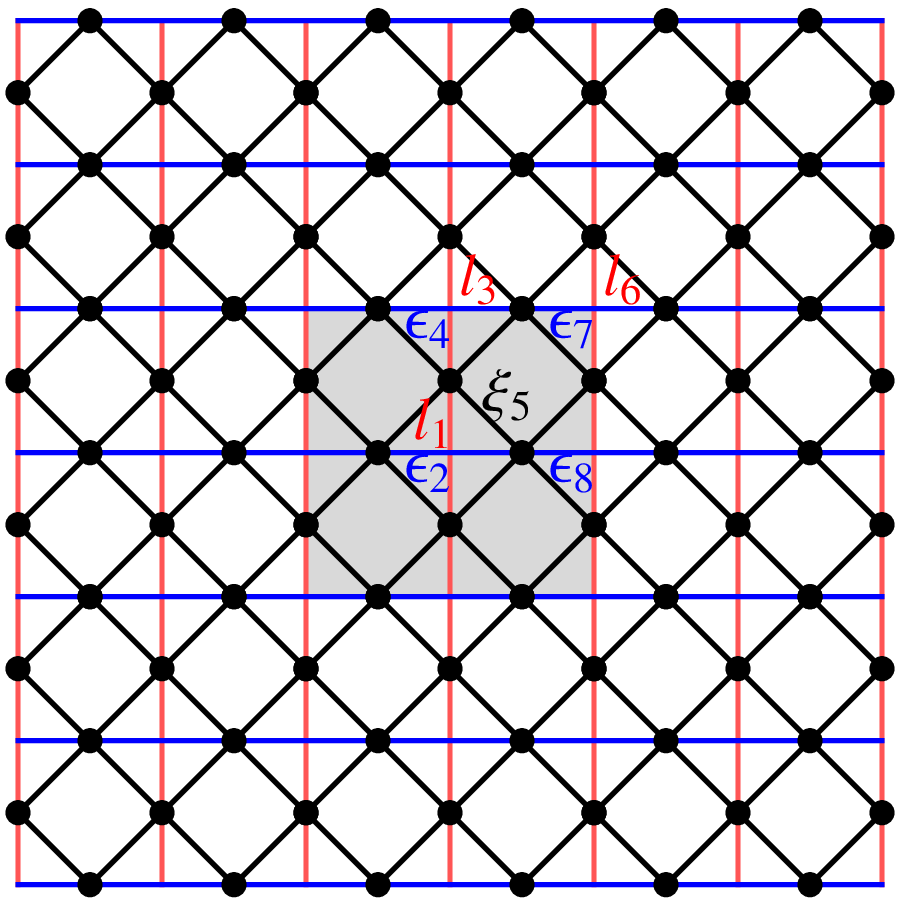}

(s) $\left(\begin{array}{cccccc}
1 & \gamma_{5} & 1 & \gamma_{1} & \gamma_{6} & \gamma_{3}\\
\alpha_{1} & 1 & \alpha_{6} & 1 & 1 & \alpha_{4}\alpha_{8}
\end{array}\right)$%
\end{minipage}%
\begin{minipage}[t]{0.33\textwidth}%
\includegraphics[width=0.9\columnwidth]{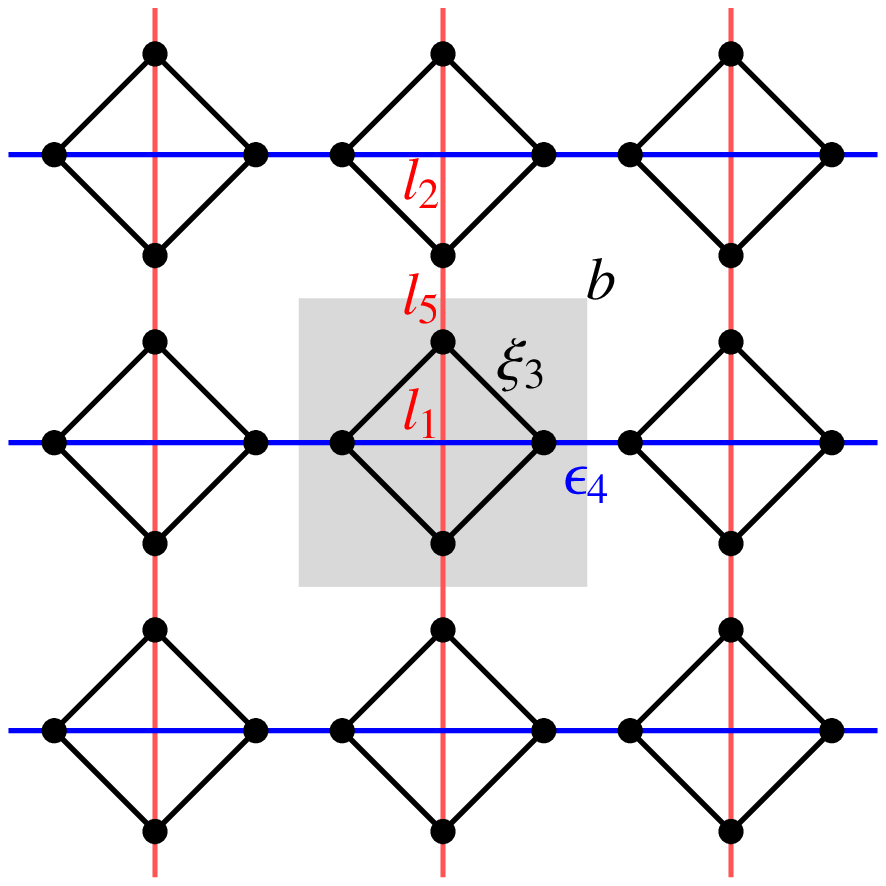}

(t) $\left(\begin{array}{cccccc}
1 & \gamma_{3} & \gamma_{4} & \gamma_{1} & b & \gamma_{5}\\
\alpha_{1} & 1 & 1 & 1 & 1 & \alpha_{4}
\end{array}\right)$%
\end{minipage}%
\begin{minipage}[t]{0.33\textwidth}%
\includegraphics[width=0.9\columnwidth]{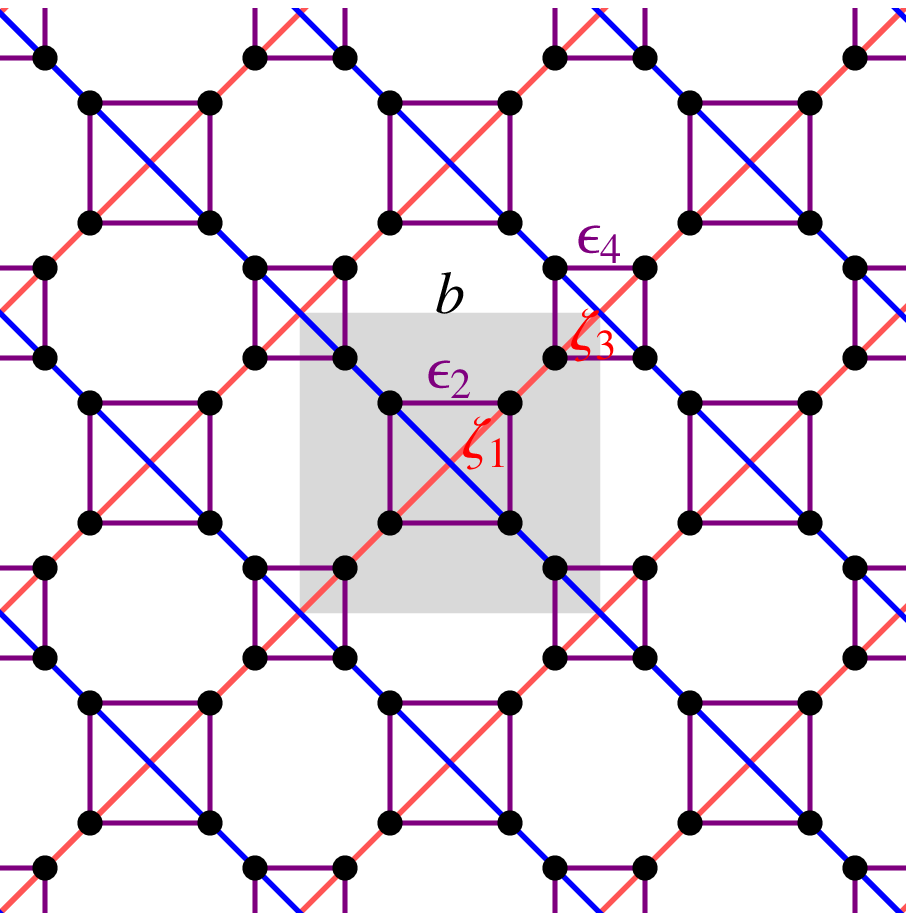}

(u) $\left(\begin{array}{cccccc}
\gamma_{2} & 1 & \gamma_{4} & \delta_{1} & \delta_{3} & b\gamma_{2}\gamma_{4}\\
1 & \beta_{1} & 1 & 1 & 1 & 1
\end{array}\right)$%
\end{minipage}

\bigskip{}

\begin{minipage}[t]{0.33\textwidth}%
\includegraphics[width=0.9\columnwidth]{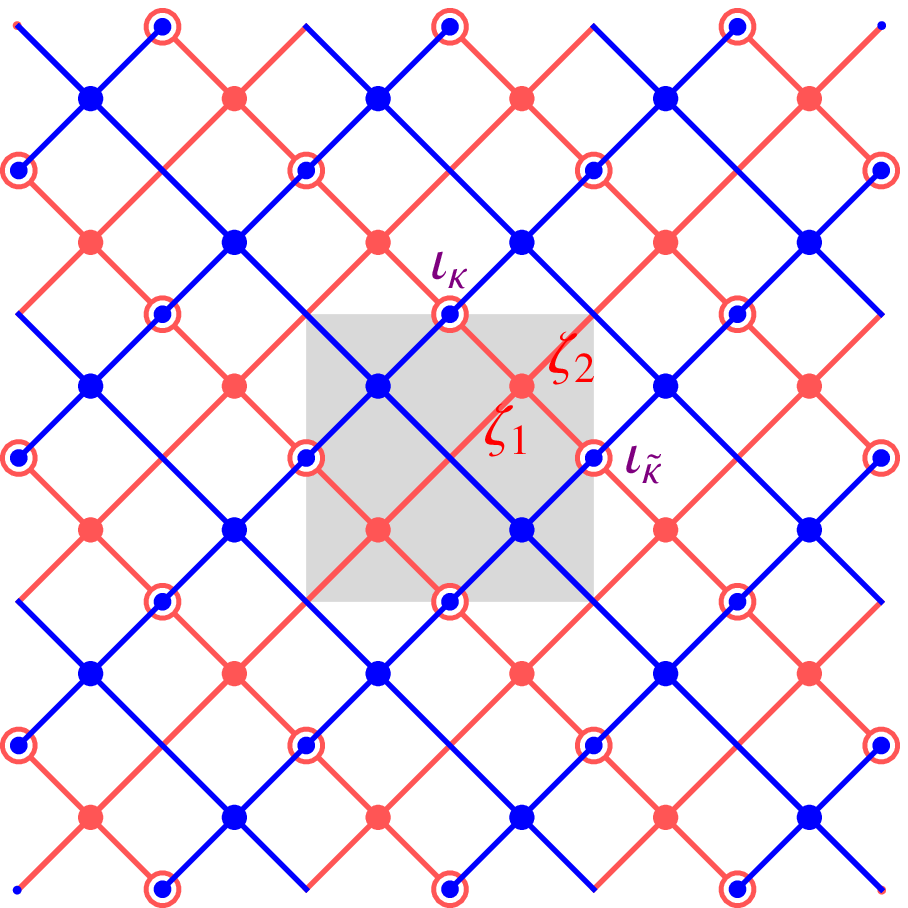}

(v) $\left(\begin{array}{cccccc}
\gamma_{\kappa} & 1 & \gamma_{\tilde{\kappa}} & \delta_{1} & \delta_{2} & \gamma_{\kappa}\gamma_{\tilde{\kappa}}\\
1 & \beta_{1} & 1 & 1 & 1 & e_{\kappa}
\end{array}\right)$%
\end{minipage}%
\begin{minipage}[t]{0.33\textwidth}%
\includegraphics[width=0.9\columnwidth]{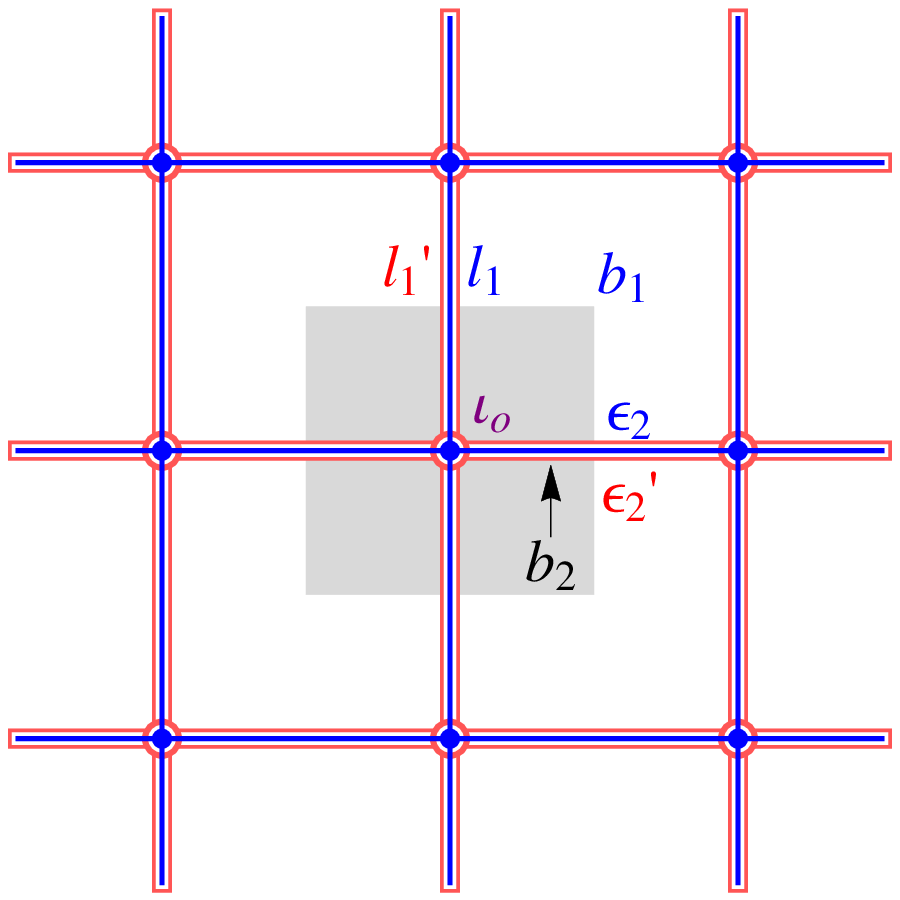}

(w) $\left(\begin{array}{cccccc}
\gamma_{o} & \delta_{o} & \gamma_{o}b_{2} & 1 & b_{1} & \gamma_{1}b_{2}\\
1 & 1 & 1 & e_{o} & 1 & 1
\end{array}\right)$%
\end{minipage}%
\begin{minipage}[t]{0.33\textwidth}%
\includegraphics[width=0.9\columnwidth]{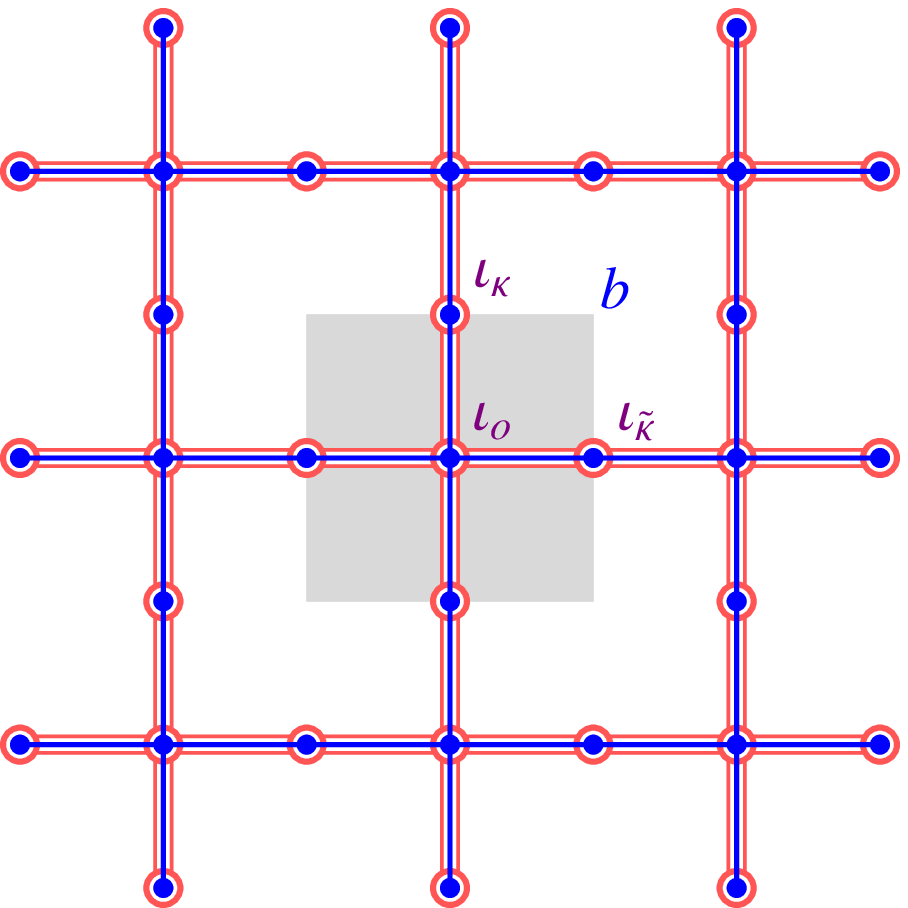}

(x) $\left(\begin{array}{cccccc}
\gamma_{o} & \delta_{o} & \gamma_{\tilde{\kappa}} & 1 & b & \gamma_{o}\gamma_{\tilde{\kappa}}\\
1 & 1 & 1 & e_{o} & 1 & e_{\kappa}
\end{array}\right)$%
\end{minipage}

\bigskip{}

\begin{minipage}[t]{0.33\textwidth}%
\includegraphics[width=0.9\columnwidth]{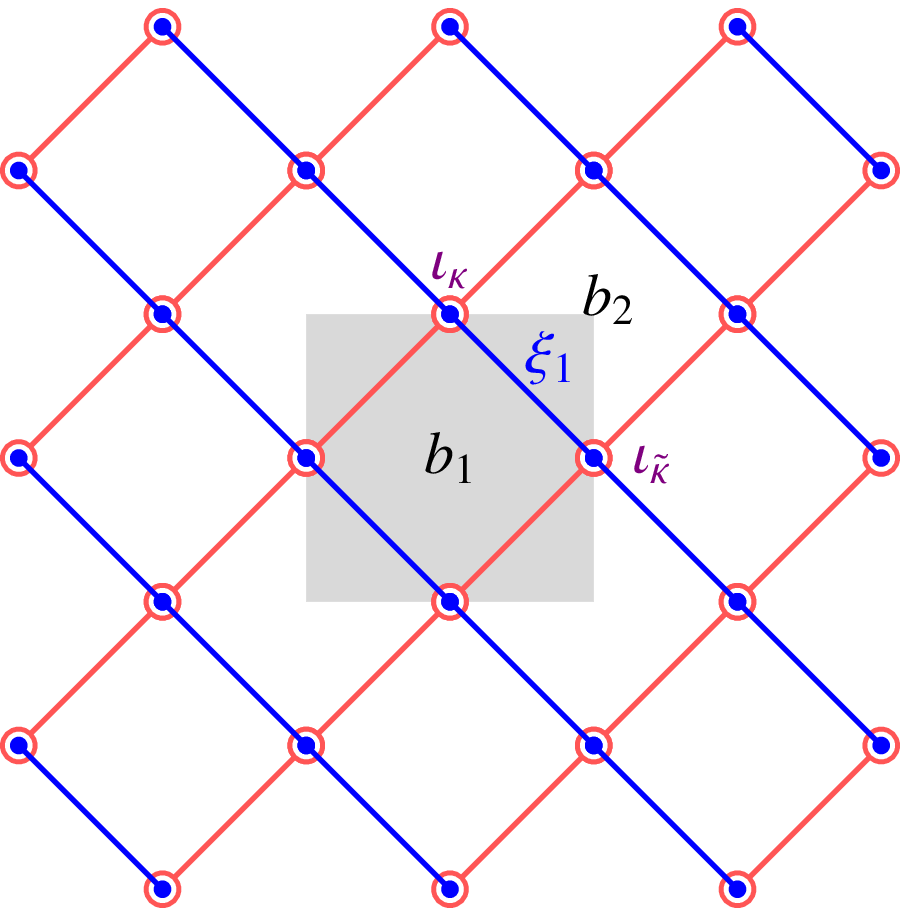}

(y) $\left(\begin{array}{cccccc}
\gamma_{\kappa} & \delta_{1} & \gamma_{\tilde{\kappa}} & b_{1} & b_{2} & \gamma_{\kappa}\gamma_{\tilde{\kappa}}\\
1 & 1 & 1 & 1 & 1 & e_{\kappa}
\end{array}\right)$%
\end{minipage}

\caption{(Color online) $TC\left(G\right)$ models (Part III). 
The shaded square is a unit cell
and the TC symmetry classes are calculated with the origin $o$ at
the center of the shaded square. Below each  lattice is the
corresponding TC symmetry class in the form \eqref{eq:sc_matrix}.
The edges are labeled by different letters according to their directions
as described in the text and in Fig.~\ref{fig:tcg-caption}. Edges that map to a single 
point under ${\mathscr P}$ are labeled by $\iota_{o}$, $\iota_{\tilde{o}}$,
$\iota_{\kappa}$, $\iota_{\tilde{\kappa}}$ with the subscript indicating their position,
and $\tilde{o}=\left(\frac{1}{2},\frac{1}{2}\right)$, $\kappa=\left(0,\frac{1}{2}\right)$,
$\tilde{\kappa}=\left(\frac{1}{2},0\right)$, in units such that
the size of the unit cell is $1\times1$. For short, we define
$\alpha_{i}=c_{\varepsilon_{i}}^{x}\left(P_{x}\right)$, $\beta_{i}=c_{\varepsilon_{i}}^{x}\left(P_{xy}\right)$,
$\gamma_{i}=c_{\varepsilon_{i}}^{z}\left(P_{x}\right)$ and $\delta_{i}=c_{\varepsilon_{i}}^{z}\left(P_{xy}\right)$,
where $\varepsilon=l,\epsilon,\xi,\zeta,\iota$ stands for a generic
edge. In addition, $e_{r}=a_{\mathscr{P}^{-1}\left(r\right)}$, and
$b$ (or $b_i$) is the eigenvalue of $B_{p}$ for the plaquette (here meaning
smallest cycle) $p$ within which $b$ (or $b_i$) is written.
In panels (w) and (x), $b_{1}$ and $b$ are the eigenvalues
of $B_{p}$ for the top plaquettes. The values of $e_{r}$ and $b$ (or $b_i$) are well-defined with respect to any local spin frame system satisfying Eqs.~(\ref{eq:gauge1}-\ref{eq:gauge3}) except in (w), where a further gauge fixing is needed and we require $c_{l_{1}}^{z}\left(P_{xy}\right)=c_{l_{1}'}^{z}\left(P_{xy}\right)=c_{l_{2}}^{z}\left(P_{x}\right)=1$.}

\label{fig:tcg3}
\end{figure*}

\newpage

\bibliography{referencez2}

\begin{thebibliography}{63}%
\makeatletter
\providecommand \@ifxundefined [1]{%
 \@ifx{#1\undefined}
}%
\providecommand \@ifnum [1]{%
 \ifnum #1\expandafter \@firstoftwo
 \else \expandafter \@secondoftwo
 \fi
}%
\providecommand \@ifx [1]{%
 \ifx #1\expandafter \@firstoftwo
 \else \expandafter \@secondoftwo
 \fi
}%
\providecommand \natexlab [1]{#1}%
\providecommand \enquote  [1]{``#1''}%
\providecommand \bibnamefont  [1]{#1}%
\providecommand \bibfnamefont [1]{#1}%
\providecommand \citenamefont [1]{#1}%
\providecommand \href@noop [0]{\@secondoftwo}%
\providecommand \href [0]{\begingroup \@sanitize@url \@href}%
\providecommand \@href[1]{\@@startlink{#1}\@@href}%
\providecommand \@@href[1]{\endgroup#1\@@endlink}%
\providecommand \@sanitize@url [0]{\catcode `\\12\catcode `\$12\catcode
  `\&12\catcode `\#12\catcode `\^12\catcode `\_12\catcode `\%12\relax}%
\providecommand \@@startlink[1]{}%
\providecommand \@@endlink[0]{}%
\providecommand \url  [0]{\begingroup\@sanitize@url \@url }%
\providecommand \@url [1]{\endgroup\@href {#1}{\urlprefix }}%
\providecommand \urlprefix  [0]{URL }%
\providecommand \Eprint [0]{\href }%
\providecommand \doibase [0]{http://dx.doi.org/}%
\providecommand \selectlanguage [0]{\@gobble}%
\providecommand \bibinfo  [0]{\@secondoftwo}%
\providecommand \bibfield  [0]{\@secondoftwo}%
\providecommand \translation [1]{[#1]}%
\providecommand \BibitemOpen [0]{}%
\providecommand \bibitemStop [0]{}%
\providecommand \bibitemNoStop [0]{.\EOS\space}%
\providecommand \EOS [0]{\spacefactor3000\relax}%
\providecommand \BibitemShut  [1]{\csname bibitem#1\endcsname}%
\let\auto@bib@innerbib\@empty
%</preamble>
\bibitem [{\citenamefont {Hasan}\ and\ \citenamefont {Kane}(2010)}]{hasan10}%
  \BibitemOpen
  \bibfield  {author} {\bibinfo {author} {\bibfnamefont {M.~Z.}\ \bibnamefont
  {Hasan}}\ and\ \bibinfo {author} {\bibfnamefont {C.~L.}\ \bibnamefont
  {Kane}},\ }\href {\doibase 10.1103/RevModPhys.82.3045} {\bibfield  {journal}
  {\bibinfo  {journal} {Rev. Mod. Phys.}\ }\textbf {\bibinfo {volume} {82}},\
  \bibinfo {pages} {3045} (\bibinfo {year} {2010})}\BibitemShut {NoStop}%
\bibitem [{\citenamefont {Hasan}\ and\ \citenamefont {Moore}(2011)}]{hasan11}%
  \BibitemOpen
  \bibfield  {author} {\bibinfo {author} {\bibfnamefont {M.~Z.}\ \bibnamefont
  {Hasan}}\ and\ \bibinfo {author} {\bibfnamefont {J.~E.}\ \bibnamefont
  {Moore}},\ }\href {\doibase 10.1146/annurev-conmatphys-062910-140432}
  {\bibfield  {journal} {\bibinfo  {journal} {Annual Review of Condensed Matter
  Physics}\ }\textbf {\bibinfo {volume} {2}},\ \bibinfo {pages} {55} (\bibinfo
  {year} {2011})}\BibitemShut {NoStop}%
\bibitem [{\citenamefont {Qi}\ and\ \citenamefont {Zhang}(2011)}]{qi11}%
  \BibitemOpen
  \bibfield  {author} {\bibinfo {author} {\bibfnamefont {X.-L.}\ \bibnamefont
  {Qi}}\ and\ \bibinfo {author} {\bibfnamefont {S.-C.}\ \bibnamefont {Zhang}},\
  }\href {\doibase 10.1103/RevModPhys.83.1057} {\bibfield  {journal} {\bibinfo
  {journal} {Rev. Mod. Phys.}\ }\textbf {\bibinfo {volume} {83}},\ \bibinfo
  {pages} {1057} (\bibinfo {year} {2011})}\BibitemShut {NoStop}%
\bibitem [{\citenamefont {Chen}\ \emph {et~al.}(2011)\citenamefont {Chen},
  \citenamefont {Gu},\ and\ \citenamefont {Wen}}]{chen11a}%
  \BibitemOpen
  \bibfield  {author} {\bibinfo {author} {\bibfnamefont {X.}~\bibnamefont
  {Chen}}, \bibinfo {author} {\bibfnamefont {Z.-C.}\ \bibnamefont {Gu}}, \ and\
  \bibinfo {author} {\bibfnamefont {X.-G.}\ \bibnamefont {Wen}},\ }\href
  {\doibase 10.1103/PhysRevB.83.035107} {\bibfield  {journal} {\bibinfo
  {journal} {Phys. Rev. B}\ }\textbf {\bibinfo {volume} {83}},\ \bibinfo
  {pages} {035107} (\bibinfo {year} {2011})}\BibitemShut {NoStop}%
\bibitem [{\citenamefont {Fidkowski}\ and\ \citenamefont
  {Kitaev}(2011)}]{fidkowski11}%
  \BibitemOpen
  \bibfield  {author} {\bibinfo {author} {\bibfnamefont {L.}~\bibnamefont
  {Fidkowski}}\ and\ \bibinfo {author} {\bibfnamefont {A.}~\bibnamefont
  {Kitaev}},\ }\href {\doibase 10.1103/PhysRevB.83.075103} {\bibfield
  {journal} {\bibinfo  {journal} {Phys. Rev. B}\ }\textbf {\bibinfo {volume}
  {83}},\ \bibinfo {pages} {075103} (\bibinfo {year} {2011})}\BibitemShut
  {NoStop}%
\bibitem [{\citenamefont {Schuch}\ \emph {et~al.}(2011)\citenamefont {Schuch},
  \citenamefont {P{\'e}rez-Garc{\'\i}a},\ and\ \citenamefont
  {Cirac}}]{schuch11}%
  \BibitemOpen
  \bibfield  {author} {\bibinfo {author} {\bibfnamefont {N.}~\bibnamefont
  {Schuch}}, \bibinfo {author} {\bibfnamefont {D.}~\bibnamefont
  {P{\'e}rez-Garc{\'\i}a}}, \ and\ \bibinfo {author} {\bibfnamefont
  {I.}~\bibnamefont {Cirac}},\ }\href {\doibase 10.1103/PhysRevB.84.165139}
  {\bibfield  {journal} {\bibinfo  {journal} {Phys. Rev. B}\ }\textbf {\bibinfo
  {volume} {84}},\ \bibinfo {pages} {165139} (\bibinfo {year}
  {2011})}\BibitemShut {NoStop}%
\bibitem [{\citenamefont {Turner}\ \emph {et~al.}(2011)\citenamefont {Turner},
  \citenamefont {Pollmann},\ and\ \citenamefont {Berg}}]{turner11}%
  \BibitemOpen
  \bibfield  {author} {\bibinfo {author} {\bibfnamefont {A.~M.}\ \bibnamefont
  {Turner}}, \bibinfo {author} {\bibfnamefont {F.}~\bibnamefont {Pollmann}}, \
  and\ \bibinfo {author} {\bibfnamefont {E.}~\bibnamefont {Berg}},\ }\href
  {\doibase 10.1103/PhysRevB.83.075102} {\bibfield  {journal} {\bibinfo
  {journal} {Phys. Rev. B}\ }\textbf {\bibinfo {volume} {83}},\ \bibinfo
  {pages} {075102} (\bibinfo {year} {2011})}\BibitemShut {NoStop}%
\bibitem [{\citenamefont {Chen}\ \emph {et~al.}(2013)\citenamefont {Chen},
  \citenamefont {Gu}, \citenamefont {Liu},\ and\ \citenamefont {Wen}}]{chen13}%
  \BibitemOpen
  \bibfield  {author} {\bibinfo {author} {\bibfnamefont {X.}~\bibnamefont
  {Chen}}, \bibinfo {author} {\bibfnamefont {Z.-C.}\ \bibnamefont {Gu}},
  \bibinfo {author} {\bibfnamefont {Z.-X.}\ \bibnamefont {Liu}}, \ and\
  \bibinfo {author} {\bibfnamefont {X.-G.}\ \bibnamefont {Wen}},\ }\href
  {\doibase 10.1103/PhysRevB.87.155114} {\bibfield  {journal} {\bibinfo
  {journal} {Phys. Rev. B}\ }\textbf {\bibinfo {volume} {87}},\ \bibinfo
  {pages} {155114} (\bibinfo {year} {2013})}\BibitemShut {NoStop}%
\bibitem [{\citenamefont {Tsui}\ \emph {et~al.}(1982)\citenamefont {Tsui},
  \citenamefont {Stormer},\ and\ \citenamefont {Gossard}}]{tsui82}%
  \BibitemOpen
  \bibfield  {author} {\bibinfo {author} {\bibfnamefont {D.~C.}\ \bibnamefont
  {Tsui}}, \bibinfo {author} {\bibfnamefont {H.~L.}\ \bibnamefont {Stormer}}, \
  and\ \bibinfo {author} {\bibfnamefont {A.~C.}\ \bibnamefont {Gossard}},\
  }\href@noop {} {\bibfield  {journal} {\bibinfo  {journal} {Phys. Rev. Lett.}\
  }\textbf {\bibinfo {volume} {48}},\ \bibinfo {pages} {1559} (\bibinfo {year}
  {1982})}\BibitemShut {NoStop}%
\bibitem [{\citenamefont {Laughlin}(1983)}]{laughlin83}%
  \BibitemOpen
  \bibfield  {author} {\bibinfo {author} {\bibfnamefont {R.~B.}\ \bibnamefont
  {Laughlin}},\ }\href {\doibase 10.1103/PhysRevLett.50.1395} {\bibfield
  {journal} {\bibinfo  {journal} {Phys. Rev. Lett.}\ }\textbf {\bibinfo
  {volume} {50}},\ \bibinfo {pages} {1395} (\bibinfo {year}
  {1983})}\BibitemShut {NoStop}%
\bibitem [{\citenamefont {de~Picciotto}\ \emph {et~al.}(1997)\citenamefont
  {de~Picciotto}, \citenamefont {Reznikov}, \citenamefont {Heiblum},
  \citenamefont {Umansky}, \citenamefont {Bunin},\ and\ \citenamefont
  {Mahalu}}]{depicciotto97}%
  \BibitemOpen
  \bibfield  {author} {\bibinfo {author} {\bibfnamefont {R.}~\bibnamefont
  {de~Picciotto}}, \bibinfo {author} {\bibfnamefont {M.}~\bibnamefont
  {Reznikov}}, \bibinfo {author} {\bibfnamefont {M.}~\bibnamefont {Heiblum}},
  \bibinfo {author} {\bibfnamefont {V.}~\bibnamefont {Umansky}}, \bibinfo
  {author} {\bibfnamefont {G.}~\bibnamefont {Bunin}}, \ and\ \bibinfo {author}
  {\bibfnamefont {D.}~\bibnamefont {Mahalu}},\ }\href {\doibase 10.1038/38241}
  {\bibfield  {journal} {\bibinfo  {journal} {Nature}\ }\textbf {\bibinfo
  {volume} {389}},\ \bibinfo {pages} {162} (\bibinfo {year}
  {1997})}\BibitemShut {NoStop}%
\bibitem [{\citenamefont {Saminadayar}\ \emph {et~al.}(1997)\citenamefont
  {Saminadayar}, \citenamefont {Glattli}, \citenamefont {Jin},\ and\
  \citenamefont {Etienne}}]{saminadayar97}%
  \BibitemOpen
  \bibfield  {author} {\bibinfo {author} {\bibfnamefont {L.}~\bibnamefont
  {Saminadayar}}, \bibinfo {author} {\bibfnamefont {D.~C.}\ \bibnamefont
  {Glattli}}, \bibinfo {author} {\bibfnamefont {Y.}~\bibnamefont {Jin}}, \ and\
  \bibinfo {author} {\bibfnamefont {B.}~\bibnamefont {Etienne}},\ }\href
  {\doibase 10.1103/PhysRevLett.79.2526} {\bibfield  {journal} {\bibinfo
  {journal} {Phys. Rev. Lett.}\ }\textbf {\bibinfo {volume} {79}},\ \bibinfo
  {pages} {2526} (\bibinfo {year} {1997})},\ \Eprint
  {http://arxiv.org/abs/cond-mat/9706307} {arXiv:cond-mat/9706307} \BibitemShut
  {NoStop}%
\bibitem [{\citenamefont {Martin}\ \emph {et~al.}(2004)\citenamefont {Martin},
  \citenamefont {Ilani}, \citenamefont {Verdene}, \citenamefont {Smet},
  \citenamefont {Umansky}, \citenamefont {Mahalu}, \citenamefont {Schuh},
  \citenamefont {Abstreiter},\ and\ \citenamefont {Yacoby}}]{Martin2004}%
  \BibitemOpen
  \bibfield  {author} {\bibinfo {author} {\bibfnamefont {J.}~\bibnamefont
  {Martin}}, \bibinfo {author} {\bibfnamefont {S.}~\bibnamefont {Ilani}},
  \bibinfo {author} {\bibfnamefont {B.}~\bibnamefont {Verdene}}, \bibinfo
  {author} {\bibfnamefont {J.}~\bibnamefont {Smet}}, \bibinfo {author}
  {\bibfnamefont {V.}~\bibnamefont {Umansky}}, \bibinfo {author} {\bibfnamefont
  {D.}~\bibnamefont {Mahalu}}, \bibinfo {author} {\bibfnamefont
  {D.}~\bibnamefont {Schuh}}, \bibinfo {author} {\bibfnamefont
  {G.}~\bibnamefont {Abstreiter}}, \ and\ \bibinfo {author} {\bibfnamefont
  {A.}~\bibnamefont {Yacoby}},\ }\href {\doibase 10.1126/science.1099950}
  {\bibfield  {journal} {\bibinfo  {journal} {Science}\ }\textbf {\bibinfo
  {volume} {305}},\ \bibinfo {pages} {980} (\bibinfo {year}
  {2004})}\BibitemShut {NoStop}%
\bibitem [{\citenamefont {Kalmeyer}\ and\ \citenamefont
  {Laughlin}(1987)}]{kalmeyer1987}%
  \BibitemOpen
  \bibfield  {author} {\bibinfo {author} {\bibfnamefont {V.}~\bibnamefont
  {Kalmeyer}}\ and\ \bibinfo {author} {\bibfnamefont {R.~B.}\ \bibnamefont
  {Laughlin}},\ }\href@noop {} {\bibfield  {journal} {\bibinfo  {journal}
  {Phys. Rev. Lett.}\ }\textbf {\bibinfo {volume} {59}},\ \bibinfo {pages}
  {2095} (\bibinfo {year} {1987})}\BibitemShut {NoStop}%
\bibitem [{\citenamefont {Wen}\ \emph {et~al.}(1989)\citenamefont {Wen},
  \citenamefont {Wilczek},\ and\ \citenamefont {Zee}}]{wen89b}%
  \BibitemOpen
  \bibfield  {author} {\bibinfo {author} {\bibfnamefont {X.~G.}\ \bibnamefont
  {Wen}}, \bibinfo {author} {\bibfnamefont {F.}~\bibnamefont {Wilczek}}, \ and\
  \bibinfo {author} {\bibfnamefont {A.}~\bibnamefont {Zee}},\ }\href@noop {}
  {\bibfield  {journal} {\bibinfo  {journal} {Phys. Rev. B}\ }\textbf {\bibinfo
  {volume} {39}},\ \bibinfo {pages} {11413} (\bibinfo {year}
  {1989})}\BibitemShut {NoStop}%
\bibitem [{\citenamefont {Wen}(1991)}]{wen91}%
  \BibitemOpen
  \bibfield  {author} {\bibinfo {author} {\bibfnamefont {X.~G.}\ \bibnamefont
  {Wen}},\ }\href {\doibase 10.1103/PhysRevB.44.2664} {\bibfield  {journal}
  {\bibinfo  {journal} {Phys. Rev. B}\ }\textbf {\bibinfo {volume} {44}},\
  \bibinfo {pages} {2664} (\bibinfo {year} {1991})}\BibitemShut {NoStop}%
\bibitem [{\citenamefont {Read}\ and\ \citenamefont {Sachdev}(1991)}]{read91}%
  \BibitemOpen
  \bibfield  {author} {\bibinfo {author} {\bibfnamefont {N.}~\bibnamefont
  {Read}}\ and\ \bibinfo {author} {\bibfnamefont {S.}~\bibnamefont {Sachdev}},\
  }\href {\doibase 10.1103/PhysRevLett.66.1773} {\bibfield  {journal} {\bibinfo
   {journal} {Phys. Rev. Lett.}\ }\textbf {\bibinfo {volume} {66}},\ \bibinfo
  {pages} {1773} (\bibinfo {year} {1991})}\BibitemShut {NoStop}%
\bibitem [{\citenamefont {Sachdev}(1992)}]{sachdev92}%
  \BibitemOpen
  \bibfield  {author} {\bibinfo {author} {\bibfnamefont {S.}~\bibnamefont
  {Sachdev}},\ }\href {\doibase 10.1103/PhysRevB.45.12377} {\bibfield
  {journal} {\bibinfo  {journal} {Phys. Rev. B}\ }\textbf {\bibinfo {volume}
  {45}},\ \bibinfo {pages} {12377} (\bibinfo {year} {1992})}\BibitemShut
  {NoStop}%
\bibitem [{\citenamefont {Balents}\ \emph {et~al.}(1999)\citenamefont
  {Balents}, \citenamefont {Fisher},\ and\ \citenamefont {Nayak}}]{balents99}%
  \BibitemOpen
  \bibfield  {author} {\bibinfo {author} {\bibfnamefont {L.}~\bibnamefont
  {Balents}}, \bibinfo {author} {\bibfnamefont {M.~P.~A.}\ \bibnamefont
  {Fisher}}, \ and\ \bibinfo {author} {\bibfnamefont {C.}~\bibnamefont
  {Nayak}},\ }\href {\doibase 10.1103/PhysRevB.60.1654} {\bibfield  {journal}
  {\bibinfo  {journal} {Phys. Rev. B}\ }\textbf {\bibinfo {volume} {60}},\
  \bibinfo {pages} {1654} (\bibinfo {year} {1999})},\ \Eprint
  {http://arxiv.org/abs/cond-mat/9811236} {arXiv:cond-mat/9811236} \BibitemShut
  {NoStop}%
\bibitem [{\citenamefont {Senthil}\ and\ \citenamefont
  {Fisher}(2000)}]{senthil00}%
  \BibitemOpen
  \bibfield  {author} {\bibinfo {author} {\bibfnamefont {T.}~\bibnamefont
  {Senthil}}\ and\ \bibinfo {author} {\bibfnamefont {M.~P.~A.}\ \bibnamefont
  {Fisher}},\ }\href {\doibase 10.1103/PhysRevB.62.7850} {\bibfield  {journal}
  {\bibinfo  {journal} {Phys. Rev. B}\ }\textbf {\bibinfo {volume} {62}},\
  \bibinfo {pages} {7850} (\bibinfo {year} {2000})},\ \Eprint
  {http://arxiv.org/abs/cond-mat/9910224} {arXiv:cond-mat/9910224} \BibitemShut
  {NoStop}%
\bibitem [{\citenamefont {Moessner}\ \emph {et~al.}(2001)\citenamefont
  {Moessner}, \citenamefont {Sondhi},\ and\ \citenamefont
  {Fradkin}}]{moessner01b}%
  \BibitemOpen
  \bibfield  {author} {\bibinfo {author} {\bibfnamefont {R.}~\bibnamefont
  {Moessner}}, \bibinfo {author} {\bibfnamefont {S.~L.}\ \bibnamefont
  {Sondhi}}, \ and\ \bibinfo {author} {\bibfnamefont {E.}~\bibnamefont
  {Fradkin}},\ }\href {\doibase 10.1103/PhysRevB.65.024504} {\bibfield
  {journal} {\bibinfo  {journal} {Phys. Rev. B}\ }\textbf {\bibinfo {volume}
  {65}},\ \bibinfo {pages} {024504} (\bibinfo {year} {2001})},\ \Eprint
  {http://arxiv.org/abs/cond-mat/0103396} {arXiv:cond-mat/0103396} \BibitemShut
  {NoStop}%
\bibitem [{\citenamefont {Wen}(2002)}]{wen02}%
  \BibitemOpen
  \bibfield  {author} {\bibinfo {author} {\bibfnamefont {X.-G.}\ \bibnamefont
  {Wen}},\ }\href {\doibase 10.1103/PhysRevB.65.165113} {\bibfield  {journal}
  {\bibinfo  {journal} {Phys. Rev. B}\ }\textbf {\bibinfo {volume} {65}},\
  \bibinfo {pages} {165113} (\bibinfo {year} {2002})},\ \Eprint
  {http://arxiv.org/abs/cond-mat/0107071} {arXiv:cond-mat/0107071} \BibitemShut
  {NoStop}%
\bibitem [{\citenamefont {Wen}(2003)}]{wen03}%
  \BibitemOpen
  \bibfield  {author} {\bibinfo {author} {\bibfnamefont {X.-G.}\ \bibnamefont
  {Wen}},\ }\href {\doibase 10.1103/PhysRevD.68.065003} {\bibfield  {journal}
  {\bibinfo  {journal} {Phys. Rev. D}\ }\textbf {\bibinfo {volume} {68}},\
  \bibinfo {pages} {065003} (\bibinfo {year} {2003})}\BibitemShut {NoStop}%
\bibitem [{\citenamefont {Kitaev}(2006)}]{kitaev06}%
  \BibitemOpen
  \bibfield  {author} {\bibinfo {author} {\bibfnamefont {A.}~\bibnamefont
  {Kitaev}},\ }\href {\doibase 10.1016/j.aop.2005.10.005} {\bibfield  {journal}
  {\bibinfo  {journal} {Annals of Physics}\ }\textbf {\bibinfo {volume}
  {321}},\ \bibinfo {pages} {2 } (\bibinfo {year} {2006})}\BibitemShut
  {NoStop}%
\bibitem [{\citenamefont {Wang}\ and\ \citenamefont
  {Vishwanath}(2006)}]{fwang06}%
  \BibitemOpen
  \bibfield  {author} {\bibinfo {author} {\bibfnamefont {F.}~\bibnamefont
  {Wang}}\ and\ \bibinfo {author} {\bibfnamefont {A.}~\bibnamefont
  {Vishwanath}},\ }\href {\doibase 10.1103/PhysRevB.74.174423} {\bibfield
  {journal} {\bibinfo  {journal} {Phys. Rev. B}\ }\textbf {\bibinfo {volume}
  {74}},\ \bibinfo {pages} {174423} (\bibinfo {year} {2006})}\BibitemShut
  {NoStop}%
\bibitem [{\citenamefont {Kou}\ and\ \citenamefont {Wen}(2009)}]{kou09}%
  \BibitemOpen
  \bibfield  {author} {\bibinfo {author} {\bibfnamefont {S.-P.}\ \bibnamefont
  {Kou}}\ and\ \bibinfo {author} {\bibfnamefont {X.-G.}\ \bibnamefont {Wen}},\
  }\href {\doibase 10.1103/PhysRevB.80.224406} {\bibfield  {journal} {\bibinfo
  {journal} {Phys. Rev. B}\ }\textbf {\bibinfo {volume} {80}},\ \bibinfo
  {pages} {224406} (\bibinfo {year} {2009})}\BibitemShut {NoStop}%
\bibitem [{\citenamefont {Huh}\ \emph {et~al.}(2011)\citenamefont {Huh},
  \citenamefont {Punk},\ and\ \citenamefont {Sachdev}}]{huh11}%
  \BibitemOpen
  \bibfield  {author} {\bibinfo {author} {\bibfnamefont {Y.}~\bibnamefont
  {Huh}}, \bibinfo {author} {\bibfnamefont {M.}~\bibnamefont {Punk}}, \ and\
  \bibinfo {author} {\bibfnamefont {S.}~\bibnamefont {Sachdev}},\ }\href
  {\doibase 10.1103/PhysRevB.84.094419} {\bibfield  {journal} {\bibinfo
  {journal} {Phys. Rev. B}\ }\textbf {\bibinfo {volume} {84}},\ \bibinfo
  {pages} {094419} (\bibinfo {year} {2011})}\BibitemShut {NoStop}%
\bibitem [{\citenamefont {Cho}\ \emph {et~al.}(2012)\citenamefont {Cho},
  \citenamefont {Lu},\ and\ \citenamefont {Moore}}]{cho12}%
  \BibitemOpen
  \bibfield  {author} {\bibinfo {author} {\bibfnamefont {G.~Y.}\ \bibnamefont
  {Cho}}, \bibinfo {author} {\bibfnamefont {Y.-M.}\ \bibnamefont {Lu}}, \ and\
  \bibinfo {author} {\bibfnamefont {J.~E.}\ \bibnamefont {Moore}},\ }\href
  {\doibase 10.1103/PhysRevB.86.125101} {\bibfield  {journal} {\bibinfo
  {journal} {Phys. Rev. B}\ }\textbf {\bibinfo {volume} {86}},\ \bibinfo
  {pages} {125101} (\bibinfo {year} {2012})}\BibitemShut {NoStop}%
\bibitem [{\citenamefont {Chen}\ \emph {et~al.}(2012)\citenamefont {Chen},
  \citenamefont {Essin},\ and\ \citenamefont {Hermele}}]{gchen12}%
  \BibitemOpen
  \bibfield  {author} {\bibinfo {author} {\bibfnamefont {G.}~\bibnamefont
  {Chen}}, \bibinfo {author} {\bibfnamefont {A.}~\bibnamefont {Essin}}, \ and\
  \bibinfo {author} {\bibfnamefont {M.}~\bibnamefont {Hermele}},\ }\href
  {\doibase 10.1103/PhysRevB.85.094418} {\bibfield  {journal} {\bibinfo
  {journal} {Phys. Rev. B}\ }\textbf {\bibinfo {volume} {85}},\ \bibinfo
  {pages} {094418} (\bibinfo {year} {2012})}\BibitemShut {NoStop}%
\bibitem [{\citenamefont {Levin}\ and\ \citenamefont {Stern}(2012)}]{levin12}%
  \BibitemOpen
  \bibfield  {author} {\bibinfo {author} {\bibfnamefont {M.}~\bibnamefont
  {Levin}}\ and\ \bibinfo {author} {\bibfnamefont {A.}~\bibnamefont {Stern}},\
  }\href {\doibase 10.1103/PhysRevB.86.115131} {\bibfield  {journal} {\bibinfo
  {journal} {Phys. Rev. B}\ }\textbf {\bibinfo {volume} {86}},\ \bibinfo
  {pages} {115131} (\bibinfo {year} {2012})}\BibitemShut {NoStop}%
\bibitem [{\citenamefont {Essin}\ and\ \citenamefont
  {Hermele}(2013)}]{essin13}%
  \BibitemOpen
  \bibfield  {author} {\bibinfo {author} {\bibfnamefont {A.~M.}\ \bibnamefont
  {Essin}}\ and\ \bibinfo {author} {\bibfnamefont {M.}~\bibnamefont
  {Hermele}},\ }\href {\doibase 10.1103/PhysRevB.87.104406} {\bibfield
  {journal} {\bibinfo  {journal} {Phys. Rev. B}\ }\textbf {\bibinfo {volume}
  {87}},\ \bibinfo {pages} {104406} (\bibinfo {year} {2013})},\ \Eprint
  {http://arxiv.org/abs/1212.0593} {arXiv:1212.0593} \BibitemShut {NoStop}%
\bibitem [{\citenamefont {Mesaros}\ and\ \citenamefont
  {Ran}(2013)}]{Mesaros2013}%
  \BibitemOpen
  \bibfield  {author} {\bibinfo {author} {\bibfnamefont {A.}~\bibnamefont
  {Mesaros}}\ and\ \bibinfo {author} {\bibfnamefont {Y.}~\bibnamefont {Ran}},\
  }\href {\doibase 10.1103/PhysRevB.87.155115} {\bibfield  {journal} {\bibinfo
  {journal} {Phys. Rev. B}\ }\textbf {\bibinfo {volume} {87}},\ \bibinfo
  {pages} {155115} (\bibinfo {year} {2013})}\BibitemShut {NoStop}%
\bibitem [{\citenamefont {Hung}\ and\ \citenamefont {Wen}(2013)}]{Hung2013}%
  \BibitemOpen
  \bibfield  {author} {\bibinfo {author} {\bibfnamefont {L.-Y.}\ \bibnamefont
  {Hung}}\ and\ \bibinfo {author} {\bibfnamefont {X.-G.}\ \bibnamefont {Wen}},\
  }\href {\doibase 10.1103/PhysRevB.87.165107} {\bibfield  {journal} {\bibinfo
  {journal} {Phys. Rev. B}\ }\textbf {\bibinfo {volume} {87}},\ \bibinfo
  {pages} {165107} (\bibinfo {year} {2013})}\BibitemShut {NoStop}%
\bibitem [{\citenamefont {Hung}\ and\ \citenamefont {Wan}(2013)}]{Hung2013b}%
  \BibitemOpen
  \bibfield  {author} {\bibinfo {author} {\bibfnamefont {L.-Y.}\ \bibnamefont
  {Hung}}\ and\ \bibinfo {author} {\bibfnamefont {Y.}~\bibnamefont {Wan}},\
  }\href {\doibase 10.1103/PhysRevB.87.195103} {\bibfield  {journal} {\bibinfo
  {journal} {Phys. Rev. B}\ }\textbf {\bibinfo {volume} {87}},\ \bibinfo
  {pages} {195103} (\bibinfo {year} {2013})}\BibitemShut {NoStop}%
\bibitem [{\citenamefont {{Lu}}\ and\ \citenamefont
  {{Vishwanath}}(2013)}]{Lu2013}%
  \BibitemOpen
  \bibfield  {author} {\bibinfo {author} {\bibfnamefont {Y.-M.}\ \bibnamefont
  {{Lu}}}\ and\ \bibinfo {author} {\bibfnamefont {A.}~\bibnamefont
  {{Vishwanath}}},\ }\href@noop {} {\bibfield  {journal} {\bibinfo  {journal}
  {ArXiv e-prints}\ } (\bibinfo {year} {2013})},\ \Eprint
  {http://arxiv.org/abs/1302.2634} {arXiv:1302.2634 [cond-mat.str-el]}
  \BibitemShut {NoStop}%
\bibitem [{\citenamefont {Xu}(2013)}]{Xu2013}%
  \BibitemOpen
  \bibfield  {author} {\bibinfo {author} {\bibfnamefont {C.}~\bibnamefont
  {Xu}},\ }\href {\doibase 10.1103/PhysRevB.88.205137} {\bibfield  {journal}
  {\bibinfo  {journal} {Phys. Rev. B}\ }\textbf {\bibinfo {volume} {88}},\
  \bibinfo {pages} {205137} (\bibinfo {year} {2013})}\BibitemShut {NoStop}%
\bibitem [{\citenamefont {Wang}\ and\ \citenamefont {Senthil}(2013)}]{cwang13}%
  \BibitemOpen
  \bibfield  {author} {\bibinfo {author} {\bibfnamefont {C.}~\bibnamefont
  {Wang}}\ and\ \bibinfo {author} {\bibfnamefont {T.}~\bibnamefont {Senthil}},\
  }\href {\doibase 10.1103/PhysRevB.87.235122} {\bibfield  {journal} {\bibinfo
  {journal} {Phys. Rev. B}\ }\textbf {\bibinfo {volume} {87}},\ \bibinfo
  {pages} {235122} (\bibinfo {year} {2013})}\BibitemShut {NoStop}%
\bibitem [{\citenamefont {Essin}\ and\ \citenamefont
  {Hermele}(2014)}]{essin14}%
  \BibitemOpen
  \bibfield  {author} {\bibinfo {author} {\bibfnamefont {A.~M.}\ \bibnamefont
  {Essin}}\ and\ \bibinfo {author} {\bibfnamefont {M.}~\bibnamefont
  {Hermele}},\ }\href {\doibase 10.1103/PhysRevB.90.121102} {\bibfield
  {journal} {\bibinfo  {journal} {Phys. Rev. B}\ }\textbf {\bibinfo {volume}
  {90}},\ \bibinfo {pages} {121102} (\bibinfo {year} {2014})}\BibitemShut
  {NoStop}%
\bibitem [{\citenamefont {Chen}\ \emph {et~al.}(2014)\citenamefont {Chen},
  \citenamefont {Burnell}, \citenamefont {Vishwanath},\ and\ \citenamefont
  {Fidkowski}}]{chen14}%
  \BibitemOpen
  \bibfield  {author} {\bibinfo {author} {\bibfnamefont {X.}~\bibnamefont
  {Chen}}, \bibinfo {author} {\bibfnamefont {F.~J.}\ \bibnamefont {Burnell}},
  \bibinfo {author} {\bibfnamefont {A.}~\bibnamefont {Vishwanath}}, \ and\
  \bibinfo {author} {\bibfnamefont {L.}~\bibnamefont {Fidkowski}},\ }\href@noop
  {} {\bibfield  {journal} {\bibinfo  {journal} {ArXiv e-prints}\ } (\bibinfo
  {year} {2014})},\ \Eprint {http://arxiv.org/abs/1403.6491} {arXiv:1403.6491
  [cond-mat.str-el]} \BibitemShut {NoStop}%
\bibitem [{\citenamefont {Gu}\ \emph {et~al.}(2014)\citenamefont {Gu},
  \citenamefont {Hung},\ and\ \citenamefont {Wan}}]{ygu14}%
  \BibitemOpen
  \bibfield  {author} {\bibinfo {author} {\bibfnamefont {Y.}~\bibnamefont
  {Gu}}, \bibinfo {author} {\bibfnamefont {L.-Y.}\ \bibnamefont {Hung}}, \ and\
  \bibinfo {author} {\bibfnamefont {Y.}~\bibnamefont {Wan}},\ }\href@noop {}
  {\bibfield  {journal} {\bibinfo  {journal} {ArXiv e-prints}\ } (\bibinfo
  {year} {2014})},\ \Eprint {http://arxiv.org/abs/1402.3356} {arXiv:1402.3356
  [cond-mat.str-el]} \BibitemShut {NoStop}%
\bibitem [{\citenamefont {Lu}\ \emph {et~al.}(2014)\citenamefont {Lu},
  \citenamefont {Cho},\ and\ \citenamefont {Vishwanath}}]{ymlu14}%
  \BibitemOpen
  \bibfield  {author} {\bibinfo {author} {\bibfnamefont {Y.-M.}\ \bibnamefont
  {Lu}}, \bibinfo {author} {\bibfnamefont {G.~Y.}\ \bibnamefont {Cho}}, \ and\
  \bibinfo {author} {\bibfnamefont {A.}~\bibnamefont {Vishwanath}},\
  }\href@noop {} {\bibfield  {journal} {\bibinfo  {journal} {ArXiv e-prints}\ }
  (\bibinfo {year} {2014})},\ \Eprint {http://arxiv.org/abs/1403.0575}
  {arXiv:1403.0575 [cond-mat.str-el]} \BibitemShut {NoStop}%
\bibitem [{\citenamefont {Huang}\ \emph {et~al.}(2014)\citenamefont {Huang},
  \citenamefont {Chen},\ and\ \citenamefont {Pollmann}}]{Huang2014}%
  \BibitemOpen
  \bibfield  {author} {\bibinfo {author} {\bibfnamefont {C.-Y.}\ \bibnamefont
  {Huang}}, \bibinfo {author} {\bibfnamefont {X.}~\bibnamefont {Chen}}, \ and\
  \bibinfo {author} {\bibfnamefont {F.}~\bibnamefont {Pollmann}},\ }\href
  {\doibase 10.1103/PhysRevB.90.045142} {\bibfield  {journal} {\bibinfo
  {journal} {Phys. Rev. B}\ }\textbf {\bibinfo {volume} {90}},\ \bibinfo
  {pages} {045142} (\bibinfo {year} {2014})}\BibitemShut {NoStop}%
\bibitem [{\citenamefont {Hermele}(2014)}]{hermele14}%
  \BibitemOpen
  \bibfield  {author} {\bibinfo {author} {\bibfnamefont {M.}~\bibnamefont
  {Hermele}},\ }\href {\doibase 10.1103/PhysRevB.90.184418} {\bibfield
  {journal} {\bibinfo  {journal} {Phys. Rev. B}\ }\textbf {\bibinfo {volume}
  {90}},\ \bibinfo {pages} {184418} (\bibinfo {year} {2014})}\BibitemShut
  {NoStop}%
\bibitem [{\citenamefont {Reuther}\ \emph {et~al.}(2014)\citenamefont
  {Reuther}, \citenamefont {Lee},\ and\ \citenamefont {Alicea}}]{reuther14}%
  \BibitemOpen
  \bibfield  {author} {\bibinfo {author} {\bibfnamefont {J.}~\bibnamefont
  {Reuther}}, \bibinfo {author} {\bibfnamefont {S.-P.}\ \bibnamefont {Lee}}, \
  and\ \bibinfo {author} {\bibfnamefont {J.}~\bibnamefont {Alicea}},\
  }\href@noop {} {\bibfield  {journal} {\bibinfo  {journal} {ArXiv e-prints}\ }
  (\bibinfo {year} {2014})},\ \Eprint {http://arxiv.org/abs/1407.4124}
  {arXiv:1407.4124 [cond-mat.str-el]} \BibitemShut {NoStop}%
\bibitem [{\citenamefont {Neupert}\ \emph {et~al.}(2014)\citenamefont
  {Neupert}, \citenamefont {Chamon}, \citenamefont {Mudry},\ and\ \citenamefont
  {Thomale}}]{neupert2014wire}%
  \BibitemOpen
  \bibfield  {author} {\bibinfo {author} {\bibfnamefont {T.}~\bibnamefont
  {Neupert}}, \bibinfo {author} {\bibfnamefont {C.}~\bibnamefont {Chamon}},
  \bibinfo {author} {\bibfnamefont {C.}~\bibnamefont {Mudry}}, \ and\ \bibinfo
  {author} {\bibfnamefont {R.}~\bibnamefont {Thomale}},\ }\href@noop {}
  {\bibfield  {journal} {\bibinfo  {journal} {ArXiv e-prints}\ } (\bibinfo
  {year} {2014})},\ \Eprint {http://arxiv.org/abs/1403.0953} {arXiv:1403.0953
  [cond-mat.str-el]} \BibitemShut {NoStop}%
\bibitem [{\citenamefont {Turner}\ \emph {et~al.}(2012)\citenamefont {Turner},
  \citenamefont {Zhang}, \citenamefont {Mong},\ and\ \citenamefont
  {Vishwanath}}]{turner10}%
  \BibitemOpen
  \bibfield  {author} {\bibinfo {author} {\bibfnamefont {A.~M.}\ \bibnamefont
  {Turner}}, \bibinfo {author} {\bibfnamefont {Y.}~\bibnamefont {Zhang}},
  \bibinfo {author} {\bibfnamefont {R.~S.~K.}\ \bibnamefont {Mong}}, \ and\
  \bibinfo {author} {\bibfnamefont {A.}~\bibnamefont {Vishwanath}},\ }\href
  {\doibase 10.1103/PhysRevB.85.165120} {\bibfield  {journal} {\bibinfo
  {journal} {Phys. Rev. B}\ }\textbf {\bibinfo {volume} {85}},\ \bibinfo
  {pages} {165120} (\bibinfo {year} {2012})}\BibitemShut {NoStop}%
\bibitem [{\citenamefont {Turner}\ \emph {et~al.}(2010)\citenamefont {Turner},
  \citenamefont {Zhang},\ and\ \citenamefont {Vishwanath}}]{turner10b}%
  \BibitemOpen
  \bibfield  {author} {\bibinfo {author} {\bibfnamefont {A.~M.}\ \bibnamefont
  {Turner}}, \bibinfo {author} {\bibfnamefont {Y.}~\bibnamefont {Zhang}}, \
  and\ \bibinfo {author} {\bibfnamefont {A.}~\bibnamefont {Vishwanath}},\
  }\href {\doibase 10.1103/PhysRevB.82.241102} {\bibfield  {journal} {\bibinfo
  {journal} {Phys. Rev. B}\ }\textbf {\bibinfo {volume} {82}},\ \bibinfo
  {pages} {241102} (\bibinfo {year} {2010})}\BibitemShut {NoStop}%
\bibitem [{\citenamefont {Fu}(2011)}]{fu11}%
  \BibitemOpen
  \bibfield  {author} {\bibinfo {author} {\bibfnamefont {L.}~\bibnamefont
  {Fu}},\ }\href {\doibase 10.1103/PhysRevLett.106.106802} {\bibfield
  {journal} {\bibinfo  {journal} {Phys. Rev. Lett.}\ }\textbf {\bibinfo
  {volume} {106}},\ \bibinfo {pages} {106802} (\bibinfo {year}
  {2011})}\BibitemShut {NoStop}%
\bibitem [{\citenamefont {Hughes}\ \emph {et~al.}(2011)\citenamefont {Hughes},
  \citenamefont {Prodan},\ and\ \citenamefont {Bernevig}}]{hughes11}%
  \BibitemOpen
  \bibfield  {author} {\bibinfo {author} {\bibfnamefont {T.~L.}\ \bibnamefont
  {Hughes}}, \bibinfo {author} {\bibfnamefont {E.}~\bibnamefont {Prodan}}, \
  and\ \bibinfo {author} {\bibfnamefont {B.~A.}\ \bibnamefont {Bernevig}},\
  }\href {\doibase 10.1103/PhysRevB.83.245132} {\bibfield  {journal} {\bibinfo
  {journal} {Phys. Rev. B}\ }\textbf {\bibinfo {volume} {83}},\ \bibinfo
  {pages} {245132} (\bibinfo {year} {2011})}\BibitemShut {NoStop}%
\bibitem [{\citenamefont {Teo}\ and\ \citenamefont {Hughes}(2013)}]{teo13}%
  \BibitemOpen
  \bibfield  {author} {\bibinfo {author} {\bibfnamefont {J.~C.~Y.}\
  \bibnamefont {Teo}}\ and\ \bibinfo {author} {\bibfnamefont {T.~L.}\
  \bibnamefont {Hughes}},\ }\href {\doibase 10.1103/PhysRevLett.111.047006}
  {\bibfield  {journal} {\bibinfo  {journal} {Phys. Rev. Lett.}\ }\textbf
  {\bibinfo {volume} {111}},\ \bibinfo {pages} {047006} (\bibinfo {year}
  {2013})}\BibitemShut {NoStop}%
\bibitem [{\citenamefont {Fang}\ \emph {et~al.}(2012)\citenamefont {Fang},
  \citenamefont {Gilbert},\ and\ \citenamefont {Bernevig}}]{fang12}%
  \BibitemOpen
  \bibfield  {author} {\bibinfo {author} {\bibfnamefont {C.}~\bibnamefont
  {Fang}}, \bibinfo {author} {\bibfnamefont {M.~J.}\ \bibnamefont {Gilbert}}, \
  and\ \bibinfo {author} {\bibfnamefont {B.~A.}\ \bibnamefont {Bernevig}},\
  }\href {\doibase 10.1103/PhysRevB.86.115112} {\bibfield  {journal} {\bibinfo
  {journal} {Phys. Rev. B}\ }\textbf {\bibinfo {volume} {86}},\ \bibinfo
  {pages} {115112} (\bibinfo {year} {2012})}\BibitemShut {NoStop}%
\bibitem [{\citenamefont {Wang}\ \emph {et~al.}(2012)\citenamefont {Wang},
  \citenamefont {Qi},\ and\ \citenamefont {Zhang}}]{zwang12b}%
  \BibitemOpen
  \bibfield  {author} {\bibinfo {author} {\bibfnamefont {Z.}~\bibnamefont
  {Wang}}, \bibinfo {author} {\bibfnamefont {X.-L.}\ \bibnamefont {Qi}}, \ and\
  \bibinfo {author} {\bibfnamefont {S.-C.}\ \bibnamefont {Zhang}},\ }\href
  {\doibase 10.1103/PhysRevB.85.165126} {\bibfield  {journal} {\bibinfo
  {journal} {Phys. Rev. B}\ }\textbf {\bibinfo {volume} {85}},\ \bibinfo
  {pages} {165126} (\bibinfo {year} {2012})}\BibitemShut {NoStop}%
\bibitem [{\citenamefont {Chiu}\ \emph {et~al.}(2013)\citenamefont {Chiu},
  \citenamefont {Yao},\ and\ \citenamefont {Ryu}}]{chiu13}%
  \BibitemOpen
  \bibfield  {author} {\bibinfo {author} {\bibfnamefont {C.-K.}\ \bibnamefont
  {Chiu}}, \bibinfo {author} {\bibfnamefont {H.}~\bibnamefont {Yao}}, \ and\
  \bibinfo {author} {\bibfnamefont {S.}~\bibnamefont {Ryu}},\ }\href {\doibase
  10.1103/PhysRevB.88.075142} {\bibfield  {journal} {\bibinfo  {journal} {Phys.
  Rev. B}\ }\textbf {\bibinfo {volume} {88}},\ \bibinfo {pages} {075142}
  (\bibinfo {year} {2013})}\BibitemShut {NoStop}%
\bibitem [{\citenamefont {Zhang}\ \emph {et~al.}(2013)\citenamefont {Zhang},
  \citenamefont {Kane},\ and\ \citenamefont {Mele}}]{fzhang13}%
  \BibitemOpen
  \bibfield  {author} {\bibinfo {author} {\bibfnamefont {F.}~\bibnamefont
  {Zhang}}, \bibinfo {author} {\bibfnamefont {C.~L.}\ \bibnamefont {Kane}}, \
  and\ \bibinfo {author} {\bibfnamefont {E.~J.}\ \bibnamefont {Mele}},\ }\href
  {\doibase 10.1103/PhysRevLett.111.056403} {\bibfield  {journal} {\bibinfo
  {journal} {Phys. Rev. Lett.}\ }\textbf {\bibinfo {volume} {111}},\ \bibinfo
  {pages} {056403} (\bibinfo {year} {2013})}\BibitemShut {NoStop}%
\bibitem [{\citenamefont {Slager}\ \emph {et~al.}(2013)\citenamefont {Slager},
  \citenamefont {Mesaros}, \citenamefont {Juri{\v{c}}i{\'c}},\ and\
  \citenamefont {Zaanen}}]{slager2013space}%
  \BibitemOpen
  \bibfield  {author} {\bibinfo {author} {\bibfnamefont {R.-J.}\ \bibnamefont
  {Slager}}, \bibinfo {author} {\bibfnamefont {A.}~\bibnamefont {Mesaros}},
  \bibinfo {author} {\bibfnamefont {V.}~\bibnamefont {Juri{\v{c}}i{\'c}}}, \
  and\ \bibinfo {author} {\bibfnamefont {J.}~\bibnamefont {Zaanen}},\ }\href
  {http://dx.doi.org/10.1038/nphys2513} {\bibfield  {journal} {\bibinfo
  {journal} {Nature Physics}\ }\textbf {\bibinfo {volume} {9}},\ \bibinfo
  {pages} {98} (\bibinfo {year} {2013})}\BibitemShut {NoStop}%
\bibitem [{\citenamefont {Fidkowski}\ \emph {et~al.}()\citenamefont
  {Fidkowski}, \citenamefont {Lindner},\ and\ \citenamefont
  {Kitaev}}]{lukaszpc}%
  \BibitemOpen
  \bibfield  {author} {\bibinfo {author} {\bibfnamefont {L.}~\bibnamefont
  {Fidkowski}}, \bibinfo {author} {\bibfnamefont {N.~H.}\ \bibnamefont
  {Lindner}}, \ and\ \bibinfo {author} {\bibfnamefont {A.}~\bibnamefont
  {Kitaev}},\ }\href@noop {} {}\bibinfo {note} {{t}o appear}\BibitemShut
  {NoStop}%
\bibitem [{\citenamefont {Barkeshli}\ \emph {et~al.}(2014)\citenamefont
  {Barkeshli}, \citenamefont {Bonderson}, \citenamefont {Cheng},\ and\
  \citenamefont {Wang}}]{barkeshli14}%
  \BibitemOpen
  \bibfield  {author} {\bibinfo {author} {\bibfnamefont {M.}~\bibnamefont
  {Barkeshli}}, \bibinfo {author} {\bibfnamefont {P.}~\bibnamefont
  {Bonderson}}, \bibinfo {author} {\bibfnamefont {M.}~\bibnamefont {Cheng}}, \
  and\ \bibinfo {author} {\bibfnamefont {Z.}~\bibnamefont {Wang}},\ }\href@noop
  {} {\bibfield  {journal} {\bibinfo  {journal} {ArXiv e-prints}\ } (\bibinfo
  {year} {2014})},\ \Eprint {http://arxiv.org/abs/1410.4540} {arXiv:1410.4540}
  \BibitemShut {NoStop}%
\bibitem [{\citenamefont {Vishwanath}\ and\ \citenamefont
  {Senthil}(2013)}]{vishwanath12}%
  \BibitemOpen
  \bibfield  {author} {\bibinfo {author} {\bibfnamefont {A.}~\bibnamefont
  {Vishwanath}}\ and\ \bibinfo {author} {\bibfnamefont {T.}~\bibnamefont
  {Senthil}},\ }\href {\doibase 10.1103/PhysRevX.3.011016} {\bibfield
  {journal} {\bibinfo  {journal} {Phys. Rev. X}\ }\textbf {\bibinfo {volume}
  {3}},\ \bibinfo {pages} {011016} (\bibinfo {year} {2013})}\BibitemShut
  {NoStop}%
\bibitem [{\citenamefont {Metlitski}\ \emph {et~al.}(2013)\citenamefont
  {Metlitski}, \citenamefont {Kane},\ and\ \citenamefont
  {Fisher}}]{metlitski13}%
  \BibitemOpen
  \bibfield  {author} {\bibinfo {author} {\bibfnamefont {M.~A.}\ \bibnamefont
  {Metlitski}}, \bibinfo {author} {\bibfnamefont {C.~L.}\ \bibnamefont {Kane}},
  \ and\ \bibinfo {author} {\bibfnamefont {M.~P.~A.}\ \bibnamefont {Fisher}},\
  }\href {\doibase 10.1103/PhysRevB.88.035131} {\bibfield  {journal} {\bibinfo
  {journal} {Phys. Rev. B}\ }\textbf {\bibinfo {volume} {88}},\ \bibinfo
  {pages} {035131} (\bibinfo {year} {2013})}\BibitemShut {NoStop}%
\bibitem [{\citenamefont {Kitaev}(2003)}]{kitaev03}%
  \BibitemOpen
  \bibfield  {author} {\bibinfo {author} {\bibfnamefont {A.~{\relax Yu}.}\
  \bibnamefont {Kitaev}},\ }\href {\doibase 10.1016/S0003-4916(02)00018-0}
  {\bibfield  {journal} {\bibinfo  {journal} {Ann. Phys.}\ }\textbf {\bibinfo
  {volume} {303}},\ \bibinfo {pages} {2} (\bibinfo {year} {2003})},\ \Eprint
  {http://arxiv.org/abs/quant-ph/9707021} {arXiv:quant-ph/9707021} \BibitemShut
  {NoStop}%
\bibitem [{Note1()}]{Note1}%
  \BibitemOpen
  \bibinfo {note} {$\protect \mathbb {Z}_2$ lattice gauge theory with fermionic
  matter also gives rise to $\protect \mathbb {Z}_2$ topological
  order.}\BibitemShut {Stop}%
\bibitem [{Note2()}]{Note2}%
  \BibitemOpen
  \bibinfo {note} {This is the case provided $K^e_v$ and $K^m_p$ are compatible
  with constraints obeyed by $A_v$ and $B_p$ operators. More precisely, we have
  $\DOTSB \prod@ \slimits@ _v A_v = 1$, which implies $K^e_v$ must satisfy
  $\DOTSB \prod@ \slimits@ _v K^e_v = 1$. In addition, suppose $P'$ is a subset
  of $P$ for which $\DOTSB \prod@ \slimits@ _{p \in P'} B_p = 1$, then we must
  have $\DOTSB \prod@ \slimits@ _{p \in P'} K^m_p = 1$.}\BibitemShut {Stop}%
\bibitem [{Note3()}]{Note3}%
  \BibitemOpen
  \bibinfo {note} {The origin of these requirements is the fact that these
  properties holds for hold for all electrically neutral bosonic degrees of
  freedom (\protect \emph {e.g.} electron spins, bosonic atoms) that can be
  microscopic constituents of a condensed matter system.}\BibitemShut {Stop}%
\end{thebibliography}%

\end{document}